\documentclass[a4paper]{article}
\usepackage[a4paper]{geometry}
\usepackage[utf8]{inputenc}

\usepackage{amsmath,amsfonts,amssymb,amsthm}
\usepackage{thmtools}
\usepackage{stmaryrd}
\usepackage{natbib}
\usepackage{url}
\usepackage{array}
\usepackage{arydshln}
\usepackage{ifthen}
\usepackage{ifpdf}
\usepackage{verbatim}
\usepackage{mathpartir}
\usepackage{hyperref}
\usepackage{listings}
\usepackage{rotating}

\usepackage{mathtools}

\DeclarePairedDelimiter\floor{\lfloor}{\rfloor}

\usepackage{appendix}

\usepackage{enumitem}
\newlist{enumproof}{enumerate}{10}
\setlist[enumproof]{label*=\arabic*.}

\usepackage{tikz}
\usepackage{multirow,bigdelim}

\declaretheorem[numbered=yes,name=Lemma,qed=$\blacksquare$]{lemma}
\declaretheorem[numbered=yes,name=Theorem,qed=$\blacksquare$]{theorem}
\declaretheorem[numbered=yes,name=Definition,qed=$\blacksquare$]{definition}
\declaretheorem[numbered=yes,name=Specification,qed=$\blacksquare$]{specification}

\newcommand{\forcenewline}{$\phantom{v}$\\}

\newcommand{\update}[2]{[#1 \mapsto #2]}
\newcommand{\sem}[1]{\left\llbracket #1 \right\rrbracket}

\newcommand{\parfun}{\rightharpoonup}
\newcommand{\finparfun}{\xrightharpoonup{\textit{\tiny{fin}}}}
\newcommand{\monnefun}{\xrightarrow{\textit{\tiny{mon, ne}}}}

\newcommand{\nefun}{\xrightarrow{\textit{\tiny{ne}}}}
\newcommand{\fun}{\rightarrow}
\newcommand{\defeq}{\stackrel{\textit{\tiny{def}}}{=}}
\newcommand{\nequal}[1][n]{\stackrel{\tiny{#1}}{=}}
\renewcommand{\nsim}[1][n]{\stackrel{\tiny{#1}}{\simeq}}

\newcommand\subsetsim{\mathrel{\ooalign{\raise.2ex\hbox{$\subset$}\cr
      \hidewidth\lower.8ex\hbox{\scalebox{0.9}{$\sim$}}\hidewidth\cr}}}
\newcommand\supsetsim{\mathrel{\ooalign{\raise.2ex\hbox{$\supset$}\cr
      \hidewidth\lower.8ex\hbox{\scalebox{0.9}{$\sim$}}\hidewidth\cr}}}
\newcommand{\nsubsim}[1][n]{\stackrel{\tiny{#1}}{\subsetsim}}
\newcommand{\nsupsim}[1][n]{\stackrel{\tiny{#1}}{\supsetsim}}

\newcommand{\nsubeq}[1][n]{\stackrel{\tiny{#1}}{\subseteq}}
\newcommand{\nsupeq}[1][n]{\stackrel{\tiny{#1}}{\supseteq}}

\newcommand{\union}{\mathbin{\cup}}
\DeclareMathOperator{\dom}{dom}
\newcommand{\blater}{\mathop{\blacktriangleright}}

\newcommand{\powerset}[1]{\mathcal{P}(#1)}

\newcommand{\false}{\mathit{false}}
\newcommand{\true}{\mathit{true}}

\newcommand{\cofe}{c.o.f.e.}
\newcommand{\cofes}{\cofe{}'s}

\newcommand\lau[1]{{\color{purple} \sf \footnotesize {LS: #1}}\\}
\newcommand\dominique[1]{{\color{purple} \sf \footnotesize {DD: #1}}\\}
\newcommand\lars[1]{{\color{purple} \sf \footnotesize {LB: #1}}\\}

\renewcommand\lau[1]{}
\renewcommand\dominique[1]{}
\renewcommand\lars[1]{}

\newcommand{\var}[1]{\mathit{#1}}
\newcommand{\hs}{\var{ms}}
\newcommand{\ms}{\hs}
\newcommand{\hv}{\var{hv}}
\newcommand{\rv}{\var{rv}}
\newcommand{\lv}{\var{lv}}
\newcommand{\gl}{\var{g}}
\newcommand{\pc}{\mathit{pc}}
\newcommand{\pcreg}{\mathrm{pc}}
\newcommand{\addr}{\var{a}}

\newcommand{\word}{\var{w}}
\newcommand{\start}{\var{base}}
\newcommand{\addrend}{\var{end}}
\newcommand{\inftyend}{-42}

\newcommand{\mem}{\var{mem}}
\newcommand{\reg}{\var{reg}}

\newcommand{\heap}{\var{mem}}

\newcommand{\perm}{\var{perm}}
\newcommand{\permp}{\var{permPair}}

\newcommand{\stdcap}[1][(\perm,\gl)]{\left(#1,\start,\addrend,\addr \right)}
\newcommand{\adv}{\var{adv}}

\newcommand{\link}{\var{link}}
\newcommand{\stk}{\var{stk}}
\newcommand{\flag}{\var{flag}}
\newcommand{\nwl}{\var{nwl}}
\newcommand{\pwl}{\var{pwl}}
\newcommand{\sta}{\var{sta}}
\newcommand{\cnst}{\var{cnst}}
\newcommand{\olf}{\var{offsetLinkFlag}}
\newcommand{\prp}{\var{prp}}
\newcommand{\env}{\var{env}}
\newcommand{\cls}{\var{cls}}
\newcommand{\unused}{\var{unused}}
\newcommand{\act}{\var{act}}

\newcommand{\plainproj}[1]{\mathrm{#1}}
\newcommand{\memheap}[1][\Phi]{#1.\plainproj{mem}}
\newcommand{\memreg}[1][\Phi]{#1.\plainproj{reg}}

\newcommand{\updateHeap}[3][\Phi]{#1\update{\plainproj{mem}.#2}{#3}}
\newcommand{\updateReg}[3][\Phi]{#1\update{\plainproj{reg}.#2}{#3}}

\newcommand{\failed}{\textsl{failed}}
\newcommand{\halted}{\textsl{halted}}

\newcommand{\plainfun}[2]{
  \ifthenelse{\equal{#2}{}}
  {\mathit{#1}}
  {\mathit{#1}(#2)}
}
\newcommand{\decode}{\plainfun{decode}{}}

\newcommand{\encodePerm}{\mathit{encodePerm}}
\newcommand{\encodePermPair}{\plainfun{encodePermPair}{}}
\newcommand{\encodeLoc}{\mathit{encodeLoc}{}}
\newcommand{\decodePermPair}{\plainfun{decodePermPair}}
\newcommand{\decodePerm}[1]{\plainfun{decodePerm}{#1}}
\newcommand{\updatePcPerm}[1]{\plainfun{updatePcPerm}{#1}}

\newcommand{\nonZero}[1]{\plainfun{nonZero}{#1}}
\newcommand{\readAllowed}[1]{\plainfun{readAllowed}{#1}}
\newcommand{\writeAllowed}[1]{\plainfun{writeAllowed}{#1}}
\newcommand{\withinBounds}[1]{\plainfun{withinBounds}{#1}}
\newcommand{\stdUpdatePc}[1]{\plainfun{updatePc}{#1}}

\newcommand{\readCond}[1]{\plainfun{readCondition}{#1}}
\newcommand{\writeCond}[1]{\plainfun{writeCondition}{#1}}
\newcommand{\execCond}[1]{\plainfun{executeCondition}{#1}}
\newcommand{\entryCond}[1]{\plainfun{enterCondition}{#1}}

\newcommand{\revokeTemp}[1]{\plainfun{revokeTemp}{#1}}
\newcommand{\erase}[2]{\floor*{#1}_{\{#2\}}}
\newcommand{\activeReg}[1]{\plainfun{active}{#1}}

\newcommand{\future}{\mathbin{\sqsupseteq}}
\newcommand{\pub}{\var{pub}}
\newcommand{\priv}{\var{priv}}
\newcommand{\futurewk}{\mathbin{\sqsupseteq}^{\var{pub}}}
\newcommand{\futurestr}{\mathbin{\sqsupseteq}^{\var{priv}}}
\newcommand{\heapSat}[3][\heap]{#1 :_{#2} #3}
\newcommand{\memSat}[3][n]{\heapSat[#2]{#1}{#3}}
\newcommand{\memSatPar}[4][n]{\heapSat[#2]{#1 , #4}{#3}}

\newcommand{\monwknefun}{\xrightarrow[\text{\tiny{$\futurewk$}}]{\textit{\tiny{mon, ne}}}}
\newcommand{\monstrnefun}{\xrightarrow[\text{\tiny{$\futurestr$}}]{\textit{\tiny{mon, ne}}}}

\newcommand{\codelabel}[1]{\mathit{#1}}

\newcommand{\malloc}{\codelabel{malloc}}

\newcommand{\asmType}{\plaindom{AsmType}}

\newcommand{\plaindom}[1]{\mathrm{#1}}
\newcommand{\Caps}{\plaindom{Cap}}
\newcommand{\Words}{\plaindom{Word}}
\newcommand{\Addrs}{\plaindom{Addr}}
\newcommand{\ExecConfs}{\plaindom{ExecConf}}
\newcommand{\RegName}{\plaindom{RegisterName}}
\newcommand{\Regs}{\plaindom{Reg}}
\newcommand{\Heaps}{\plaindom{Mem}}
\newcommand{\Mems}{\Heaps}
\newcommand{\HeapSegments}{\plaindom{MemSegment}}
\newcommand{\MemSegments}{\HeapSegments}
\newcommand{\Confs}{\plaindom{Conf}}
\newcommand{\Instrs}{\plaindom{Instructions}}
\newcommand{\nats}{\mathbb{N}}
\newcommand{\ints}{\mathbb{Z}}
\newcommand{\Perms}{\plaindom{Perm}}
\newcommand{\Globals}{\plaindom{Global}}

\newcommand{\Rels}{\plaindom{Rels}}
\newcommand{\States}{\plaindom{State}}
\newcommand{\RegionNames}{\plaindom{RegionName}}
\newcommand{\Regions}{\plaindom{Region}}

\newcommand{\Worlds}{\plaindom{World}}
\newcommand{\Wor}{\plaindom{Wor}}
\newcommand{\Worwk}{\Wor_{\futurewk}}
\newcommand{\Worstr}{\Wor_{\futurestr}}

\newcommand{\UPred}[1]{\plaindom{UPred}(#1)}

\newcommand{\intr}[2]{\mathcal{#1}}
\newcommand{\valueintr}[1]{\intr{V}{#1}}
\newcommand{\exprintr}[1]{\intr{E}{#1}}

\newcommand{\regintr}[1]{\intr{R}{#1}}
\newcommand{\stdvr}{\valueintr{\asmType}}
\newcommand{\stder}{\exprintr{\asmType}}
\newcommand{\stdrr}{\regintr{\asmType}}

\newcommand{\observations}{\mathcal{O}}
\newcommand{\npair}[2][n]{\left(#1,#2 \right)}

\newcommand{\refreg}[1]{\lfloor #1 \rfloor}
\newcommand{\refheap}[1]{\langle #1 \rangle_m}

\newcommand{\zinstr}[1]{\mathtt{#1}}
\newcommand{\fail}{\zinstr{fail}}
\newcommand{\halt}{\zinstr{halt}}
\newcommand{\oneinstr}[2]{\zinstr{#1} \; #2}
\newcommand{\jmp}[1]{\oneinstr{jmp}{#1}}
\newcommand{\twoinstr}[3]{\zinstr{#1} \; #2 \; #3}
\newcommand{\restricttwo}[2]{\twoinstr{restrict}{#1}{#2}}
\newcommand{\jnz}[2]{\twoinstr{jnz}{#1}{#2}}
\newcommand{\isptr}[2]{\twoinstr{isptr}{#1}{#2}}
\newcommand{\geta}[2]{\twoinstr{geta}{#1}{#2}}
\newcommand{\getb}[2]{\twoinstr{getb}{#1}{#2}}
\newcommand{\gete}[2]{\twoinstr{gete}{#1}{#2}}
\newcommand{\getp}[2]{\twoinstr{getp}{#1}{#2}}
\newcommand{\getl}[2]{\twoinstr{getl}{#1}{#2}}
\newcommand{\move}[2]{\twoinstr{move}{#1}{#2}}
\newcommand{\store}[2]{\twoinstr{store}{#1}{#2}}
\newcommand{\load}[2]{\twoinstr{load}{#1}{#2}}
\newcommand{\lea}[2]{\twoinstr{lea}{#1}{#2}}
\newcommand{\threeinstr}[4]{\zinstr{#1} \; #2 \; #3 \; #4}
\newcommand{\restrict}[3]{\threeinstr{restrict}{#1}{#2}{#3}}
\newcommand{\subseg}[3]{\threeinstr{subseg}{#1}{#2}{#3}}
\newcommand{\plus}[3]{\threeinstr{plus}{#1}{#2}{#3}}
\newcommand{\minus}[3]{\threeinstr{minus}{#1}{#2}{#3}}
\newcommand{\lt}[3]{\threeinstr{lt}{#1}{#2}{#3}}

\newcommand{\plainperm}[1]{\textsc{#1}}
\newcommand{\noperm}{\plainperm{o}}
\newcommand{\readonly}{\plainperm{ro}}
\newcommand{\readwrite}{\plainperm{rw}}
\newcommand{\exec}{\plainperm{rx}}
\newcommand{\entry}{\plainperm{e}}
\newcommand{\rwx}{\plainperm{rwx}}
\newcommand{\readwritel}{\plainperm{rwl}}
\newcommand{\rwl}{\readwritel}
\newcommand{\rwlx}{\plainperm{rwlx}}

\newcommand{\local}{\plainperm{local}}
\newcommand{\glob}{\plainperm{global}}

\newcommand{\localityReg}{\var{localityReg}}

\newcommand{\plainview}[1]{\mathrm{#1}}
\newcommand{\perma}{\plainview{perm}}
\newcommand{\temp}{\plainview{temp}}
\newcommand{\revoked}{\plainview{revoked}}

\newcommand{\step}[1][]{\rightarrow_{#1}}

\newcommand{\lookingat}[3]{\ensuremath{#1} \text{ is looking at } \ensuremath{#2} \text{ followed by } \ensuremath{#3}}
\newcommand{\pointstostack}[3]{\ensuremath{#1} \text{ points to stack with } \ensuremath{#2} \text{ used and } \ensuremath{#3} \text{ unused}}
\newcommand{\linksto}[4]{\ensuremath{#1} \text{ links } \ensuremath{#2} \text{ as } \ensuremath{#3} \text{ to } \ensuremath{#4}}
\newcommand{\nonlocal}[1]{\ensuremath{#1} \text{ is non-local}}

\newcommand{\scall}[3]{\mathtt{scall} \; #1([#2],[#3])}

\newcommand{\isdef}{\mathrel{\overset{\makebox[0pt]{\mbox{\normalfont\tiny\sffamily def}}}{=}}}
\newcommand\bnfdef{\mathrel{::=}}
\title{Reasoning About a Machine with Local Capabilities\\
 Provably Safe Stack and Return Pointer Management\\
Technical Appendix Including Proofs and Details}

\author{Lau Skorstengaard\\Aarhus University\\\texttt{lau@cs.au.dk}
        \and 
        Dominique Devriese\\Vrije~Universiteit~Brussel\\\texttt{dominique.devriese@vub.be}
        \and 
        Lars Birkedal\\Aarhus University\\\texttt{birkedal@cs.au.dk}}

\begin{document}
\maketitle
\tableofcontents

\section{Capability Machine Definition and Operational Semantics}
\subsection{Domains and Notation}

\begin{align*}
  \Addrs &\isdef \nats\\
  \Words &\isdef \Caps + \ints \\
  \Regs  &\isdef \RegName \rightarrow \Words\\
  \Heaps &\isdef \Addrs \rightarrow \Words \\
  \Perms &\bnfdef  \noperm\mid \readonly\mid \readwrite\mid \readwritel\mid \exec\mid \entry\mid \rwx\mid \rwlx\\
  \ExecConfs  &\isdef \Regs \times \Heaps \\
  \Globals & \bnfdef \glob\mid \local \\
  \Caps  &\isdef (\Perms \times \Globals) \times \Addrs \times (\Addrs + \{ \infty \}) \times \Addrs\\
  \Confs &\isdef \ExecConfs + \{\failed \} + \{\halted\} \times \Heaps \\
  \HeapSegments &\isdef \Addrs \parfun \Words
\end{align*}
Local capabilities have been added by adding a new domain $\Globals$ which represents whether a capability is local or global. There are two new permissions $\readwritel$ and $\rwlx$ that permits writing local capabilities. They are otherwise the same as their non-''permit write local'' counterparts.

As we have $\infty$ as a possible address, but our words cannot express $\infty$. We pick $\inftyend$ as a representative for $\infty$ when it is in memory (we could have picked any negative noumber). Note that $\inftyend$ is not an address, so for address operations $\inftyend$ only represents $\infty$. It is the responsible of the programmer to keep track of what represents addresses (and take necessary precautions).

Define the following predicate:
\begin{definition}
  \label{def:non-local-cap}
  We say word $w$ "\nonlocal{w}" iff either
  \begin{itemize}
  \item $w = \stdcap{(\perm,\glob)}$ for some $\perm$, $\addr$, $\start$, and $\addrend$; or
  \item $w \in \ints$
  \end{itemize}
\end{definition}

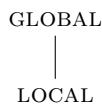
\begin{figure}[!h]
  \centering
  \begin{tikzpicture}[main node/.style={}]
    \node[main node] (1) {$\glob$};
    \node[main node] (2) [below of=1] {$\local$};

    \path[every node/.style={font=\sffamily\small}]
    (1) edge (2);
  \end{tikzpicture}
  \caption{Locality hierarchy}
  \label{fig:glob-hier}
\end{figure}

Things to note:
\begin{itemize}
\item $\RegName$ contains $\pcreg$, but is otherwise a sufficiently
  large finite set.
\item Table~\ref{tab:permission-list} describes what all the permissions grant access to.
\item Figure~\ref{fig:perm-hier} shows the ordering of the permissions, i.e, the elements of $\Perms$.
\item Figure~\ref{fig:glob-hier} shows the ordering of $\local$ and $\glob$, i.e., the elements of $\Globals$.
\item The ordering of $\Perms \times \Globals$ is pointwise.
\end{itemize}

\begin{table}[!h]
  \centering
  \begin{tabular}[!h]{r |  p{7cm} }
    $\noperm$ & No permissions. Grants no permissions\\
    \hline
    $\readonly$ & Read only. Grants read permission \\
    \hline
    $\readwrite$ & Read-write. Grants read and write permission. Storage of local capabilities prohibited. \\
    \hline
    $\readwritel$ & Read-write, permit write local. Grants read and write permission. Storage of local capabilities possible. \\
    \hline
    $\exec$ & Execute permission. Grants execute and read permissions.\\
    \hline
    $\entry$ & Enter permission. This permission grants no access, but when jumped to, it will turn into an $\exec$ permission.\\
    \hline
    $\rwx$ & Read-write-execute permission. Grants read, write, and execute permissions. Storage of local capabilities prohibited. \\
    \hline
    $\rwlx$ & Read-write-execute, permit write local. Grants read, write, and execute permissions. Storage of local capabilities possible.
  \end{tabular}

  \caption{The permissions in this capability system}
  \label{tab:permission-list}
\end{table}
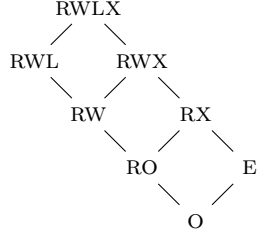
\begin{figure}[!h]
  \centering
  \begin{tikzpicture}[main node/.style={}]
    \node[main node] (7) {$\rwlx$};
    \node[main node] (8) [below left of=7] {$\readwritel$};
    \node[main node] (1) [below right of=7] {$\rwx$};
    \node[main node] (2) [below right of=1] {$\exec$};
    \node[main node] (3) [below right of=2] {$\entry$};
    \node[main node] (4) [below left of=1] {$\readwrite$};
    \node[main node] (5) [below right of=4] {$\readonly$};
    \node[main node] (6) [below right of=5] {$\noperm$};

    \path[every node/.style={font=\sffamily\small}]
    (7) edge (8)
    (7) edge (1)
    (8) edge (4)
    (1) edge (2)
    (2) edge (3)
    (2) edge (5)
    (3) edge (6)
    (1) edge (4)
    (4) edge (5)
    (5) edge (6);
  \end{tikzpicture}

  \caption{Permission hierarchy}
  \label{fig:perm-hier}
\end{figure}
Notation:
\[
  \begin{array}{rcl}
    i       &\in& \Instrs \\
    r       &\in& \RegName\\
    \pc     &\in& \Caps \\
    \pcreg  &\in& \RegName \\
    \Phi    &\in& \ExecConfs \\
    m, \memheap&\in& \Heaps \\
    \memreg &\in& \Regs \\
    \addr   &\in& \Addrs\\
    \perm   &\in& \Perms\\
    ((\perm,\gl),\start,\addrend,\addr) &\in& \Caps \\
    n       &\in& \ints\\
    \ms     &\in& \MemSegments
  \end{array}
\]
Words and instructions:
\[
  \begin{array}{rcl}
    \lv    &::=& \refreg{r} \\
    \hv    &::=& \refheap{r}\\
    \rv    &::=& n \mid \lv \\
    i      &::=& 
                 \jmp{\lv} \mid 
                 \jnz{\lv}{\lv} \mid
                 \move{\lv}{\rv} \mid 
                 \load{\lv}{\hv} \mid 
                 \store{\hv}{\rv} \mid  \\
           &   & \plus{\lv}{\rv}{\rv} \mid 
                 \minus{\lv}{\rv}{\rv} \mid 
                 \lt{\lv}{\rv}{\rv} \mid 
                 \lea{\lv}{\rv} \mid 
                 \restricttwo{\lv}{\rv} \mid 
                 \subseg{\lv}{\rv}{\rv} \mid  \\
           &   & \isptr{\lv}{\rv} \mid 
                 \getp{\lv}{\lv} \mid 

                 \getl{\lv}{\lv} \mid 
                 \getb{\lv}{\lv} \mid
                 \gete{\lv}{\lv} \mid
                 \geta{\lv}{\lv} \mid \\
           &   & \fail \mid
                 \halt 
  \end{array}
\]
Further define $\reg_0 \in \Regs$ such that
\[
  \forall r \in \RegName \ldotp \reg_0(r) = 0
\]

\subsection{Operational Semantics}
Assume a $\decode$ function that decodes words to instructions:
\begin{align*}
  \decode &:\Words \fun \Instrs
\end{align*}
\dominique{mention that it is a simplification to take decode total?}
Assume an $\encodePerm$, $\encodeLoc$, and $\encodePermPair$ function that encodes a permissions, locality, and permission pair, respectively, as an integer:
\begin{align*}
  \encodePerm &: \Perms \fun \ints \\
  \encodeLoc &: \Globals \fun \ints \\
  \encodePermPair &: (\Perms \times \Globals) \fun \ints \\
\end{align*}
Further, assume a left inverse function, $\decodePermPair{}$, that decodes permissions
\[
  \decodePermPair{} : \ints \fun (\Perms \times \Globals)
\]

We define the operational semantics as follows:
\begin{align*}
  \Phi & \step \sem{\decode(\memheap(\addr))}(\Phi) & &                                   
                                                              \arraycolsep=0pt
                                                              \begin{array}{l}
                                                                \text{if $\memreg(\pcreg) = \stdcap$}\\
                                                                \quad\text{and $\start \leq \addr \leq \addrend$}\\
                                                                \quad\text{and $\perm \in \{ \exec,\rwx, \rwlx \}$ }
                                                              \end{array}\\
  \Phi & \rightarrow \failed                                 & & \text{otherwise}
\end{align*}
\lau{With respect to our talk about whether the upper bound should be included in the range of authority or not, I have found one example where it would work better when the upper-bound is not included. If we have a stack and call and want to pass the empty part on, then if the unused part of the stack is 0 cells, then we cannot pass anything that ``looks like a stack''.}
A number of functions and predicates used in the definition of $\sem{-}$ (defined later). Notice all of them are total.
\begin{align*}
  \readAllowed{\perm} &=
                        \begin{cases}
                          \true & \text{if } \perm \in \{ \rwx, \rwlx, \exec, \readwrite, \readwritel, \readonly \} \\
                          \false & \text{otherwise}
                        \end{cases} \\
  \writeAllowed{\perm} &=
                         \begin{cases}
                           \true &
                           \text{if } \perm \in \{ \rwx, \rwlx, \readwrite, \readwritel\} \\
                           \false & \text{otherwise}
                         \end{cases} \\
  \updatePcPerm{w} &=
                     \begin{cases}
                       ((\exec,\gl),\start,\addrend,\addr) & \text{if $w = ((\entry,\gl),\start,\addrend,\addr)$}\\
                       w & \text{otherwise} 
                     \end{cases} \\
  \nonZero{w} &=
                \begin{cases}
                  \true & \text{if $w\in \Caps$ or $w\in \ints$ and $w \neq 0$}\\
                  \false & \text{otherwise}
                \end{cases} \\
  \withinBounds{(\_,\start,\addrend,\addr)} &=
                                              \begin{cases}
                                                \true  & \text{if $\start \leq \addr \leq \addrend$} \\
                                                \false & \text{otherwise}
                                              \end{cases} \\
  \stdUpdatePc{\Phi} &=
                       \begin{cases}
                         \updateReg{\pcreg}{\var{newPc}} & 
                         \arraycolsep=0pt
                         \begin{array}[t]{l}
                           \text{if $\memreg(\pcreg) = \stdcap$}\\
                           \quad\text{and $\var{newPc} = ((\perm,\gl),\start,\addrend,\addr + 1)$}\\
                         \end{array} \\
                         \failed & \text{otherwise}
                       \end{cases} \\
\end{align*}
\begin{align*}
  \sem{\fail}(\Phi)                        & = \failed \\
  \sem{\halt}(\Phi)                        & = (\halted,\memheap) \\
  \sem{\jmp{\lv}}(\Phi)                    & = \updateReg{\pcreg}{\updatePcPerm{\memreg(\lv)}} \\
  \sem{\jnz{\lv}{\rv}}(\Phi)               & = 
                                             \begin{cases}
                                               \updateReg{\pcreg}{\updatePcPerm{\memreg(\lv)}} &
                                               \arraycolsep=0pt
                                               \begin{array}[t]{l}
                                                 \text{if $\nonZero{\memreg(\rv)}$} 
                                               \end{array}\\
                                               \stdUpdatePc{\Phi} & \text{if not $\nonZero{\memreg(\rv)}$}\\
                                               \failed & \text{otherwise }
                                             \end{cases} \\
  \sem{\load{\refreg{r_1}}{\refheap{r_2}}}(\Phi)  & = 
                                              \begin{cases}
                                                \stdUpdatePc{\updateReg{r_1}{\var{w}}} &
                                                \arraycolsep=0pt
                                                \begin{array}[t]{l}
                                                  \text{if }\memreg(r_2) = \stdcap = \var{c} \\
                                                  \quad\text{and }\readAllowed{\perm} \text{ and } \withinBounds{\var{c}} \\
                                                  \quad\text{and }\var{w} = \memheap(\addr)
                                                \end{array}\\
                                                \failed & \text{otherwise }
                                              \end{cases}\\
  \sem{\store{\refheap{r_1}}{\refreg{r_2}}}(\Phi) & = 
                                              \begin{cases}
                                                \stdUpdatePc{\updateHeap{\addr}{\var{w}}} &
                                                \arraycolsep=0pt
                                                \begin{array}[t]{l}
                                                  \text{if }\memreg(r_1) = \stdcap = \var{c} \\
                                                  \quad\text{and }\writeAllowed{\perm} \text{ and } \withinBounds{\var{c}} \\
                                                  \quad\text{and }\var{w} = \memreg(r_2)\\
                                                  \quad\text{and if } \var{w} = ((\_,\local),\_,\_,\_) \text{,} \\
                                                  \quad\text{ then } \perm \in \{\rwlx,\readwritel \}
                                                \end{array}\\
                                                \failed & \text{otherwise }
                                              \end{cases}\\
  \sem{\move{\refreg{r_1}}{\rv}}(\Phi)            & = 
                                                    \begin{cases}
                                                      \stdUpdatePc{\updateReg{r_1}{\rv}} & \rv \in \ints \\
                                                      \stdUpdatePc{\updateReg{r_1}{\memreg(\rv)}} & \text{otherwise}
                                                    \end{cases}
  \\
  \sem{\lea{\refreg{r_1}}{\rv}}(\Phi)            & =
                                             \begin{cases}
                                               \stdUpdatePc{\updateReg{r_1}{\var{c}}} &
                                               \arraycolsep=0pt
                                               \begin{array}[t]{l}
                                                 \text{if either $n = \rv$ or $\rv = \refreg{r_2}$ and $n = \memreg(r_2)$} \\
                                                 \quad\text{and in either case $n \in \ints $} \\
                                                 \quad\text{and $\memreg(r_1) = \stdcap$}\\
                                                 \quad\text{and $\perm \neq \entry$}\\
                                                 \quad\text{and $\var{c} = ((\perm,\gl),\start,\addrend,\addr + n)$}
                                               \end{array}\\
                                               \failed               & \text{otherwise}
                                             \end{cases} 
  \\
  \sem{\restricttwo{\refreg{r}}{\rv}}(\Phi)           & =
                                                  \begin{cases}
                                                    \stdUpdatePc{\updateReg{r}{\var{c}}}  &
                                                    \arraycolsep=0pt
                                                    \begin{array}[t]{l}
                                                      \text{if $\memreg(r) = \stdcap[\permp]$}\\
                                                      \quad\text{and either $\rv = n$ or $\memreg(\rv) = n$}\\
                                                      \quad\text{and in either case $n \in \ints$}\\
                                                      \quad\text{and $\decodePermPair{n}\sqsubseteq \permp$}\\ 
                                                      \quad\text{and $c = (\decodePermPair{n},\start,\addrend,\addr)$}
                                                    \end{array}\\
                                                    \failed                   & \text{otherwise}
                                                  \end{cases} 
\end{align*}
\begin{align*}
  \sem{\plus{\refreg{r_1}}{\rv_1}{\rv_2}}(\Phi)               & =
                                                          \begin{cases}
                                                            \stdUpdatePc{\updateReg{r_1}{n_1+n_2}} &
                                                            \arraycolsep=0pt
                                                            \begin{array}[t]{l}
                                                              \text{if for $i \in \{1,2\}$}\\
                                                              \quad\text{$n_i = \rv_i$ or $n_i = \memreg(\rv_i)$}\\
                                                              \quad\text{and in either case $n_i \in \ints$}
                                                            \end{array}\\
                                                            \failed & \text{otherwise}
                                                          \end{cases}\\
  \sem{\minus{\refreg{r_1}}{\rv_1}{\rv_2}}(\Phi)               & =
                                                          \begin{cases}
                                                            \stdUpdatePc{\updateReg{r_1}{n_1-n_2}} &
                                                            \arraycolsep=0pt
                                                            \begin{array}[t]{l}
                                                              \text{if for $i \in \{1,2\}$}\\
                                                              \quad\text{$n_i = \rv_i$ or $n_i = \memreg(\rv_i)$}\\
                                                              \quad\text{and in either case $n_i \in \ints$}
                                                            \end{array}\\
                                                            \failed & \text{otherwise}
                                                          \end{cases}\\
  \sem{\lt{\refreg{r_1}}{\rv_1}{\rv_2}}(\Phi)               & =
                                                          \begin{cases}
                                                            \stdUpdatePc{\updateReg{r_1}{1}} &
                                                            \arraycolsep=0pt
                                                            \begin{array}[t]{l}
                                                              \text{if for $i \in \{1,2\}$}\\
                                                              \quad\text{$n_i = \rv_i$ or $n_i = \memreg(\rv_i)$}\\
                                                              \quad\text{and in either case $n_i \in \ints$}\\
                                                              \quad\text{and $n_1 < n_2$}\\                                                                                                            \end{array}\\
                                                            \stdUpdatePc{\updateReg{r_1}{0}} &
                                                            \arraycolsep=0pt
                                                            \begin{array}[t]{l}
                                                              \text{if for $i \in \{1,2\}$}\\
                                                              \quad\text{$n_i = \rv_i$ or $n_i = \memreg(\rv_i)$}\\
                                                              \quad\text{and in either case $n_i \in \ints$}\\
                                                              \quad\text{and $n_1 \not< n_2$}\\                                                                                                            \end{array}\\
                                                            \failed & \text{otherwise}
                                                          \end{cases}
\\
  \sem{\subseg{\refreg{r}}{\rv_1}{\rv_2}}(\Phi) & = 
                                            \begin{cases}
                                              \stdUpdatePc{\updateReg{r}{\var{c}}} &
                                              \arraycolsep=0pt
                                              \begin{array}[t]{l}
                                                \text{if $\memreg(r) = \stdcap$} \\
                                                \quad\text{and for $i \in \{1,2\}$}\\
                                                \quad\text{$n_i = \rv_i$ or $n_i = \memreg(\rv_i)$}\\
                                                \quad\text{and in either case $n_1 \in \nats$}\\
                                                \quad\text{and $\start \leq n_1$}\\
                                                \quad\text{and $n_2 \leq \addrend$ where $n_2 \in \nats$}\\
                                                \quad\quad\text{or $n_2=\inftyend$ and $\addrend = \infty$}\\
                                                \quad\text{and $\perm \neq \entry$}\\
                                                \quad\text{and $c = ((\perm,\gl),n_1,n_2,\addr)$}
                                              \end{array} \\
                                              \failed & \text{otherwise}
                                            \end{cases}
  \\
  \sem{\geta{\refreg{r_1}}{\refreg{r_2}}}(\Phi) & = 
                                            \begin{cases}
                                              \stdUpdatePc{\updateReg{r_1}{\addr}} &
                                              \arraycolsep=0pt
                                              \begin{array}[t]{l}
                                                \text{if $\memreg(r_2) = ((\_,\_),\_,\_,\addr)$}
                                              \end{array} \\
                                              \failed & \text{otherwise}
                                            \end{cases}
  \\
  \sem{\getb{\refreg{r_1}}{\refreg{r_2}}}(\Phi) & = 
                                            \begin{cases}
                                              \stdUpdatePc{\updateReg{r_1}{\start}} &
                                              \arraycolsep=0pt
                                              \begin{array}[t]{l}
                                                \text{if $\memreg(r_2) = ((\_,\_),\start,\_,\_)$}
                                              \end{array} \\
                                              \failed & \text{otherwise}
                                            \end{cases}
  \\
  \sem{\gete{\refreg{r_1}}{\refreg{r_2}}}(\Phi) & = 
                                            \begin{cases}
                                              \stdUpdatePc{\updateReg{r_1}{\addrend}} &
                                              \arraycolsep=0pt
                                              \begin{array}[t]{l}
                                                \text{if $\memreg(r_2) = ((\_,\_),\_,\addrend,\_)$ and $\addrend \neq \infty$}
                                              \end{array} \\
                                              \stdUpdatePc{\updateReg{r_1}{\inftyend}} &
                                              \arraycolsep=0pt
                                              \begin{array}[t]{l}
                                                \text{if $\memreg(r_2) = ((\_,\_),\_,\infty,\_)$}
                                              \end{array} \\
                                              \failed & \text{otherwise}
                                            \end{cases}
  \\
  \sem{\getp{\refreg{r_1}}{\refreg{r_2}}}(\Phi) & = 
                                            \begin{cases}
                                              \stdUpdatePc{\updateReg{r_1}{\encodePerm(\perm)}} &
                                              \arraycolsep=0pt
                                              \begin{array}[t]{l}
                                                \text{if $\memreg(r_2) = ((\perm,\_),\_,\_,\_)$}
                                              \end{array} \\
                                              \failed & \text{otherwise}
                                            \end{cases}
  \\
  \sem{\getl{\refreg{r_1}}{\refreg{r_2}}}(\Phi) & = 
                                            \begin{cases}
                                              \stdUpdatePc{\updateReg{r_1}{\encodeLoc(\gl)}} &
                                              \arraycolsep=0pt
                                              \begin{array}[t]{l}
                                                \text{if $\memreg(r_2) = ((\_,\gl),\_,\_,\_)$}
                                              \end{array} \\
                                              \failed & \text{otherwise}
                                            \end{cases}
  \\
  \sem{\isptr{\refreg{r}}{\rv}}(\Phi) & =  
                                  \begin{cases}
                                    \stdUpdatePc{\updateReg{r_1}{1}} & \text{if $\memreg(\rv) \in \Caps$ } \\
                                    \stdUpdatePc{\updateReg{r_1}{0}} & \text{otherwise} 
                                  \end{cases}
\end{align*}
\lau{Our instruction set could use a less than operator.}
\lau{Dominique, apparently there was a reason I had not added the new condition in store to $\writeAllowed{}$. The new condition depends on a case distinction on $w$. It is only if $w$ is a capability that we look into whether it is local or not. But even if $w$ is just an integer, we still need the capability we write through to have some kind of write permission.}
\dominique{I see, but maybe writeAllowed could take $w$ and $c$ as an argument, ratheppppr than just $\perm$? It is now weird that some conditions for writing are in $\writeAllowed{}$, but not all.}

\dominique{factor out reoccurring conditions of the form $n_i = \rv_i$ or $n_i =
  \memreg(\rv_i)$ into a function that evaluates an $\rv$ using a system state?
  This would generalise nicely to the case where we would have additional
  addressing modes.}

\dominique{why doesn't store take an $\rv$?}

Define the following macros: $\mathtt{restrict}$, $\mathtt{subseg}$, and $\mathtt{lea}$ that does not overwrite the source register. A $\mathtt{store}$ that allows integers to be stored directly. $\mathtt{store}$ requires a register $r_t$ for storage of temporary values to be available.
\begin{align*}
  \restrict{r_1}{r_2}{r_3} \; r_4 &\defeq
                                    \begin{aligned}[t]
                                      & \move{r_1}{r_2} \\
                                      & \restrict{r_1}{r_3}{r_4}
                                    \end{aligned} \\
  \subseg{r_1}{r_2}{r_3} \; r_4   &\defeq
                                    \begin{aligned}[t]
                                      & \move{r_1}{r_2} \\
                                      & \subseg{r_1}{r_3}{r_4}
                                    \end{aligned}\\
  \lea{r_1}{r_2} \; r_3           &\defeq
                                    \begin{aligned}[t]
                                      & \move{r_1}{r_2} \\
                                      & \lea{r_1}{r_3}
                                    \end{aligned}\\
  \store{r}{n} &\defeq
                 \begin{aligned}[t]
                   & \move{r_t}{n} \\
                   & \store{r}{r_t}
                 \end{aligned}
\end{align*}

\begin{lemma}[Determinacy]
  \label{lem:determinacy}
  If $\Phi \step \Phi'$ and $\Phi \step \Phi''$, then $\Phi' = \Phi''$.
  If $\Phi \step[n] \Phi'$ and $\Phi \step[n] \Phi''$, then $\Phi' = \Phi''$.
  If $\Phi \step[n] \Phi'$ and $\Phi \step[n'] (\halted,\mem'')$, then $n \leq n'$
  and $\Phi' \step[n'-n] (\halted,\mem'')$.
\end{lemma}
\begin{proof}
  By easy inspection of the definition of the operational semantics.
\end{proof}

\section{Malloc specification}
\newcommand{\hsfoot}{\hs_\var{footprint}}
\newcommand{\hsframe}{\hs_\var{frame}}
\newcommand{\size}{\var{size}}
\newcommand{\rio}{r_{io}}
\newcommand{\advb}{\var{adv_{base}}}
\newcommand{\adve}{\var{adv_{end}}}
\newcommand{\initb}{\var{init}_{base}}
\newcommand{\inite}{\var{init}_{end}}
\newcommand{\mrlen}{5cm}
\newcommand{\retm}{\var{ret}_{\malloc}}
\newcommand{\reta}{\var{ret}_{\adv}}
\newcommand{\base}{\var{base}}
\newcommand{\eend}{\var{end}}
\newcommand{\bracket}[1]{\multirow{#1}{*}{\ensuremath{
      \left . \vphantom{\begin{array}{l}
                          \ifthenelse{\equal{#1}{1}}{3\\}{
                          \ifthenelse{\equal{#1}{2}}{3\\3\\}{
                          \ifthenelse{\equal{#1}{3}}{3\\3\\3\\}{
                          \ifthenelse{\equal{#1}{4}}{3\\3\\3\\3\\}{
                          \ifthenelse{\equal{#1}{5}}{3\\3\\3\\3\\3\\}{
                          \ifthenelse{\equal{#1}{6}}{3\\3\\3\\3\\3\\3\\}{
                          3\\3\\3\\3\\3\\3\\3\\ 
                          }}}}}}
                        \end{array}} \right \}}}
              }
\newcommand{\annotate}[2]{\multirow{#1}{\mrlen}{\scriptsize #2}}

              \begin{specification}[Malloc Specification]
                \label{spec:malloc}
                $c_\malloc$ satisfies the specification for malloc iff
                \[  
                  \begin{aligned}
                    &c_\malloc = ((\entry,\glob),\_,\_,\_) \land\\
                    &\exists \iota_{\malloc,0} \ldotp \\
                    &\quad (\forall \iota' \futurestr \iota_{\malloc,0} \ldotp \forall W,i \ldotp W(i)=\iota' \Rightarrow \iota'.H (\iota'.s) (\xi^{-1}(W)) = \iota'.H (\iota'.s) (\xi^{-1}([i \mapsto W(i)])) ) \; \land \\
                    &\quad \iota_{\malloc,0}.v = \perma \; \land \\
                    &\quad (\forall \Phi \in \ExecConfs \ldotp \forall \ms_{\var{footprint}}, \hsframe \in \HeapSegments \ldotp \\
                    &\qquad \forall i, n, \size \in \nats \ldotp \forall
                    w_{\var{ret}} \in \Words \ldotp \\
                    &\qquad \quad \forall \iota_\malloc \futurestr \iota_{\malloc,0} \land \\
                    &\qquad \quad \memheap = \ms_{\var{footprint}} \uplus \hsframe \land \heapSat[\ms_{\var{footprint}}]{n}{[i \mapsto \iota_\malloc]} \land \\
                    &\qquad \quad \memreg(r_1) = \size \land \size \geq 0 \land  \memreg(r_0) = w_{\var{ret}} \land \\
                    &\qquad \quad \memreg(\pcreg) = \updatePcPerm{c_\malloc} \\ 
                    &\qquad \quad \Rightarrow \\
                    &\qquad \qquad\exists \Phi' \in \ExecConfs \ldotp \exists \ms_{\var{footprint}}', \ms_{\var{alloc}} \in \HeapSegments\ldotp\\
                    &\qquad \qquad \quad \exists j \in \nats \ldotp j > 0 \land \exists b',e'\in \Addrs \ldotp \exists \iota_\malloc' \in \Regions \ldotp \\
                    &\qquad \qquad \qquad \Phi \step[j] \Phi' \land \\
                    &\qquad \qquad \qquad \memheap[\Phi']=\ms_{\var{footprint}}' \uplus \hs_{\var{alloc}} \uplus \hsframe \land\\
                    &\qquad \qquad \qquad \iota_{\malloc}' \futurewk \iota_\malloc \land \\
                    &\qquad \qquad \qquad \heapSat[\ms_{\var{footprint}}']{n-j}{[i \mapsto \iota_\malloc']} \land \\
                    &\qquad \qquad \qquad \dom(\hs_{\var{alloc}}) = [b',e'] \land \forall a \in [b',e']\ldotp \hs_{\var{alloc}}(a) = 0  \land \\
                    &\qquad \qquad \qquad \memreg[\Phi'] = \memreg[\Phi]\update{\pcreg}{\updatePcPerm{w_{\var{ret}}}}\update{r_1}{((\rwx,\glob),b',e',b')} \land \\
                    &\qquad \qquad \qquad \size - 1 = e'-b' ) \land\\
                    &\quad (\forall \Phi \in \ExecConfs \ldotp (\memreg(r_1)
                    \not\in \ints \vee \memreg(r_1) < 0) \land \memreg(\pcreg) = \updatePcPerm{c_\malloc} \Rightarrow \exists j \in \nats \ldotp \Phi \step[j] \failed)
                  \end{aligned}
                \]
              \end{specification}
              In the specification above $\iota_\malloc'$ is a future region of the initial region that governs malloc.
              \lau{What about $c_\malloc$ and the value relation? Do we add this to the specification or try to prove it based on the specification?}

              \section{Macros}
              In order to write readable example programs, we provide macros (macro-instructions) that can be implemented in terms of the instruction set given in the formalisation.
              \dominique{Similarly: how is malloc invoked?  Do we trust malloc enough to not encapsulate ourselves from it, i.e. provide an rx return capability and use callee-save registers?}
              \lau{ I think this is a conceptual question as we can make it work with either. We already trust malloc to give us a fresh piece of memory and not reuse it later on, so we already assume malloc to be somewhat trusted, so why not go all the way? }
              \lau{Agreed.}

              \dominique{what does ``fetch the capability ...'' mean?}
              \lau{ We don't know where the capability resides, but if it is in memory, then it will be loaded into a register. }
              \dominique{Wouldn't it be more clear to provide call with two explicit lists of registers: those which need to be stored, and those which are provided as arguments (i.e. which do not need to be erased)?}
              \dominique{Perhaps you could also provide an explicit syntax for ``undefined symbols'' that should be filled in by a linker?}

              In order to compute offsets and the like, the macros need registers to keep temporary computations in. We assume such a small set of registers $\RegName_t \subseteq \RegName$ is available and that $\RegName_t$ does not contain registers explicitely named in a program nor $r_0$, $r_\stk,$ or $\pcreg$ (but clearing all registers still clears the temporary registers).

              \subsection{Linking and ABI}
              In order to make capabilities to trusted code (and possibly untrusted code) available, we assume that some sort of linker has made these available. This is done in the following way: For every function, the first memory cell the capability for that function governs contains a capability for the linking table. Each function name in a program corresponds to an offset in the table, e.g., \texttt{malloc} could be at offset 0. When a name is used in a program, it indicates what entry from the linking table to pick. The table should always be accessible by taking a copy of the capability in the $\pcreg$-register and adjusting it to point to the first cell it governs.

              The capability linking table can be shared between multiple functions that are linked to the same capabilities as it is accessed through read-only capabilities.

              \subsection{Flag table}
              A function may use flags to signal failure. We use the convention that a flag table is available in the second memory cell of a functions code (so just after the linking table). The flag table is accessed through a read-write capability and initially it contains all zero. Like the linking table, each entry is associated with a name which may appear in the macros.

              The flag table should never be shared between distrusting parties.

              We will often want to make room in memory for a linking-table capability and a flag-table capability. We therefore define a constant that represents the offset of the actual code of a function caused by these two capabilities:
              \[
                \olf \defeq 2
              \]

              \subsection{Macro definitions}
In the following, we describe each of the macros. The descriptions are so detailed that it should be a simple matter to implement the macros. We provide a proposed implementation for each of the macros in order to install some confidence in the fact that it is possible to implement each of the macro.
              \begin{description}
              \item[\texttt{fetch} $r$ $f$] load the entry of the linking table corresponding to $f$ to register $r$.\\
One possible fetch implementation (r\_t1 and r\_t2 are registers in RegName\_t).
\begin{lstlisting}
move r pc
getb r_t1 r
geta r_t2 r
minus r_t1 r_t1 r_t2 // Offset to first address, i.e., linking table (b-a)
lea r r_t1
load r r
lea r ... // ... replaced with offset to f in the linking table
move r_t1 0
move r_t2 0
load r r // f capability loaded to register r
\end{lstlisting}
              \item[\texttt{call} $r(\bar{r}_{\var{args}},\bar{r}_{\var{priv}})$] \forcenewline
                $\bar{r}_{\var{args}}$ and $\bar{r}_{\var{priv}}$ are lists of registers. An overview of this call:
                \begin{itemize}
                \item Set up activation record
                \item Create local enter capability for activation (protected return pointer)
                \item Clear unused registers
                \item Jump
                \item Upon return: Run activation code
                \end{itemize}
                A more detailed description of each of the above steps:
                \begin{description}
                \item [Set up activation record]\forcenewline
                  \begin{itemize}
                  \item Run malloc to get a piece of memory with space for:
                    \begin{itemize}
                    \item Words in $\bar{r}_{\var{priv}}$
                    \item Code return capability (opc)
                    \item Activation code
                    \end{itemize}
                  \item Store the words in $\bar{r}_{\var{priv}}$ to the activation record. 
                  \item Adjust a copy of the current pc to point to the return address in code and save it to the activation record.
                  \item Write the activation code to the activation record.
                  \end{itemize}
                \item [Create local enter capability for activation] Adjust the capability for the activation record to point to the beginning of the activation record and restrict it to a local enter-capability. Place this capability in $r_0$.
                \item [Clear unused registers]
                  Clear all the register that are not $\pcreg$, $r$, $r_0$ or in $\bar{r}_{\var{args}}$.
                \item [Jump] Jump to register $r$
                \item [Activation code] The activation code does the following:
                  \begin{itemize}
                  \item Move the stored ``private'' words in to their respective $\bar{r}_{\var{priv}}$ registers.
                  \item Load the return capability to $\pcreg$
                  \end{itemize}
                \end{description} 
Possible implementation. We will use $\texttt{malloc $r$ $n$}$ and \texttt{rclear $\bar{r}$} (defined below). Assume $\bar{r_{\var{priv}}} = r_{\var{priv},1}, \dots, r_{\var{priv},n}$
\begin{lstlisting}
  malloc r_t ... // ... is the size of activation record
// store private state in activation record
  store r_t r_priv,1
  lea r_t 1
  store r_t r_priv,2
  lea r_t 1
  ...
  lea r_t 1
  store r_t r_priv,n
  lea r_t 1
// store old pc
  move r_t1 pc
  lea r_t1 ... // ... is the offset to return address
  store r_t r_t1
  lea r_t1 1
// store activation record
  store r_t encode(i_1)
  lea r_t1 1
  ...
  lea r_t1 1  
  store r_t encode(i_m)
  lea r_t1 k //  k is m-1, i.e. the offset to the first instruction of the activation code.
  restrict r_t1 encodePermPair((Local,e))
  move r_0 r_t1
  rclear R // R = RegisterName - {r,pc,r_0,r_args}
  jmp r
\end{lstlisting}
Activation record. The instructions correspond to $i_1,\dots,i_m$ in the above.
\begin{lstlisting}
  move r_t pc
  getb r_t1 r_t
  geta r_t2 r_t
  minus r_t1 r_t1 r_t2
// load private state
  lea r_t r_t1
  load r_priv,1 r_t
  lea r_t 1
  load r_priv,2 r_t
  lea r_t 1
  ...
  lea r_t 1
  load r_priv,n r_t
  lea r_t 1
// load old pc
  load pc r_t
\end{lstlisting}
 
              \item[\texttt{malloc $r$ $n$}] Calls malloc to allocates a piece of memory of size $n$. The capability will be stored in register $r$. 
One possible malloc implementation (r\_t1 is a register in RegName\_t) and r\_1 is the register from the malloc specification.
\begin{lstlisting}
fetch r malloc
move r_1 n 
// save return pointer
move r_t1 r_0 
// setup new return pointer
move r_0 pc
lea r_0 4 // 4 is the offset to just after jmp r
restrict r_0 encodePerm(e)
jmp r
move r r_1
move r_0 r_t1 // restore return pointer
move r_1 0
move r_t1 0
\end{lstlisting}
              \item[\texttt{assert$_{\var{flag}}$ $r_1$ $r_2$}] Compares the words in register $r_1$ and $r_2$ (if one of them is an integer, then use that in the comparison). If they are equal, then execution continues. If they are unequal, then the assertion flag named $\var{flag}$ in the flag list is set to 1 and execution halts (if no flag is specified, then the first flag in the list is set to 1).\\
There are four different asserts based on whether $r_1$ and $r_2$ are registers or numbers. If $r_1$ and $r_2$ are registers:
\begin{lstlisting}
          // setup pointer to fail.
          move r_t3 pc
          lea r_t3 ... // ... is the offset to fail
          // make sure both registers contain either capability or integer
          isptr r_t1 r_1
          isptr r_t2 r_2
          minus r_t1 r_t1 r_t2
          jnz   r_t3 r_t1
          // set up capability for cap case:
          move r_t4 pc
          lea r_t4 ... // ... is the offset to caps
          jnz r_t4 r_t2 // jump to caps if r_t2 contains a capability
          // the two registers contain an integer
          minus r_t1 r_1 r_2
          jnz r_t3 r_t1
          // the two integers in the registers are equal
          move r_t4 pc
          lea r_t4 ... // .. offset to success
caps:
          geta r_t1 r_1
          geta r_t2 r_2
          minus r_t1 r_t1 r_t2
          jnz   r_t3 r_t1          
          getb r_t1 r_1
          getb r_t2 r_2
          minus r_t1 r_t1 r_t2
          jnz   r_t3 r_t1          
          gete r_t1 r_1
          gete r_t2 r_2
          minus r_t1 r_t1 r_t2
          jnz   r_t3 r_t1          
          getp r_t1 r_1
          getp r_t2 r_2
          minus r_t1 r_t1 r_t2
          jnz   r_t3 r_t1          
          getl r_t1 r_1
          getl r_t2 r_2
          minus r_t1 r_t1 r_t2
          jnz   r_t3 r_t1          
          // the two capabilities in the registers are equal
          move r_t4 pc
          lea r_t4 ... // .. offset to success
fail:
          // get the flag capability
          move r_t3 pc
          getb r_t1 pc
          geta r_t2 pc
          minus r_t1 r_t1 r_t2
          lea r_t3 r_t1
          lea r_t3 1 // the flag table capability is at the second address of cap.
          load r_t1 r_t3
          lea r_t1 ... // ... is the offset of flag in the table
          store r_t1 1
          halt
success:                   
          // clean up
          move r_t1 0
          move r_t2 0
          move r_t3 0
          move r_t4 0
\end{lstlisting}
If $r_1$ is a register, but $r_2$ is a constant:
\begin{lstlisting}
          // setup pointer to fail.
          move r_t3 pc
          lea r_t3 ... // ... is the offset to fail
          // make sure both registers contain either capability or integer
          isptr r_t1 r_1
          jnz   r_t3 r_t1
          minus r_t1 r_1 r_2
          jnz r_t3 r_t1
          // the two integers in the registers are equal
          move r_t3 pc
          lea r_t3 ... // .. offset to success
fail:
          // get the flag capability
          move r_t3 pc
          getb r_t1 pc
          geta r_t2 pc
          minus r_t1 r_t1 r_t2
          lea r_t3 r_t1
          lea r_t3 1 // the flag table capability is at the second address of cap.
          load r_t1 r_t3
          lea r_t1 ... // ... is the offset of flag in the table
          store r_t1 1
          halt
success:                   
          // clean up
          move r_t1 0
          move r_t3 0
\end{lstlisting}
The case where $r_1$ is a constant and $r_2$ is a register is omitted. The case where both are constant is also omitted - if the constants are the same, then the macro is nothing. If they are different, then it corresponds to the failed part of both of the above implementations.
              \item[\texttt{mclear $r$}] Stores 0 to all the memory cells the capability $r$ governs.\footnote{This may in some cases seem like an unreasonable slow instruction. In a real system it would probably be implemented as a vector operation which allows modification of continuous segments of memory rather fast.} \\
Possible implementation:
\begin{lstlisting}
        move r_t r
        getb r_t1 r_t
        geta r_t2 r_t
        minus r_t2 r_t1 r_t2
        lea r_t r_t2
        gete r_t2
        minus r_t1 r_t2 r_t1
        plus r_t1 r_t1 1
        move r_t2 pc
        lea r_t2 ... // ... is the offset to end
        move r_t3 pc
        lea r_t3 ... // ... is the offset to iter
iter:
        jnz r_t2 r_t1
        store r_t 0
        lea r_t 1
        plus r_t1 r_t1 1
        jmp r_t3
end:
        move r_t 0
        move r_t1 0
        move r_t2 0
        move r_t3 0
\end{lstlisting}
              \item[\texttt{rclear $\bar{r}$}] Moves 0 to all the registers in the list $\bar{r}$.\\
Possible implementation: Say $\bar{r} = r_1,\dots, r_n$
\begin{lstlisting}
move r_1 0
move r_2 0
// ...
move r_n 0
\end{lstlisting}

\end{description}
Note:
\begin{itemize}
\item \texttt{call} will fail if we have local capabilities in one of the registers of the ``private'' register list as it relies on a capability returned by malloc which will not be permit-write-local. This severely limits how \texttt{scall} can be used and it provides very little in terms of control-flow integrety when nested. Below, we introduce \texttt{scall} which can handle local capabilities in the ``private'' state. 
\end{itemize}
\subsection{Stack}
Some programs will assume access to a stack which will be in part indicated by the program macros but also in the correctness lemma. The stack is accessed through a local $\rwlx$-capability. Programs will assume that the stack resides in some register, say $r_{\var{stk}}$.

The stack resides entirely in memory. There is no separation between the memory and the stack, so when we talk about the stack it is as a conceptual thing. %

Even though the memory is infinite, we will only use a finite part for the stack. If we have allocated too little memory for the stack, and we try to push something anyway, then the execution will fail. As we consider failing admissible, we are okay with this. 

When not in the middle of a push or a pop, the stack capability points to the top word of the stack. For an empty stack, the stack capability points to the address just below of the range of authority for the stack capability.

The stack grows upwards

\begin{description}
\item[\texttt{push $r$}] Pushes the word in register \texttt{r} to the stack by incrementing the address of the stack capability by one and storing the word through the stack capability.\\
Possible implementation:
\begin{lstlisting}
  lea r_stk 1
  store r_stk r
\end{lstlisting}

\item[\texttt{pop $r$}] Pops the top word of the stack by loading it to register $r$, and decrementing the address of the stack capability.
\begin{lstlisting}
  load r r_stk
  minus r_t1 0 1
  lea r_stk r_t1
\end{lstlisting}

\item[\texttt{scall} $r(\bar{r}_{\var{args}},\bar{r}_{\var{priv}})$] \forcenewline$\bar{r}_{\var{args}}$ and $\bar{r}_{\var{priv}}$ are lists of registers. This call assumes $r_{\var{stk}}$ contains a stack capability. An overview of this call:
  \begin{itemize}
  \item Push ``private'' registers to the stack.
  \item Push the restore code to the stack.
  \item Push return address capability 
  \item Push stack capability
  \item Create protected return pointer
  \item Restrict stack capability to unused part 
  \item Clear the part of the stack we release control over
  \item Clear unused registers
  \item Jump
  \item Upon return: Run the on stack restore code
  \item Return address in caller-code:  Restore ``private'' state
  \end{itemize}
  A more detailed description of the above steps:
  \begin{description}
  \item [Push ``private'' registers to the stack]
    Push all the words in the registers in  $\bar{r}_{\var{priv}}$ to the stack.
  \item [Push the restore code to the stack]
    Push the restore code to the stack (described later). This code needs to be on the stack to make sure the stack capability can be restored. We keep the restore code on the stack minimal. The caller code does the rest of the restoration.
  \item [Push return address capability]
    Push a capability for the return address (in the memory) to the stack.
  \item [Push stack capability]
    Push the full stack capability to the stack.
  \item [Create protected return pointer]
    Make a new version of the stack pointer that points to the beginning of the restoration code. Restrict it to a local enter-capability and put it in $r_0$. \lau{Do we want to use ``protected return pointer'' to mean a local enter capability?} \lau{It is also used when we want to call code for the first time, so calling it a protected return pointer is more of a conceptual thing.}
  \item [Restrict stack capability to unused part]
    Make the stack capability only govern the unused part.
  \item [Clear the part of the stack we release control over]
    Store 0 to all the memory cells the restricted stack pointer has authority over.
  \item [Clear unused registers]
    Clear all registers but $\pcreg$, $r$, $r_0$, $r_{\var{stk}}$, and $\bar{r}_{\var{args}}$.
  \item [Jump] Jump to register $r$.
    \\ \lau{scall used to be dependent on linking, but now this part has been extracted to the fetch macro to make scall independent of linking.}
  \item [Run the on stack restore code]
    Load the stack capability to $r_{\var{stk}}$. Pop the old program counter (the return address in caller-code) from the stack to $\pcreg$.
  \item [Return address in caller-code: Restore ``private'' state] \forcenewline
    \begin{itemize}
    \item Pop the restore code of the stack
    \item Pop the private state on the stack into their respective $\bar{r}_{\var{priv}}$ registers.
    \end{itemize}
  \end{description}
\end{description}
Possible implementation, say $\bar{r}_{\var{args}} = r_{\var{args},1},\dots, r_{\var{args},m}$ and $\bar{r}_{\var{priv}} = r_{\var{priv},1},\dots, r_{\var{priv},n}$:
\begin{lstlisting}
// push private state
  push r_priv,1
  ...
  push r_priv,n
// push activation code
  push encode(i_1)
  ...
  push encode(i_4)
// push old pc
  move r_t1 pc
  lea r_t1 ... // ... is the offset to after
  push r_t1
// push stack pointer
  push r_stk
// set up protected return pointer
  move r_0 r_stk
  lea r_0 -5 // -5 is the offset to the first instruction of the activation code
  restrict r_0 encodePermPair((Local,e))
// restrict stack capability
  geta r_t1 r_stk
  plus r_t1 r_t1 1
  getb r_t2 r_stk
  subseg r_stk r_t1 r_t2
// clear unused part of the stack
  mclear r_stk
// clear non-argument registers
  rclear R // where R = RegisterName - {pc,r_stk,r_0,r,r_args}
  jmp r
after:
// pop the restore code
  pop r_t1
  pop r_t1
  pop r_t1
  pop r_t1
// pop the private state into approriate registers
  pop r_priv,1
  ...
  pop r_priv,n
  
\end{lstlisting}
where the restore code is as follows:
\begin{lstlisting}
 i_1 = move r_t1 pc
 i_2 = lea r_t1 5 // 5 is the offset to the address where the old stack pointer is located
 i_3 = load r_stk r_t1
 i_4 = pop pc
\end{lstlisting}
Note:
\begin{itemize}
\item If we want to have local capabilities as part of our private state, then we need to have a stack and use \texttt{scall}. If we do not have any local capabilities we want to keep around, then we can use \texttt{call}, but it will incur a small memory leak as the activation records cannot be recycled! It is also possible to use a combination of \texttt{scall} and \texttt{call}, but when \texttt{call} is used, then we have no way to store the stack, so we cannot use \texttt{scall} after that.
\item As a rule of thumb: If you have provided an untrusted entity access to part of the stack, then it needs to be cleared before it is passed to an untrusted party.
\item As a rule of thumb: If you receive a stack from an untrusted source, then you need to check that it is a local $\rwlx$-capability and clear it! If any callbacks are provided, then they need to be global.
\end{itemize}
\begin{description}
\item[\texttt{crtcls $[(x_1,r_1),\dots(x_n,r_n)]$ $r_{\var{code}}$}] \forcenewline
  $[(x_1,r_1),\dots(x_n,r_n)]$ is a list of variable bindings. If an instruction refers to a variable, then it will assume that an environment is available in a designated register (say $r_{\var{env}}$). The register $r_{\var{code}}$ should contain a capability governs the code of the closure and that is executable when jumped to.
  \begin{description}
  \item[Allocate memory for variable environment]
  \item[Store register contents to environment]
  \item[Allocate memory for record with environment capability, code capability, and activation code]
  \item[Store capabilities and activation code to record]
  \item[Restrict the capability for the ``closure pair'' to an enter capability]
  \item[Activation code:] \forcenewline
    \begin{itemize}
    \item Load the environment capability to a designated register
    \item Load the code capability.
    \item Jump to the code.
    \end{itemize}
  \end{description}
  A more detailed description of each step:
  \begin{description}
  \item[Allocate memory for variable environment] Have malloc allocate a piece of memory of size $n$ (the size of the variable environment). 
  \item[Store register contents to environment] Store the contents of each of the registers $r_1,\dots,r_n$ to the newly allocated memory.
  \item[Allocate memory for record with environment capability, code capability, and activation code] Allocate a new piece of memory with room for a capability for the environment.
  \item[Store capabilities and activation code to record] Store the environment capability and code capability in the record followed by the activation code. 
  \item[Restrict the capability for the ``closure pair'' to an enter capability] Adjust the capability to point to the start of the activation code and restrict it to a global enter-capability.
  \item[Activation code:] \forcenewline
    \begin{itemize}
    \item Load the environment capability to a designated register.
    \item Load the code capability.
    \item Jump to the code.
    \end{itemize}
  \end{description}
Possible implementation of \texttt{crtcls $\overline{(x,r_v)}$ $r_{\var{code}}$} where $|\overline{(x,r_v)}| = n$ ($i_1$,...,$i_6$, i.e. the activation code, is defined later):
\begin{lstlisting}
malloc r_t1 n
store r_t1 r_v1
lea r_t1 1
store r_t1 r_v2
lea r_t1 1
...
lea r_t1 1
store r_t1 r_vn
lea r_t1 -n
restrict r_t1 encodePermPair((Global,rw))
malloc r_1 8 //length of activation record
store r_1 r_code // code capability
lea r_1 1
store r_1 r_t1 // environment capability
move r_t1 0
lea r_1 1
store r_1 encode(i_1)
lea r_1 1
store r_1 encode(i_2)
lea r_1 1
...
lea r_1 1
store r_1 encode(i_6)
lea r_1 -5 //offset to first instruction
restrict r_1 encodePerm(e)
\end{lstlisting}
Activation code ($i_1$,...,$i_6$):
\begin{lstlisting}
i_1 = move r_t1 pc
i_2 = lea r_t1 -2
i_3 = load r_env r_t1
i_4 = lea r_t1 1
i_5 = load r_t1 r_t1
i_6 = jmp r_t1
\end{lstlisting}

\item[\texttt{load $r$ $x$}] Assumes environment capability available in register $r_{\var{env}}$. Loads the word at the index associated with $x$ in the environment list. Loads from this capability into $r$.\\
Possible implementation:
\begin{lstlisting}
  move r_t1 r_env
  lea r_t1 ... // ... corresponds to offset of x in environment
  load r r_t1
  move r_t1 0
\end{lstlisting}
\item[\texttt{store $x$ $r$}] Assumes environment capability available in register $r_{\var{env}}$. Loads the word at the index associated with $x$ in the environment list. Stores the contents of register $r$ through this capability.
\begin{lstlisting}
  move r_t1 r_env
  lea r_t1 ... // ... corresponds to offset of x in environment
  store r_t1 r
  move r_t1 0
\end{lstlisting}
\item[\texttt{reqglob $r$}] Tests if register $r$ contains a $\glob$ capability. If not fail, otherwise continue execution.\\
Possible implementation:
\begin{lstlisting}
  getl r_t1 r
  minus r_t1 r_t1 encodeLoc(Global)
  move r_t2 pc
  lea r_t2 4 // 4 is the offset to just after fail
  jnz r_t1 r_t2
  fail
  move r_t1 0
  move r_t2 0
\end{lstlisting}
\item[\texttt{reqperm $r$ $n$}] Tests if register $r$ contains a capability with permission $\decodePerm{n}$. If not fail, otherwise continue execution.\\
Possible implementation:
\begin{lstlisting}
  getp r_t1 r
  minus r_t1 r_t1 n
  move r_t2 pc
  lea r_t2 4 // 4 is the offset to just after fail
  jnz r_t1 r_t2
  fail
  move r_t1 0
  move r_t2 0
\end{lstlisting}

\item[\texttt{prepstack} $r$] Tests if register $r$ contains a capability with permission $\rwlx$. If not fail, otherwise assume $r$ points to $((\rwlx,\gl),\start,\addrend,\addr)$ adjust it to $((\rwlx,\gl),\start,\addrend,\start - 1)$.\\
Possible implementation
\begin{lstlisting}
  reqperm r encodePerm(rwlx)
  getb r_t1 r
  geta r_t2 r
  minus r_t1 r_t1 r_t2
  lea r r_t1
  minus r_t1 0 1
  lea r r_t1
  move r_t1 0
  move r_t2 0
\end{lstlisting}
  \lau{What if the stack starts at address 0? Changing the stack convention to always point at the first free cell will get rid of this problem (because our memory us uncapped.)}\lau{Seeing as we did not change this: It does not give a problem as such because this would cause the machine to fail which is considered acceptable. It would also be possible to make a more ``sofisticated'' implementation that special cases on this and simply throws the first address away.}
\end{description}
Note:
\begin{itemize}
\item In a real setting due to a limited number of registers, some of the arguments might be spilled to the stack. It would be possible to do something similar here, but to keep
  matters simple, we opt not to do so.
\item \texttt{reqperm} can be used to test whether something can pass as a stack.
\item \texttt{reqglob} can be used to test whether a callback is admissible in the presence of a stack.
\item The code of a closure will often be found in conjunction with the code that creates it.
\item \texttt{prepstack} as ``prepare stack''. This ensures that the register contains something that looks like a stack and it is prepared for our stack convention.
\end{itemize}
\begin{figure}
  \label{fig:stack-before-call}
  \centering
  \begin{tabular}[!h]{r | >{\raggedright\arraybackslash}p{3cm} |}
    \multicolumn{2}{l}{Stack} \\
    \cline{2-2}
 & \\
 & $\vdots$\\
    \cline{2-2}
 & 0 \\
    \cline{2-2}
    $c_{\var{stk}} \rightarrow$   & local stack\\
 & $\vdots$\\
    \cline{2-2}
  \end{tabular}
  \hspace{1cm}
  \begin{tabular}{r | >{\centering\arraybackslash}p{0.75cm} |}
    \multicolumn{2}{r}{Register file} \\
    \cline{2-2}
    $\pcreg$ & $c_{\pc}$\\
    \cline{2-2}
    $r_0$  & $c_0$ \\
    \cline{2-2}
    $r_{\var{stk}}$  & $c_{\var{stk}}$ \\
    \cline{2-2}
    $r_{\var{args},1}$ & $w_{a,1}$ \\
    \cline{2-2}
             & $\vdots$ \\
    \cline{2-2}
    $r_{\var{args},n}$ & $w_{a,n}$\\
    \cline{2-2}
    $r_{\var{priv},1}$ & $w_{p,1}$\\
    \cline{2-2}
             & $\vdots$ \\
    \cline{2-2}
    $r_{\var{priv},m}$ & $w_{p,m}$\\
    \cline{2-2}
             & $\vdots$ \\
    \cline{2-2}
  \end{tabular}
  \caption{This is the first figure of 6 that illustrates how \texttt{scall} works. In this example, the call \texttt{scall $r([r_{\var{args},1},\dots,r_{\var{args},n}],[r_0,r_{\var{priv},1},\dots,r_{\var{priv},m}])$. In this example the two lists of registers are disjoint even though that does not have to be the case.}}
\end{figure}

\begin{figure}
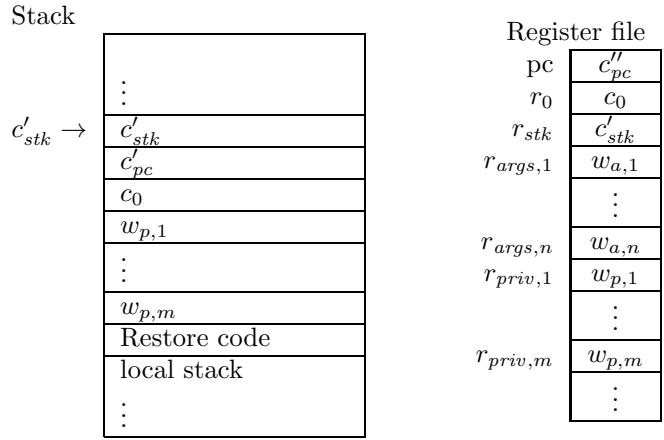

  \label{fig:stack-after-push}
  \centering
  \begin{tabular}[!h]{r | >{\raggedright\arraybackslash}p{3cm} |}
    \multicolumn{2}{l}{Stack} \\
    \cline{2-2}
 & \\
 & $\vdots$\\
    \cline{2-2}
    $c_{\var{stk}}' \rightarrow$  & $c_{\var{stk}}'$ \\
    \cline{2-2}
 & $c_\pc'$ \\
    \cline{2-2}
 & $c_0$ \\
    \cline{2-2}
 & $w_{p,1}$ \\
    \cline{2-2}
 & $\vdots$ \\
    \cline{2-2}
 & $w_{p,m}$ \\
    \cline{2-2}
 & Restore code \\
    \cline{2-2}
 & local stack\\
 & $\vdots$ \\
    \cline{2-2}
  \end{tabular}
  \hspace{1cm}
  \begin{tabular}{r | >{\centering\arraybackslash}p{0.75cm} |}
    \multicolumn{2}{r}{Register file} \\
    \cline{2-2}
    $\pcreg$ & $c_{\pc}''$\\
    \cline{2-2}
    $r_0$  & $c_0$ \\
    \cline{2-2}
    $r_{\var{stk}}$  & $c_{\var{stk}}'$ \\
    \cline{2-2}
    $r_{\var{args},1}$ & $w_{a,1}$ \\
    \cline{2-2}
             & $\vdots$ \\
    \cline{2-2}
    $r_{\var{args},n}$ & $w_{a,n}$\\
    \cline{2-2}
    $r_{\var{priv},1}$ & $w_{p,1}$\\
    \cline{2-2}
             & $\vdots$ \\
    \cline{2-2}
    $r_{\var{priv},m}$ & $w_{p,m}$\\
    \cline{2-2}
             & $\vdots$ \\
    \cline{2-2}
  \end{tabular}
  \caption{Stack and register-file after the restore code, ``private'' registers (remember $r_0$ is here private.), return address ($c_\pc'$), and stack capability ($c_{\var{stk}}'$) have been pushed to the stack.}
\end{figure}

\begin{figure}
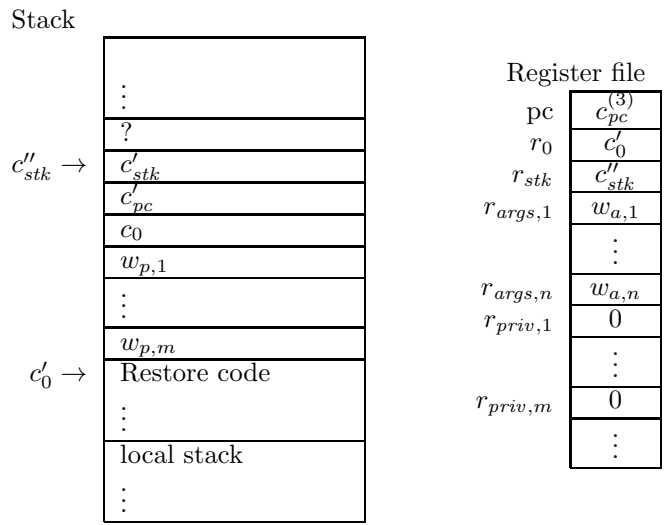

  \label{fig:stack-after-restrict-and-zero}
  \centering
  \begin{tabular}[!h]{r | >{\raggedright\arraybackslash}p{3cm} |}
    \multicolumn{2}{l}{Stack} \\
    \cline{2-2}
 & \\
 & $\vdots$\\
    \cline{2-2}
 & $?$\\
    \cline{2-2}
    $c_{\var{stk}}'' \rightarrow$  & $c_{\var{stk}}'$ \\
    \cline{2-2}
 & $c_\pc'$ \\
    \cline{2-2}
 & $c_0$ \\
    \cline{2-2}
 & $w_{p,1}$ \\
    \cline{2-2}
 & $\vdots$ \\
    \cline{2-2}
 & $w_{p,m}$ \\
    \cline{2-2}
    $c_0' \rightarrow$   & Restore code \\
 & $\vdots$\\
    \cline{2-2}
 & local stack\\
 & $\vdots$\\
    \cline{2-2}
  \end{tabular}
  \hspace{1cm}
  \begin{tabular}{r | >{\centering\arraybackslash}p{0.75cm} |}
    \multicolumn{2}{r}{Register file} \\
    \cline{2-2}
    $\pcreg$ & $c_{\pc}^{(3)}$\\
    \cline{2-2}
    $r_0$  & $c_0'$ \\
    \cline{2-2}
    $r_{\var{stk}}$  & $c_{\var{stk}}''$ \\
    \cline{2-2}
    $r_{\var{args},1}$ & $w_{a,1}$ \\
    \cline{2-2}
             & $\vdots$ \\
    \cline{2-2}
    $r_{\var{args},n}$ & $w_{a,n}$\\
    \cline{2-2}
    $r_{\var{priv},1}$ & 0\\
    \cline{2-2}
             & $\vdots$ \\
    \cline{2-2}
    $r_{\var{priv},m}$ & 0 \\
    \cline{2-2}
             & $\vdots$ \\
    \cline{2-2}
  \end{tabular}
  \caption{ Stack and register-file after the $c_{\var{stk}}'$ has been limited to only give authority over the empty part of the stack (the new capability is $c_{\var{stk}}''$). The empty part of the stack has been cleared. $c_0'$ is made from $c_{\var{stk}}'$ by setting it to point to the restore code and restricting it to a local enter-capability. The ``private'' registers have been cleared.}
\end{figure}

\begin{figure}
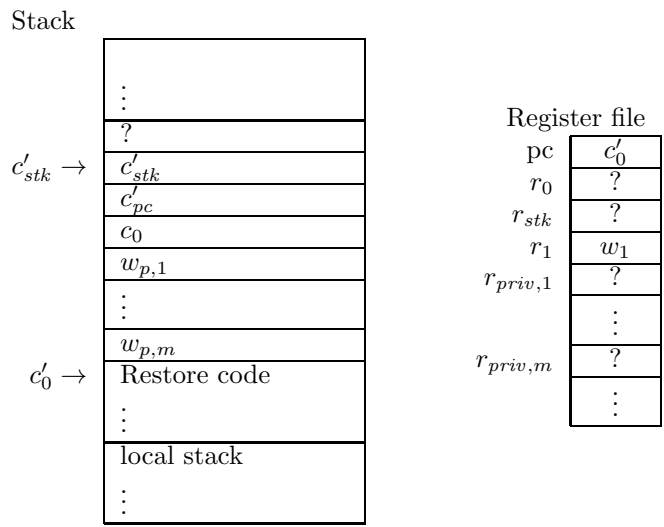

  \label{fig:stack-upon-return}
  \centering
  \begin{tabular}[!h]{r | >{\raggedright\arraybackslash}p{3cm} |}
    \multicolumn{2}{l}{Stack} \\
    \cline{2-2}
 & \\
 & $\vdots$\\
    \cline{2-2} 
 & $?$\\
    \cline{2-2} 
    $c_{\var{stk}}' \rightarrow$  & $c_{\var{stk}}'$ \\
    \cline{2-2}
 & $c_\pc'$ \\
    \cline{2-2}
 & $c_0$ \\
    \cline{2-2}
 & $w_{p,1}$ \\
    \cline{2-2}
 & $\vdots$ \\
    \cline{2-2}
 & $w_{p,m}$ \\
    \cline{2-2}
    $c_0' \rightarrow$   & Restore code \\
 & $\vdots$\\
    \cline{2-2}
 & local stack\\
 & $\vdots$\\
    \cline{2-2}
  \end{tabular}
  \hspace{1cm}
  \begin{tabular}{r | >{\centering\arraybackslash}p{0.75cm} |}
    \multicolumn{2}{r}{Register file} \\
    \cline{2-2}
    $\pcreg$ & $c_0'$\\
    \cline{2-2}
    $r_0$  &  ? \\
    \cline{2-2}
    $r_{\var{stk}}$  & ? \\
    \cline{2-2}
    $r_1$ & $w_1$ \\
    \cline{2-2}
    $r_{\var{priv},1}$ & ?\\
    \cline{2-2}
             & $\vdots$ \\
    \cline{2-2}
    $r_{\var{priv},m}$ & ? \\
    \cline{2-2}
             & $\vdots$ \\
    \cline{2-2}
  \end{tabular}
  \caption{ Stack and register-file upon return from $f$. At this point we have no idea what is in the register-file apart from the $\pcreg$ which we know points to the restore code. The contents of the stack we released access to is also unknown. (Notice that we have changed the order of the registers as we are no longer interested in the argument registers. By convention we expect a return value to be in $r_1$, which is why we have named that word, but the words in the remaining non-special-purpose registers could also be considered return values.)}
\end{figure}

\begin{figure}
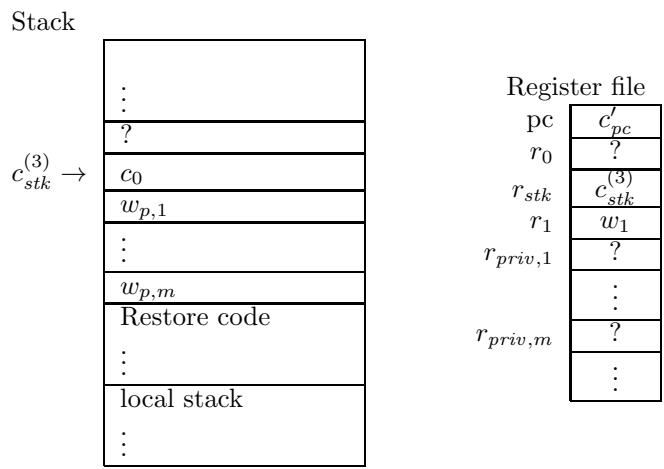

  \label{fig:stack-after-restore-code}
  \centering
  \begin{tabular}[!h]{r | >{\raggedright\arraybackslash}p{3cm} |}
    \multicolumn{2}{l}{Stack} \\
    \cline{2-2}
 & \\
 & $\vdots$\\
    \cline{2-2}
 & ? \\
    \cline{2-2}
    $c_{\var{stk}}^{(3)} \rightarrow$  & $c_0$ \\
    \cline{2-2}
 & $w_{p,1}$ \\
    \cline{2-2}
 & $\vdots$ \\
    \cline{2-2}
 & $w_{p,m}$ \\
    \cline{2-2}
 & Restore code \\
 & $\vdots$\\
    \cline{2-2}
 & local stack\\
 & $\vdots$\\
    \cline{2-2}
  \end{tabular}
  \hspace{1cm}
  \begin{tabular}{r | >{\centering\arraybackslash}p{0.75cm} |}
    \multicolumn{2}{r}{Register file} \\
    \cline{2-2}
    $\pcreg$ & $c_\pc'$\\
    \cline{2-2}
    $r_0$  &  ? \\
    \cline{2-2}
    $r_{\var{stk}}$  & $c_{\var{stk}}^{(3)}$ \\
    \cline{2-2}
    $r_1$ & $w_1$ \\
    \cline{2-2}
    $r_{\var{priv},1}$ & ?\\
    \cline{2-2}
             & $\vdots$ \\
    \cline{2-2}
    $r_{\var{priv},m}$ & ? \\
    \cline{2-2}
             & $\vdots$ \\
    \cline{2-2}
  \end{tabular}
  \caption{ Stack and register-file after executing the restore code. The old stack capability has been restored and the  $\pcreg$-register now points to the return address in memory. }
\end{figure}

\begin{figure}
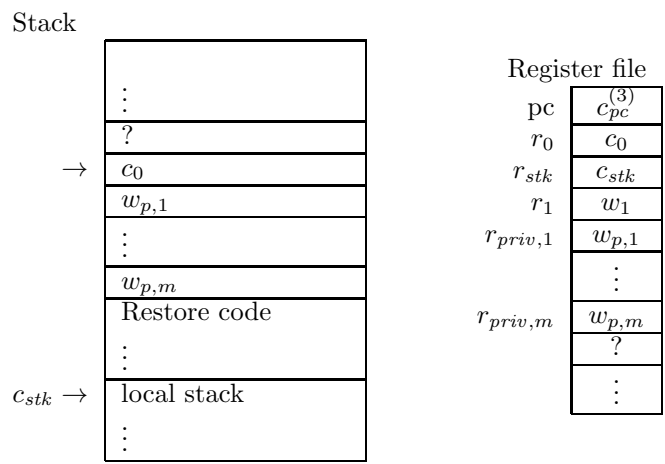

  \label{fig:stack-after-restore-code}
  \centering
  \begin{tabular}[!h]{r | >{\raggedright\arraybackslash}p{3cm} |}
    \multicolumn{2}{l}{Stack} \\
    \cline{2-2}
 & \\
 & $\vdots$\\
    \cline{2-2}
 & ? \\
    \cline{2-2}
    $ \rightarrow$  & $c_0$ \\
    \cline{2-2}
 & $w_{p,1}$ \\
    \cline{2-2}
 & $\vdots$ \\
    \cline{2-2}
 & $w_{p,m}$ \\
    \cline{2-2}
 & Restore code \\
 & $\vdots$\\
    \cline{2-2}
    $c_\stk \rightarrow$  & local stack\\
 & $\vdots$\\
    \cline{2-2}
  \end{tabular}
  \hspace{1cm}
  \begin{tabular}{r | >{\centering\arraybackslash}p{0.75cm} |}
    \multicolumn{2}{r}{Register file} \\
    \cline{2-2}
    $\pcreg$ & $c_\pc^{(3)}$\\
    \cline{2-2}
    $r_0$  &  $c_0$ \\
    \cline{2-2}
    $r_{\var{stk}}$  & $c_{\var{stk}}$ \\
    \cline{2-2}
    $r_1$ & $w_1$ \\
    \cline{2-2}
    $r_{\var{priv},1}$ & $w_{p,1}$\\
    \cline{2-2}
             & $\vdots$ \\
    \cline{2-2}
    $r_{\var{priv},m}$ & $w_{p,m}$ \\
    \cline{2-2}
             & $?$ \\
    \cline{2-2}
             & $\vdots$ \\
    \cline{2-2}
  \end{tabular}
  \caption{ Stack and register-register file after the clean up code has been run. The ``private'' words have been popped to their respective registers. The restore code has been popped off the stack. }
\end{figure}

\subsection{Labels}
\texttt{l:} is a meta level label that can be used to refer to a specific address. When placed on the line of a macro, it refers to the first instruction of this macro.
\clearpage
\section{Examples}
\label{sec:examples}


\subsection{Encapsulation of Local State}
Assembly program not using stack. Assume that $\mathtt{r_l} \not\in \{\pcreg,r_0 \}$ is a register.
\begin{verbatim}
f1: malloc r_l 1
    store r_l 1
    fetch r_adv adv
    call r_adv([],[r_l])
    assert r_l 1
1f: halt
\end{verbatim}
For \texttt{f1} to work, its local state needs to be encapsulated. 
\begin{lemma}[Correctness lemma for \texttt{f1}] \forcenewline
  \label{lem:correctness-f1}
  For all $n \in \nats$
  let
  \begin{align*}
    c_{\var{adv}} & \defeq ((\entry,\glob),\start_{\adv},\addrend_{\adv},\start_{\adv}+\olf) \\
    c_{f1} & \defeq ((\rwx,\glob),\mathtt{f1}-\olf,\mathtt{1f},\mathtt{f1}) \\
    c_\malloc & \defeq ((\entry,\glob),\start_\malloc,\addrend_\malloc,\start_\malloc+\olf) \\
    m & \defeq \hs_{f1} \uplus 
        \hs_\flag \uplus                
        \ms_{\var{link}} \uplus 
        \hs_\adv \uplus 
        \ms_{\malloc} \uplus 
        \hs_{\var{frame}} 
  \end{align*}
  and
  \begin{itemize}
  \item $c_\malloc$ satisfies the specification for malloc and $\iota_{\malloc,0}$ is the region from the specification.
  \end{itemize}
  where 
  \begin{align*}
    &\dom(\hs_{f1}) = [\mathtt{f1}-\olf,\mathtt{1f}] \\
    &\dom(\hs_\flag) = [\flag,\flag] \\
    &\dom(\ms_\link) = [\link,\link+1]\\
    &\dom(\hs_{\adv}) = [\start_\adv,\addrend_\adv] \\
    &\heapSat[\hs_{\malloc}]{n}{[0 \mapsto \iota_{\malloc,0}]}
  \end{align*}
  and
  \begin{itemize}
  \item $\ms_{f1}(\mathtt{f1}-\olf) = ((\readonly,\glob),\link,\link+1,\link)$, $\ms_{f1}(\mathtt{f1}-\olf+1) = ((\readwrite,\glob),\flag,\flag,\flag)$, the rest of $\hs_{f1}$ contains the code of $f1$.
  \item $\ms_\flag = [\flag \mapsto 0]$
  \item $\ms_{\var{link}} = [\var{link} \mapsto c_\malloc, \var{link} + 1 \mapsto c_\adv]$
  \item $\hs_\adv$ contains a global read-only capability for $\hs_\link$ on its first address. The remaining cells of the memory segment only contain instructions.
  \end{itemize}
  if 
  \[
    (\reg\update{\pcreg}{c_{f1}},m) \step[n] (\halted,m'),
  \]
  then
  \[
    m'(\flag) = 0
  \]  
\end{lemma}
\begin{proof}[Proof of Lemma~\ref{lem:correctness-f1}]
  Let $n$ be given and assume the premises in the lemma.  Consider the following part of the execution:
  \[
    (\reg\update{\pcreg}{c_{f1}},m) \step[i] (\reg_0\update{\pcreg}{c_\malloc}\update{r_0}{c_{f1}'}\update{r_1}{1},m)
  \]
  Where $c_{f1}'$ is the return address. Use the malloc specification with
  \begin{align*}
    \iota_\malloc &= \iota_{\malloc,0} \\
    \ms_{\var{footprint}} & = \ms_\malloc \\
    \memreg(r_1) & = \size = 1
  \end{align*}
  to get 
  \[
    (\reg_0\update{\pcreg}{c_\malloc}\update{r_0}{c_{f1}'}\update{r_1}{1},m) \step[j] (\reg_0\update{\pcreg}{c_{f1}'}\update{r_0}{c_{f1}'}\update{r_1}{c_l},m')
  \]
  for some $j$ where for some $\iota_{\malloc}' \futurewk \iota_{\malloc,0}$
  \begin{enumerate}
  \item $m' = \hs_{f1} \uplus 
    \hs_\flag \uplus                
    \ms_{\var{link}} \uplus 
    \hs_\adv \uplus 
    \ms_l \uplus
    \ms_{\malloc}' \uplus 
    \hs_{\var{frame}} $
  \item $\heapSat[\ms_{\malloc}']{n-j}{[0 \mapsto \iota_{\malloc}]}$ \label{f1:mallocsat}
  \item $\dom(\hs_l) = [l,l]$
  \item $c_l = ((\rwx,\glob),l,l,l)$ \label{test}
  \item $\ms_l(l) = 0$
  \end{enumerate}
  Continue the execution to the next malloc hidden in \texttt{call}.
  \[
    (\reg_0\update{\pcreg}{c_{f1}'}\update{r_0}{c_{f1}'}\update{r_1}{c_l},m')
    \step[k]
    (\reg_0\update{\pcreg}{c_\malloc}\update{r_0}{c_{f1}''}\update{r_1}{\var{len}_{\var{ar}}}\update{r_l}{c_l},m'')
  \]
  where
  \begin{enumerate}[resume]
  \item $m'' = m'[l\mapsto 1]$
  \end{enumerate}
  Use the malloc specification notice:
  \begin{itemize}
  \item $\var{len}_{\var{ar}}$ is the needed size for the activation record. 
  \item \ref{f1:mallocsat}.\ and (downwards closure) gives us the needed memory segment satisfaction.
  \item $\ms_{\var{footprint}} =  \ms_\malloc' $
  \end{itemize}
  Get:
  \[
    (\reg_0\update{\pcreg}{c_\malloc}\update{r_0}{c_{f1}''}\update{r_1}{\var{len}_{\var{ar}}}\update{r_l}{c_l},m'')
    \step[j']
    (\reg_0\update{\pcreg}{c_{f1}''}\update{r_0}{c_{f1}''}\update{r_1}{c_{\var{ar}}}\update{r_l}{c_l},m^{(3)})
  \]
  for some $j'$ where for some $[0 \mapsto \iota_{\malloc}'] \futurewk [0 \mapsto \iota_{\malloc}]$
  \begin{enumerate}[resume]
  \item $m'' = \hs_{f1} \uplus 
    \hs_\flag \uplus                
    \ms_{\var{link}} \uplus 
    \hs_\adv \uplus 
    \ms_{l} \uplus
    \ms_{\var{ar}} \uplus
    \ms_{\malloc}'' \uplus 
    \hs_{\var{frame}} $
  \item $\heapSat[\ms_{\malloc}'']{n-j-j'}{[0 \mapsto \iota_{\malloc}']}$ \label{f1:mallocsat}
  \item $\dom(\hs_{\var{ar}}) = [b,e]$, and $e-b = \var{len}_{\var{ar}}$
  \item $c_l = ((\rwx,\glob),b,e,b)$ \label{test}
  \item $\forall a \in [b,e] \ldotp \ms_{\var{ar}}(a) = 0$
  \end{enumerate}
  Continue execution until just after the jump to $\var{adv}$.
  \[
    (\reg_0\update{\pcreg}{c_{f1}''}\update{r_0}{c_{f1}''}\update{r_1}{c_{\var{ar}}}\update{r_l}{c_l},m^{(3)})
    \step[k']
    (\reg_0\update{\pcreg}{\updatePcPerm{c_{\var{adv}}}}\update{r_1}{c_{\var{adv}}}\update{r_0}{c_{ar}'},m^{(3)})
  \]
  for some $k'$ where
  \begin{itemize}
  \item $m^{(3)} = \hs_{f1} \uplus 
    \hs_\flag \uplus                
    \ms_{\var{link}} \uplus 
    \hs_\adv \uplus 
    \ms_{l} \uplus
    \ms_{\var{ar}}' \uplus
    \ms_{\malloc}'' \uplus 
    \hs_{\var{frame}} $
  \item $\ms_{\var{ar}}'$ contains the activation record, i.e., $c_l$, $c_{f1}^{(3)}$ (the return address in f1), and activation code.
  \item $c_{\var{ar}}' = ((\entry,\local)b,e,b+\var{offset})$ where $b+\var{offset}$ is the first address of the activation code.
  \end{itemize}
  Define
  \begin{itemize}
  \item $W = [0 \mapsto \iota_\malloc']
    [1 \mapsto \iota^{\nwl,p}_{\start_\adv,\addrend_\adv}]
    [2 \mapsto \iota^\sta (\perm,\ms_{f1} \uplus \ms_{\var{ar}} \uplus \ms_l \uplus \ms_\flag)]
    [3 \mapsto \iota^{\sta,u}(\perm,\ms_\link)]$
  \end{itemize}
  define
  \begin{enumerate}
  \item $\ms = \hs_{f1} \uplus 
    \hs_\flag \uplus                
    \ms_{\var{link}} \uplus 
    \hs_\adv \uplus 
    \ms_{l} \uplus
    \ms_{\var{ar}}' \uplus
    \ms_{\malloc}'' $ \label{f1:ms-disj-union}
  \end{enumerate}
  Use the FTLR on $\updatePcPerm{c_{\var{adv}}}$ using world $W$, so show
  \begin{itemize}
  \item $\npair{(\start_{\adv},\addrend_{\adv})} \in \readCond{}(\glob)(W)$
    \begin{itemize}
    \item Show: $\iota^{\nwl,p}_{\start_\adv,\addrend_\adv} \nsubsim[n] \iota^\pwl_{\start_\adv,\addrend_\adv}$: Follows from Lemma~\ref{lem:nwlp-subset-pwl}.
    \end{itemize}
  \end{itemize}
  Have
  \begin{enumerate}[resume]
  \item $\npair{\updatePcPerm{c_{\var{adv}}}} \in \stder(W)$
  \end{enumerate}
  Let $n' = n - j - j'-k-k'$ and show
  \begin{enumproof}[resume]
  \item $\heapSat[\ms]{n'}{W}$
    \begin{enumproof}
    \item Split the memory into the disjoint unions of \ref{f1:ms-disj-union} and show:
      \begin{enumproof}
      \item case: $\npair[n']{\ms_\malloc} \in \iota_\malloc'.H (\iota_\malloc'.s) (W)$ 
        \begin{enumproof}
        \item Use $\heapSat[\ms_\malloc]{n'}{[0 \mapsto \iota_\malloc']}$ with malloc specification context independence property.
        \end{enumproof}
      \item case: $\npair[n']{\ms_\adv} \in H^\nwl_{\start_\adv,\addrend_\adv} 1 W$ \label{lem:f1-adv-mem-sat}
        \begin{enumproof}
        \item Show $\forall a \in [\start_\adv,\addrend_\adv] \ldotp (\npair[n'-1]{\ms(a)} \in \stdvr(W) \land \nonlocal{\ms(a)})$
        \end{enumproof}
        \begin{enumproof}
        \item $a \neq \start_\adv$ : trivial, contains instruction only and they are non-local.
        \item $a = \start_\adv$: show $((\readonly,\glob),\link,\link+1,\link) \in \stdvr(W)$\\
          $\glob$ capabilities are non-local.\\
          SFTS $\iota^{\sta,u}(\perm,\ms_\link) \nsubsim[n'] \iota^\pwl_{\link,\link+1}$ which follows from Lemma~\ref{lem:stau-subset-pwl}.
        \end{enumproof}
      \item $\npair[n']{\ms_\link} \in H^{\sta,u}(1)(W)$:\\
        This boils down to showing:
        \begin{enumproof}
        \item $\npair[n'-1]{c_\malloc} \in \stdvr(W)$: Follows from Lemma~\ref{lem:malloc-in-vr}.

        \item $\npair[n'-1]{c_\adv} \in \stdvr(W)$: for $n'' < n'-1$ and $W'
          \futurestr W$ show: \\$\npair[n'']{\updatePcPerm{c_\adv}} \in
          \stder(W')$. Follows from Lemma~\ref{lem:safe-values-safe-invoke},
          together with Lemma~\ref{lem:stdvr-glob-priv-mono} and the fact that
          $c_\adv$ is non-local.  \label{f1:adv}
        \end{enumproof}

      \item The last case follows from Lemma~\ref{lem:mem-sat-static}
      \end{enumproof}
    \end{enumproof}
  \item $\npair[n']{\reg_0\update{\pcreg}{\updatePcPerm{c_{\var{adv}}}}\update{r_1}{c_{\var{adv}}}\update{r_0}{c_{ar}'}} \in \stdrr(W)$
    \begin{enumproof}
    \item case: $\npair[n']{c_{\var{adv}}} \in \stdvr(W)$
      \begin{enumproof}
      \item Similar to \ref{f1:adv}
      \end{enumproof}
    \item case: $\npair[n']{c_{\var{ar}}'} \in \stdvr(W)$.
      \begin{enumproof}
      \item Let $n'' < n'$ and $W' \futurewk W$ be given and show $\npair[n'']{\updatePcPerm{c_{\var{ar}}'}} \in \stder(W')$\\
        Let $n^{(3)} \leq n''$, $\heapSat[\ms']{n^{(3)}}{W'}$, and $\npair[n^{(3)}]{\reg}$ be given\\
        Show: $\npair[n^{(3)}]{(\reg\update{\pcreg}{\updatePcPerm{c_{\var{ar}}'}},\ms')} \in \observations(W')$\\
        Assume $(\reg\update{\pcreg}{\updatePcPerm{c_{\var{ar}}'}},\ms' \uplus \ms_{\var{frame}}) \step[k''] (\halted, m')$, for some $k'' \leq n^{(3)}$, $m'$ and $\ms_{\var{frame}}$. Due to $\heapSat[\ms']{n^{(3)}}{W'}$, $\ms_{f1}$, $\ms_\flag$, $\ms_{\var{ar}}'$, and $\ms_l$ are unchanged. \\
        The execution loads $c_l$ to $r_l$ and jumps to $c_{f1}^{(3)}$ (the point just before the assertion). As $\ms_l = 1$, the assertion is successful and the execution halts. In other words, there were no changes to the memory.\\
        Use $W'$, $\ms_r = \emptyset$, and $\ms'$ to get the desired result, i.e., $m' = \ms' \uplus \ms_{\var{frame}}$ and $\heapSat[\ms']{n^{(3)}-k''}{W'}$ (using downwards closure of memory satisfaction).          
      \end{enumproof}
    \item case: $\npair[n']{0} \in \stdvr(W)$ (the contents remaining registers)\\
      Trivial to show.
    \end{enumproof}
  \end{enumproof}
  Get
  \[
    \npair[n']{(\reg_0\update{\pcreg}{\updatePcPerm{c_{\var{adv}}}}\update{r_1}{c_{\var{adv}}}\update{r_0}{c_{ar}'},m^{(3)})} \in \observations(W)
  \]
  By initial assumption of the lemma, the execution halts. Use $\ms_{\var{frame}}$, $m'$ and the number of steps it takes to halt to get:
  $W' \futurestr W$, $\ms_r$ and $\ms'$ s.t. $m' = \ms_r \uplus \ms' \uplus \ms_{\var{frame}}$ and $\heapSat[\ms']{n}{W'}$. As $\iota_\flag$ is a permanent region, we know it is still in $W'$, so $m'(\flag) = 0$.
\end{proof}

\subsection{Encapsulation of Local State Using Local Capabilities and \texttt{scall}}
\label{subsec:example-loc-cap}

Assembly program using the stack. This program assumes a $r_\stk \not\in \{\pcreg,r_0\}$ register that contains a stack capability (a local $\rwlx$-capability):
\begin{verbatim}
f2: push 1
    fetch r1 adv
    scall r1([],[])
    pop r1
    assert r1 1
2f: halt
\end{verbatim}
              
\begin{lemma}[Correctness lemma for \texttt{f2}]
  \label{lem:correctness-f2}
  let
  \begin{align*}
    c_{\var{adv}} & \defeq ((\entry,\glob),\start_{\adv},\addrend_{\adv},\start_{\adv}+\olf) \\
    c_{f2} & \defeq ((\rwx,\glob),\mathtt{f2}-\olf,\mathtt{2f},\mathtt{f2}) \\
    c_\malloc & \defeq ((\entry,\glob),\start_\malloc,\addrend_\malloc,\start_\malloc+\olf) \\
    c_{\var{stk}} & \defeq ((\rwlx,\local),\start_\stk,\addrend_\stk,\start_\stk-1) \\
    c_\link & \defeq ((\readonly,\glob),\link,\link+1,\link)\\
    \reg & \in \Regs \\
    m & \defeq \hs_{f2} \uplus 
        \hs_\flag \uplus                
        \ms_{\var{link}} \uplus 
        \hs_\adv \uplus 
        \ms_{\malloc} \uplus 
        \ms_{\var{stk}} \uplus
        \ms_{\var{frame}} 
  \end{align*}
  and
  \begin{itemize}
  \item $c_\malloc$ satisfies the specification for malloc and $\iota_{\malloc,0}$ is the region from the specification.
  \end{itemize}
  where 
  \begin{align*}
    &\dom(\hs_{f2}) = [\mathtt{f2}-\olf,\mathtt{2f}] \\
    &\dom(\hs_\flag) = [\flag,\flag] \\
    &\dom(\ms_\link) = [\link,\link+1]\\
    &\dom(\ms_\stk) = [\start_\stk, \addrend_\stk]\\
    &\dom(\hs_{\adv}) = [\start_\adv,\addrend_\adv] \\
    &\heapSat[\hs_{\malloc}]{n}{[0 \mapsto \iota_{\malloc,0}]} \qquad \text{ for all $n \in \nats$}
  \end{align*}
  and
  \begin{itemize}
  \item $\ms_{f2}(\mathtt{f2}-\olf) = ((\readonly,\glob),\link,\link+1,\link)$, $\ms_{f2}(\mathtt{f2}-\olf+1) = ((\readwrite,\glob),\flag,\flag,\flag)$, the rest of $\hs_{f2}$ contains the code of $f2$.
  \item $\ms_\flag = [\flag \mapsto 0]$
  \item $\ms_{\var{link}} = [\var{link} \mapsto c_\malloc, \var{link} + 1 \mapsto c_\adv]$
  \item $\hs_\adv(\start_\adv) = c_\link$ and $\forall \addr \in [\start_\adv+1,\addrend]\ldotp \ms_\adv(a) \in \ints$
  \end{itemize}
  if 
  \[
    (\reg\update{\pcreg}{c_{f2}}\update{r_\stk}{c_\stk},m) \step[n] (\halted,m'),
  \]
  then
  \[
    m'(\flag) = 0
  \]  
\end{lemma}
\begin{proof}[Proof of Lemma~\ref{lem:correctness-f2} (using \texttt{scall} lemma)]
    Let $n$ be given and make the assumptions of the lemma. If we can show
    \begin{equation}
      \label{pf:f2-sc:sfts-obs}
      \npair{(\reg\update{\pcreg}{c_{f2}}\update{r_\stk}{c_\stk},\ms \uplus \ms_\stk)} \in \observations(W)
    \end{equation}
    for 
    \[
      \ms = \hs_{f2} \uplus 
            \hs_\flag \uplus                
            \ms_{\var{link}} \uplus 
            \hs_\adv \uplus 
            \ms_{\malloc} 
    \]
    and
    \[
      W = [0 \mapsto \iota_{\malloc,0}]
          [1 \mapsto \iota^\sta (\perma, \ms_{f2} \uplus \ms_\flag)]
          [2 \mapsto \iota^{\sta,u} (\perma, \ms_\link)]
          [3 \mapsto \iota^{\nwl,p}_{\start_\adv,\addrend_\adv}]
    \]
    then we are done as we by assumption has
    \[
      (\reg\update{\pcreg}{c_{f2}}\update{r_\stk}{c_\stk},m) \step[n] (\halted,m')
    \]
    so \ref{pf:f2-sc:sfts-obs} gives us a $W' \futurestr W$ where $W'$ satisfy part of $m'$. As $\ms_\flag$ is governed by a $\perma$ region, so it is unchanged. In other words
    \[
      m'(\flag) = 0
    \]
    So it suffices to show \ref{pf:f2-sc:sfts-obs}. To this end use Lemma~\ref{lem:anti-red-obs}. Let $\ms_f$ be given, then
    \[
    (\reg\update{\pcreg}{c_{f2}}\update{r_\stk}{c_\stk},\ms \uplus \ms_\stk \uplus \ms_f) \step[k] 
    (\reg',\ms \uplus \ms_\stk' \uplus \ms_f)
  \]
  where
  \begin{itemize}
  \item \lookingat{(\reg',\ms)}{\scall{r_\adv}{}{r_l}}{c_{\var{next}}}
  \item $c_{\var{next}}$ is $c_{f2}$ that points to the instruction after the \texttt{scall}.
  \item \pointstostack{\reg'}{[\start_\stk \mapsto 1]}{\ms_\unused}
    \begin{itemize}
    \item for some $\ms_\unused$ where $\ms_\stk' = [\start_\stk \mapsto 1] \uplus \ms_\unused$.
    \end{itemize}
  \item $\reg'(r_\adv) = c_\adv$
  \end{itemize}
  In order to show the observation part necessary for Lemma~\ref{lem:anti-red-obs}, we use the "\texttt{scall} works"-Lemma (Lemma~\ref{lem:scall-works}). Show the following
  \begin{enumproof}
  \item $\memSat[n-k]{\ms}{W}$\\
    Use Lemma~\ref{lem:disj-mem-sat} with
    \begin{enumproof}
    \item $\memSat[n-k]{\ms_{f2} \uplus \ms_\flag}{[1 \mapsto \iota^\sta (\perma, \ms_{f2} \uplus \ms_\flag)]}$\\
      Lemma~\ref{lem:mem-sat-static}
      \item $\memSat[n-k]{\ms_\adv \uplus \ms_\malloc \uplus \ms_\link}{W_{\var{part}}}$ \\
        where
        \[
          W_{\var{part}} = [0 \mapsto \iota_{\malloc,0}][2 \mapsto \iota^{\sta,u} (\perma, \ms_\link)][3 \mapsto \iota^{\nwl,p}_{\start_\adv,\addrend_\adv}]
        \]
        This amounts to
        \begin{enumproof}
        \item $\npair[n-k-1]{\ms_\malloc} \in H \; 1 \; W_{\var{part}}$ where $H$ is the interpretaion of the $\iota_{\malloc,0}$ region.\\
          Follows from the malloc specification.
        \item $\npair[n-k-1]{\ms_\adv} \in H^\nwl \; 1 \; W_{\var{part}}$\\
          Can be shown using Lemma~\ref{lem:stau-subset-pwl}.
        \item $\npair[n-k-1]{\ms_\link} \in H^{\sta,u} (\ms_\link) \; 1 \; W_{\var{part}}$\\
          This amounts to showing
          \begin{enumproof}
            \item $\npair[n-k-2]{c_\malloc} \in \stdvr(W_{\var{part}})$
              Follows from Lemma~\ref{lem:malloc-in-vr}.
            \item $\npair[n-k-2]{c_\adv} \in \stdvr(W_{\var{part}})$\\
              Follows from Theorem~\ref{thm:ftlr} using Lemma~\ref{lem:nwlp-subset-pwl}.
          \end{enumproof}
        \end{enumproof}
    \end{enumproof}
  \item Hyp-Callee\\
    Assume
    \begin{itemize}
    \item $\dom(\ms_{\mathit{unused}}) = \dom(\ms_{\mathit{act}} \uplus \ms_{\mathit{unused}}')$,
    \item $W' = \revokeTemp{W}[\iota^{\sta}(\temp,\ms_\stk\uplus\ms_{\mathit{act}}),\iota^{\pwl}(\dom(\ms_{\mathit{unused}}'))]$,
    \item $\memSat[n-k-1]{\ms''}{W'}$
    \item $\reg' \text{ points to stack with $\emptyset$ used and $\ms_{\mathit{unused}}'$ unused}$
    \item $\reg'= \reg_0[\pcreg\mapsto\updatePcPerm{c_\adv},r_0 \mapsto c_{\mathit{ret}}, r_\stk \mapsto c_\stk', r_\adv \mapsto c_\adv]$ 
    \item $\npair[n-k-1]{c_{\mathit{ret}}} \in \stdvr(W')$
    \item $\npair[n-k-1]{c_\stk'} \in \stdvr(W')$
    \end{itemize}
    Show
    \[
      (n-k-1,(\reg',\ms'')) \in \observations(W')
    \]
    By Theorem~\ref{thm:ftlr} we get
    \[
      \npair[n-k-1]{\updatePcPerm{c_\adv}} \in \stder(W')
    \]
    getting the desired result amounts to\footnote{We have memory satisfaction by assumption and the above entails the register-file is in the register-file relation.}
    \begin{enumproof}
      \item $\npair[n-k-1]{c_\adv} \in \stdvr(W)$ \\
        To this end let $n' < n-k-1$ and $W'' \futurestr W'$ be given and show
        \[
          \npair[n']{\updatePcPerm{c_\adv}} \in \stder(W'')
        \]
        Follows from Theorem~\ref{thm:ftlr} and Lemma~\ref{lem:nwlp-subset-pwl}.
    \end{enumproof}
  \item Hyp-Cont\\
    Assume
    \begin{itemize}
    \item $n' \leq n-2$
    \item $W'' \futurewk \revokeTemp{W}$
    \item $\memSat[n']{\ms''}{\revokeTemp{W''}}$ 
    \item $\reg''(\pcreg) = c_{\mathit{next}}$
   \item $\reg'' \text{ points to stack with $\ms_\stk$ used and $\ms_{\mathit{unused}}''$ unused}$ for some $\ms_{\mathit{unused}}''$
    \end{itemize}
    and show
    \[
      \npair[n']{(\reg'',\ms'' \uplus [\start_\stk \mapsto 1] \uplus \ms_{\var{unused}}'')} \in \observations(W'')
    \]
    From $\memSat[n']{\ms''}{\revokeTemp{W''}}$, we get that $\ms_{f2}$ is unchanged. Given a frame $\ms_f'$ and assuming $n'$ is sufficiently large, the execution continues as follows:
    \[
      (\reg'',\ms'' \uplus [\start_\stk \mapsto 1] \uplus \ms_{\var{unused}}'' \uplus \ms_f) \step[k] (\halted,\ms'' \uplus [\start_\stk \mapsto 1] \uplus \ms_{\var{unused}}'' \uplus \ms_f)
    \]
    because 1 is popped of the stack to a register, then it is compared with 1 in the assertion, so the assertion succeeds and halts immediately after.

    By assumption we had $\memSat[n']{\ms''}{\revokeTemp{W''}}$ which gives us exactly the memory satisfaction required by $\observations(W'')$.
  \end{enumproof}
\end{proof}

\lau{I have skipped the following program in favour of the one in the next section. }
              ML-like program:
\begin{verbatim}
let f = fun adv =>
          let l = 1 in
          adv(l);
          l := 1;
          adv(0);
          assert(!l == 1)
\end{verbatim}
              In this example \texttt{let l = 1 in} allocates a new local capability \texttt{l} with read-write permissions. Assuming \texttt{adv} has no access to capabilities with permit write local, they cannot store \texttt{l} and thus change its value in the second call.

\subsection{Well-Bracketedness Using Local Capabilities and \texttt{scall}}

\begin{verbatim}
f3: push 1
    fetch r1 adv
    scall r1([],[])
    pop r1
    assert r1 1
    push 2
    fetch r1 adv
    scall r1([],[])
3f: halt
\end{verbatim}
The assertion of $f3$ may seem a bit awkward because it is between two calls. If an adversary could capture the protected return pointer from the first call and save it until the second call, then the adversary could jump to it again. At this point the top of the stack would be 2, so when the execution reaches the assertion, it would fail. However, the produced return pointer is passed as a local capability, so the only place the adversary can store it is on the stack. The adversary loses control of the stack when control is returned to $f3$ where the $\mathtt{scall}$ makes sure to sanitise the stack and register file before control is passed back to the adversary. In other words, the adversary has no way to capture the continuation which makes the above safe and well-bracketed.

\begin{lemma}[Correctness lemma for \texttt{f3}]
  \label{lem:correctness-f3}
  For all $n \in \nats$
  let
  \begin{align*}
    c_{\var{adv}} & \defeq ((\entry,\glob),\start_{\adv},\addrend_{\adv},\start_{\adv}+\olf) \\
    c_{f3} & \defeq ((\rwx,\glob),\mathtt{f3}-\olf,\mathtt{3f},\mathtt{f3}) \\
    c_{\var{stk}} & \defeq ((\rwlx,\local),\start_\stk,\addrend_\stk,\start_\stk-1) \\
    c_\malloc & \defeq ((\entry,\glob),\start_\malloc,\addrend_\malloc,\start_\malloc+\olf) \\
    c_\link & \defeq ((\readonly,\glob),\link,\link+1,\link) \\
    \reg & \in \Regs \\
    m & \defeq \hs_{f3} \uplus 
        \hs_\flag \uplus                
        \ms_{\var{link}} \uplus 
        \hs_\adv \uplus 
        \ms_{\malloc} \uplus 
        \ms_{\var{stk}} \uplus
        \ms_{\var{frame}} 
  \end{align*}
  and
  \begin{itemize}
  \item $c_\malloc$ satisfies the specification for malloc.
  \end{itemize}
  where 
  \begin{align*}
    &\dom(\hs_{f3}) = [\mathtt{f3}-\olf,\mathtt{3f}] \\
    &\dom(\hs_\flag) = [\flag,\flag] \\
    &\dom(\ms_\link) = [\link,\link+1]\\
    &\dom(\ms_\stk) = [\start_\stk, \addrend_\stk]\\
    &\dom(\hs_{\adv}) = [\start_\adv,\addrend_\adv] \\
    &\heapSat[\hs_{\malloc}]{n}{[0 \mapsto \iota_{\malloc,0}]}
  \end{align*}
  and
  \begin{itemize}
  \item $\ms_{f3}(\mathtt{f3}-\olf) = ((\readonly,\glob),\link,\link+1,\link)$, $\ms_{f3}(\mathtt{f3}-\olf+1) = ((\readwrite,\glob),\flag,\flag,\flag)$, the rest of $\hs_{f3}$ contains the code of $f3$.
  \item $\ms_\flag = [\flag \mapsto 0]$
  \item $\ms_{\var{link}} = [\var{link} \mapsto c_\malloc, \var{link} + 1 \mapsto c_\adv]$
  \item $\hs_\adv(\start_\adv) = c_\link$ and all other addresses of $\ms_\adv$ contain instructions.
  \end{itemize}
  if 
  \[
    (\reg\update{\pcreg}{c_{f3}}\update{r_\stk}{c_\stk},m) \step[n] (\halted,m'),
  \]
  then
  \[
    m'(\flag) = 0
  \]  
\end{lemma}
In an attempt to aid the reader, we first provide to high-level descriptions of possible proof of Lemma~\ref{lem:correctness-f3} followed by a more detailed proof.
\begin{proof}[Proof of Lemma~\ref{lem:correctness-f3} (high-level description)]
Executing $f2$ until just after the jump in the first scall brings us to a configuration where the stack contains $1$ followed by some activation code followed by all zeros. The $\pcreg$-register contains an executable adversary capability, register $r_0$ contains a protected return pointer - that is a local enter capability for the execution code, and the $r_\stk$ contains a capability for the cleared part of the stack.

At this point we can define a world with permanent regions
\begin{itemize}
\item fixing the assertion flag, the code of $f2$, and the linking table.
\item the initial malloc region
\item a $\iota^{\nwl,p}$ region
\end{itemize}
and temporary regions
\begin{itemize}
\item a region fixing the private part of the stack
\item a $\iota^\pwl$ region for the rest of the stack 
\end{itemize}
From the FTLR, we get that in any future world of $W$, the adversary capability and its executable counter part is in the expression relation and thus safe to execute in suitable configurations. If the configuration we consider right now is suitable, then the execution produces a memory where the permanent invariants of $W$ are kept which means that the flag is $0$.

To argue that the configuration is suitable, we need to argue that invoking the continuation produces an admissible result. As the continuation is a $\local$ capability, we take a public future world of $W$. In this public world, the private part of the stack remains the same as before the jump, so when we reach the assertion it succeeds and execution continues. At the point of the jump in the second scall, the stack contains $2$ instead of $1$, but otherwise essentially the same. Here we again use that it is safe to execute the adversary and that the continuation in this case halts immediately in a configuration where the assertion flag must be $0$.
\end{proof}

\begin{proof}[Proof of Lemma~\ref{lem:correctness-f3} (high-level description 2)]
\lau{Another go at a proof sketch. More true to the proof below.}
If we can show
\begin{equation}
  (\reg\update{\pcreg}{c_{f3}}\update{r_\stk}{c_\stk},\ms_{\malloc} \uplus \ms' \uplus \ms_{\adv} \uplus \ms_\stk) \in \observations(W),
\end{equation}
for a world $W$ where the assertion flag is permanently 0, then it is still 0 in any configuration the execution halts in. $W$ also needs to require the program and the linking table to permanently remain the same, have a region that governs $\malloc$ and a standard permanent no-write local region for the adversary.

Due to Lemma~\ref{lem:scall-works} the \texttt{scall} lemma, for each \texttt{scall} we have to argue that the adversary and continuation produces results that respect the regions of $W$. Using Lemma~\ref{lem:anti-red-obs} the $\observations$ anti reduction lemma, it suffices to argue that each part of $f3$ between \texttt{scall}s produces admissible results.

Executing until the first \texttt{scall} only pushes 1 to the stack, so the invariants of $W$ are preserved. Due to the \texttt{scall} lemma, we need to argue that that the adversary and the continuation produce admissible results.

Using the FTLR, we get that the executable capability for the adversary is in the $\stder$-relation. As we provide no arguments to the adversary, most of the conditions are satisfied by assumptions and Lemma~\ref{lem:stack-cap-vr}, which makes sure that the stack capability is in the value relation. Which gives us that the adversary produces an admissible result.

With respect to the continuation, it is passed to the adversary as a local capability, so when we reason about it, we consider public future worlds. The \texttt{scall} uses temporary regions for the stack and these persist in public future worlds. This allows us to assume that the private part of the stack still contains 1 after the call. Further, the program, flags, and linking table remain the same in any kind of future world. Therefore, we know that the execution continues by popping 1 from the stack and then asserting that it is indeed 1, which is indeed the case, so 2 is pushed to the stack. At this point we reach another scall. No changes where made to the permanent part of the stack, so the invariants are still satisfied. At this point we use the \texttt{scall} lemma one last time. The adversary call code is well-behaved for the same reasons as in the first call. The \texttt{scall} lemma lets us assume that the continuation continues in a memory that satisfies the invariants of $W$. The execution halts immediately in the continuation, so it produces an admissible result.
\end{proof}

\begin{proof}[Proof of Lemma~\ref{lem:correctness-f3}]
\lau{proof using scall lemma}
Assume the premises of the lemma.
Now define 
  \begin{align*}
    W = \;& [0 \mapsto \iota_{\malloc,0}] \\
          & [1 \mapsto \iota^\sta (\perma,\ms_\flag \uplus \ms_{f2})]\\
          & [2 \mapsto \iota^{\sta,u} (\perma,\ms_\link)]\\
          & [3 \mapsto \iota^{\nwl,p}_{\start_\adv,\addrend_\adv}]\\
  \end{align*}
Further define
\[
\ms' = \ms_\flag \uplus \ms_{f2} \uplus \ms_\link \uplus \ms_\adv
\]
If we can show 
\begin{equation}
\label{eq:f3-obs1}
  \npair[n+1]{(\reg\update{\pcreg}{c_{f3}}\update{r_\stk}{c_\stk},\ms_{\malloc} \uplus \ms' \uplus \ms_{\adv} \uplus \ms_\stk)} \in \observations(W),
\end{equation}
then using $\ms_{\var{frame}}$ as the frame and $m'$ as the resulting memory, we get that $m' = \ms'' \uplus \ms_r \uplus \ms_{\var{frame}}$ for some $\ms'$ and $\ms_r$ s.t. $\memSat[1]{\ms''}{W}$. Region $1$ guarantees that the assertion flag is unchanged, so we have
\[
  m'(\flag) = 0
\]
So SFTS \ref{eq:f3-obs1}. To do so, we use Lemma \ref{lem:anti-red-obs}. Let $\ms_f$ be given. The execution proceeds as follows:
\[
  (\reg\update{\pcreg}{c_{f3}}\update{r_\stk}{c_\stk},\ms' \uplus \ms_\stk \uplus \ms_f) \step[i] (\reg',\ms' \uplus [\start_\stk \mapsto 1] \uplus \ms_\stk|_{\start_\stk+1,\addrend_\stk} \uplus \ms_f),
\]
where
\[
  \lookingat{(\reg',\ms')}{\scall{r}{}{}}{c_{\var{next}}}
\]
where $c_{\var{next}}$ is $c_{f3}$ adjusted to point to the next instruction, namely \texttt{pop r1}. Further we have
\begin{itemize}
\item \pointstostack{\reg'}{[\start_\stk \mapsto 1]}{\ms_\stk|_{\start_\stk+1,\addrend_\stk}}
\end{itemize}
and $i$ is a suitable number of steps.

To show
\[
  \npair[n-i]{(\reg',\ms' \uplus [\start_\stk \mapsto 1] \uplus \ms_\stk |_{\start_\stk+1,\addrend_\stk})} \in \observations(W)
\]
We use Lemma~\ref{lem:scall-works} (we do not use the local frame in the lemma) which requires us to show
\begin{enumproof}
  \item $\memSat[n-i]{\ms'}{W}$ \\
    Partition $\ms'$ as follows:
    \begin{enumproof}
      \item $\ms_\malloc$: governed by $\iota_{\malloc,0}$, use malloc specification.
      \item $\ms_\flag \uplus \ms_{f2}$: governed by region 1, only this memory segment is accepted.
      \item $\ms_\link$: governed by region 2, only this memory segment is accepted. We also need to show that the contents is safe, i.e. shoe
        \begin{enumproof}
          \item $\npair[n-i]{c_\malloc} \in \stdvr(W)$: Follows from Lemma~\ref{lem:malloc-in-vr}.
          \item $\npair[n-i]{c_\adv} \in \stdvr(W)$: \\
            We will show
            \begin{equation}
              \label{eq:f3:adv-in-stder}
              \forall W' \futurestr W \ldotp \npair{c_\adv} \in \stdvr(W')
            \end{equation}
            which will give us what we need using downwards closure as well as a result for later use.

            Let $W' \futurestr W$ be given and show
            \[
              \npair{(\start_\adv,\addrend_\adv,\start_\adv+\olf)} \in \entryCond{}(\glob)(W')
            \]
            to this end let $W'' \futurestr W'$ and $n' < n$ be given and show
            \[
              \npair[n']{\updatePcPerm{c_\adv}} \in \stder(W'')
            \]
            This follows from the FTLR (Theorem~\ref{thm:ftlr}) if we can show
            \[
              \npair[n']{\start_\adv,\addrend_\adv} \in \readCond{}(\glob)(W'')
            \]
            $\iota^{\nwl,p}_{\start_\adv,\addrend_\adv}$ governs the adversary, so the result follows from Lemma~\ref{lem:nwlp-subset-pwl}.
          \end{enumproof}
      \item $\ms_\adv$:
        Follows from Lemma~\ref{lem:stau-subset-pwl}.
    \end{enumproof}
  \item Hyp-Callee \\
    Assume
    \begin{itemize}
    \item $\dom(\ms_\stk |_{\start_\stk+1,\addrend_\stk}) = \dom(\ms_{\mathit{act}} \uplus \ms_{\mathit{unused}}')$
    \item $W' = \revokeTemp{W}[\iota^{\sta}(\temp,[\start_\stk \mapsto 1]\uplus\ms_{\mathit{act}}),\iota^{\pwl}(\dom(\ms_{\mathit{unused}}'))]$
    \item $\memSat[n-i-1]{\ms''}{W'}$
    \item $\reg'' \text{ points to stack with $\emptyset$ used and $\ms_{\mathit{unused}}'$ unused}$
    \item $\reg''= \reg_0[\pcreg\mapsto\updatePcPerm{\reg'(r)},r_0 \mapsto c_{\mathit{ret}}, r_\stk \mapsto c_\stk', r \mapsto \reg'(r)]$ 
    \item $\npair[n-i-1]{c_{\mathit{ret}}} \in \stdvr(W')$
    \item $\npair[n-i-1]{c_\stk'} \in \stdvr(W')$
    \end{itemize}
    for some $\ms_{\mathit{act}}$, $\ms_{\mathit{unused}}$, $\ms''$, $\reg''$, $c_{\mathit{ret}}$.
    
    Using the FTLR, we get $\npair[n-i-1]{\updatePcPerm{c_\adv}} \in \stder(W')$, from
    \begin{enumproof}
      \item $\memSat[n-i-1]{\ms''}{W'}$ : By the above assumptions
      \item $\npair[n-i-1]{\reg''} \in \stdvr(W')$:\\
        show
        \begin{enumproof}
          \item $\npair[n-i-1]{c_{\mathit{ret}}} \in \stdvr(W')$ : by above assumptions.
          \item $\npair[n-i-1]{c_\stk'} \in \stdvr(W')$ : by above assumptions.
          \item $\npair[n-i-1]{c_\adv} \in \stdvr(W')$ : follows from \ref{eq:f3:adv-in-stder}.
          \item The remaining registers we need to consider contain 0 and are thus trivial to show.
        \end{enumproof}
    \end{enumproof}
    we get
    \[
      \npair[n-i-1]{(\ms'',\reg'')} \in \observations(W')
    \]
  \item Hyp-Cont\\
    Assume:
    \begin{itemize}
    \item $n' \leq n-i-2$
    \item $W'' \futurewk \revokeTemp{W}$
    \item $\memSat[n']{\ms''}{\revokeTemp{W''}}$ 
    \item for all $r$, we have that:
      \begin{equation*}
        \reg''(r)
        \begin{cases}
          = c_{\mathit{next}} &\text{ if } r = \pcreg\\
          \in \stdvr(W'') &\text{ if $\reg''(r)$ is a global capability and } r \not\in \{\pcreg, r_\stk\}
        \end{cases}
      \end{equation*}
    \item $\reg'' \text{ points to stack with $[\start_\stk \mapsto 1]$ used and $\ms_{\mathit{unused}}''$ unused}$ for some $\ms_{\mathit{unused}}''$
    \end{itemize}
    and show
    \begin{enumproof}
    \item $(\reg'',\ms'' \uplus [\start_\stk \mapsto 1] \uplus \ms_{\var{unused}}'') \in \observations(\revokeTemp{W''})$\\
      As $W'' \futurestr W$, we know that the program, assertion flag, and linking table remain unchanged in $\ms''$. Given some frame $\ms_f'$, then the execution proceeds by first succeeding the assertion and then pushing 2 to the stack:
      \[
        (\reg'',\ms'' \uplus [\start_\stk \mapsto 1] \uplus \ms_{\var{unused}}'' \uplus \ms_f') \step[k] (\reg^{(3)},\ms'' \uplus [\start_\stk \mapsto 2] \uplus \ms_{\var{unused}}'' \uplus \ms_f')
      \]
      where 
      \begin{itemize}
      \item \lookingat{(\reg^{(3)},\ms'')}{\mathtt{scall}\;r([],[])}{c_{\var{next}}'}
      \item \pointstostack{\reg^{(3)}}{[\start \mapsto 2]}{\ms_{\var{unused}}''}
      \item $\reg^{(3)}(r) = c_\adv$
      \end{itemize}
      By Lemma~\ref{lem:anti-red-obs} it suffices to show
      \begin{enumproof}
      \item $\npair[n'-k]{(\reg^{(3)},\ms'' \uplus [\start_\stk \mapsto 2] \uplus \ms_{\var{unused}}'')} \in \observations(\revokeTemp{W''})$\\
        Show this using Lemma~\ref{lem:scall-works} a. Show:
        \begin{enumproof}
          \item $\memSat[n'-k]{\ms''}{\revokeTemp{W''}}$ is satisfied by one of the first Hyp-cont assumptions and Lemma~\ref{lem:heap-sat-dc}.
          \item Hyp-Callee\\
            Assume:
            \begin{itemize}
            \item $\dom(\ms_{\mathit{unused}}'') = \dom(\ms_{\mathit{act}}' \uplus \ms_{\mathit{unused}}^{(3)})$
            \item $W^{(3)} = \revokeTemp{W''}[\iota^{\sta}(\temp,[\start_\stk \mapsto 2]\uplus\ms_{\mathit{act}}'),\iota^{\pwl}(\dom(\ms_{\mathit{unused}}^{(3)}))]$
            \item $\memSat[n'-k-1]{\ms^{(3)}}{W^{(3)}}$
            \item $\reg^{(4)} \text{ points to stack with $\emptyset$ used and $\ms_{\mathit{unused}}^{(3)}$ unused}$
            \item $\reg^{(4)}= \reg_0[\pcreg\mapsto\updatePcPerm{c_\adv},r_0 \mapsto c_{\mathit{ret}}', r_\stk \mapsto c_\stk'', r \mapsto c_\adv]$ 
            \item $(n'-k-1,c_{\mathit{ret}}') \in \stdvr(W^{(3)})$
            \item $(n'-k-1,c_\stk'') \in \stdvr(W^{(3)})$
            \end{itemize}
            This argument is almost identical to the one we just did for the first call:\\
            Using the FTLR, we get $\npair[n-i-1]{\updatePcPerm{c_\adv}} \in \stder(W^{(3)})$. Which we use with
            \begin{enumproof}
              \item $\memSat[n'-k-1]{\ms^{(3)}}{W^{(3)}}$: By assumption.
              \item $\npair[n'-k-1]{\reg^{(4)}} \in \stdrr(W^{(3)})$: Show:
                \begin{enumproof}
                  \item $\npair[n'-k-1]{c_\adv} \in \stdvr(W^{(3)})$ by Assumption \ref{eq:f3:adv-in-stder}.
                  \item $\npair[n'-k-1]{c_{\mathit{ret}}'}$ by assumption.
                  \item $\npair[n'-k-1]{c_\stk''}$ by assumption
                \end{enumproof}
            \end{enumproof}
            to get
            \[
              \npair[n'-k-1]{(\reg^{(4)},\ms^{(3)})} \in \observations(W^{(3)})
            \]
          \item Hyp-Cont\\
            Assume
            \begin{itemize}
            \item $n'' \leq n'-k-2$
            \item $W^{(3)} \futurewk \revokeTemp{W''}[\iota^{\sta}(\temp,\ms_\stk)][\iota^\sta(\temp,\ms_{\mathit{unused}}^{(3)})]$
            \item $\memSat[n'']{\ms^{(3)}}{\revokeTemp{W^{(3)}}}$ 
            \item for all $r$, we have that:
              \begin{equation*}
                \reg^{(4)}(r)
                \begin{cases}
                  = c_{\mathit{next}}' &\text{ if } r = \pcreg\\
                  \in \stdvr(W'') &\text{ if $\reg^{(4)}(r)$ is a global capability and } r \not\in \{\pcreg, r_\stk\}
                \end{cases}
              \end{equation*}
            \item $\reg' \text{ points to stack with $[\start_\stk \mapsto 2]$ used and $\ms_{\mathit{unused}}^{(3)}$ unused}$ for some $\ms_{\mathit{unused}}^{(3)}$
            \end{itemize}
            and show
            \[
              (n'', (\reg^{(3)},\ms^{(3)} \uplus [\start_\stk \mapsto 2] \uplus \ms_{\var{unused}}^{(3)})) \in \observations(\revokeTemp{W^{(3)}})
            \]
            To this end let $\ms_f''$, $m''$, and $j \leq n''$ be given and assume
            \[
              (\reg^{(3)},\ms^{(3)} \uplus [\start_\stk \mapsto 2] \uplus \ms_{\var{unused}}^{(3)} \uplus \ms_f'') \step[j] (\halted,m'')
            \]
            As the execution halts immediately, 
            \[
              m'' = \ms^{(3)} \uplus [\start_\stk \mapsto 2] \uplus \ms_{\var{unused}}^{(3)} \uplus \ms_f''
            \]
            By assumption we had $\memSat[n'']{\ms^{(3)}}{\revokeTemp{W^{(3)}}}$ and the frame is unchanged, so we can split the memory as needed.
        \end{enumproof}
      \end{enumproof}
    \end{enumproof}
  \end{enumproof}
\end{proof}

\subsection{Inverted Control and Return From Closure}
The following example is constructed to investigate the difficulties of preserving an adversary's local frame. There is no assertion as this is (slightly) beside the point. The lemma we would prove about this should look like Lemma~\ref{lem:correctness-g1}, but it is not state and proven here.
\begin{lstlisting}
g2:  move $r_3$ $\pcreg$
     lea $r_3$ $\dots$
     crtcls $[]$ $r_3$
     rclear $\RegName \setminus \{\pcreg,r_0,r_1 \}$
2g:  jmp $r_0$
f5:  reqglob $r1$
     prepstack $r_\stk$
     scall $r_1$($[]$,$[r_0,r_{\var{env}}]$)
     mclear $r_\stk$
     rclear $\RegName \setminus \{r_0,\pcreg\}$
5f:  jmp $r_0$
\end{lstlisting}

\subsection{Variant of the ``awkward'' example}
Assembly variant of the ``awkward'' example from \citep[p.~11]{Dreyer:2010:IHS:1863543.1863566} which roughly was:
\begin{verbatim}
g = fun _ => let x = 0 in
               fun f =>
                 x := 0;
                 f();
                 x := 1;
                 f();
                 assert(x == 1)
\end{verbatim}
Our translation of the example:
\begin{lstlisting}
g1:  malloc $r_2$ 1
     store $r_2$ 0
     move $r_3$ $\pcreg$
     lea $r_3$ $\dots$
     crtcls $[(x, r_2)]$ $r_3$
     rclear $\RegName \setminus \{\pcreg,r_0,r_1 \}$
1g:  jmp $r_0$
f4:  reqglob $r_1$
     prepstack $r_\stk$
     store $x$ 0
     scall $r_1$($[]$,$[r_0,r_1,r_{\var{env}}]$)
     store $x$ 1
     scall $r_1$($[]$,$[r_0,r_{\var{env}}]$)
     load $r_1$ x
     assert $r_1$ 1
     mclear $r_\stk$
     rclear $\RegName \setminus \{r_0,\pcreg\}$
4f:  jmp $r_0$
\end{lstlisting}
Where the $\dots$ is the appropriate offset to make the capability point to \texttt{f4}.

\begin{lemma}[Correctness of $g1$]
  \label{lem:correctness-g1}
  For all $n \in \nats$
  let
  \begin{align*}
    c_{\var{adv}} & \defeq ((\rwx,\glob),\start_{\adv},\addrend_{\adv},\start_{\adv}+\olf) \\
    c_{g1} & \defeq ((\entry,\glob),\mathtt{g1}-\olf,\mathtt{4f},\mathtt{g1}) \\
    c_\stk & \defeq ((\rwlx,\local),\start_\stk,\addrend_\stk,\start_\stk-1) \\
    c_\malloc & \defeq ((\entry,\glob),\start_\malloc,\addrend_\malloc,\start_\malloc+\olf) \\
    c_\link & \defeq ((\readonly,\glob),\link,\link,\link) \\
    m & \defeq \hs_{g1} \uplus 
        \ms_\flag \uplus                
        \ms_\link \uplus 
        \ms_\adv \uplus 
        \ms_\malloc \uplus 
        \ms_\stk \uplus
        \ms_{\var{frame}} 
  \end{align*}
  where 
  \begin{itemize}
  \item $c_\malloc$ satisfies the specification for malloc with $\iota_{\malloc,0}$
  \end{itemize}
  \begin{align*}
    &\dom(\hs_{g1}) = [\mathtt{g1}-\olf,\mathtt{4f}] \\
    &\dom(\hs_\flag) = [\flag,\flag] \\
    &\dom(\ms_\link) = [\link,\link]\\
    &\dom(\ms_\stk) = [\start_\stk, \addrend_\stk]\\
    &\dom(\hs_{\adv}) = [\start_\adv,\addrend_\adv] \\
    &\heapSat[\hs_{\malloc}]{n}{[0 \mapsto \iota_{\malloc,0}]}
  \end{align*}
  and
  \begin{itemize}
  \item $\ms_{g1}(\mathtt{g1}-\olf) = ((\readonly,\glob),\link,\link,\link)$, $\ms_{g1}(\mathtt{g1}-\olf+1) = ((\readwrite,\glob),\flag,\flag,\flag)$, the rest of $\hs_{g1}$ contains the code of $g1$ immediately followed by the code of $f4$.
  \item $\ms_\flag = [\flag \mapsto 0]$
  \item $\ms_{\var{link}} = [\link \mapsto c_\malloc]$
  \item $\hs_\adv(\start_\adv) = c_\link$ and all other addresses of $\ms_\adv$ contain instructions.
  \item $\forall a \in \dom(\ms_\stk) \ldotp \ms_\stk(a) = 0$ 
  \end{itemize}
  if 
  \[
    (\reg_0\update{\pcreg}{c_\adv}\update{r_\stk}{c_\stk}\update{r_1}{c_{g1}},m) \step[n] (\halted,m'),
  \]
  then
  \[
    m'(\flag) = 0
  \]  
\end{lemma}
In the proof of Lemma~\ref{lem:correctness-g1}, we will use the following region
\begin{definition}
  \label{def:iotax-region}
  \begin{align*}
    \iota_x   & = (\perm, 0, \phi_\pub, \phi, H_x) \\
    \phi_\pub & = \{(0,1)\}^* \\
    \phi      & = (1,0) \union \phi_\pub \\
    H_x \; s \; \hat{W} & = \{\npair{\ms} \mid \ms(x) = s \land n > 0 \} \union \{\npair[0]{\ms}\}
  \end{align*}
\end{definition}
\begin{lemma}
  \label{lem:iotax-is-a-region}
  Definition~\ref{def:iotax-region} defines a region.
\end{lemma}
\begin{proof}[Proof of Lemma~\ref{lem:iotax-is-a-region}] \forcenewline
  \begin{itemize}
  \item $\phi_\pub$ is defined as the reflexive transitive closure, so it is immediately well formed.
  \item $\phi$ adds a transition to $\phi_\pub$ and is also reflexive and transitive. 
  \item $H_x$ is trivially non-expansive in the state.
  \item $H_x$ does not depend on the $\hat{W}$, so it also becomes trivially non-expansive and (privately) monotone in $\hat{W}$.
  \end{itemize}
\end{proof}

\begin{proof}[Proof of Lemma~\ref{lem:correctness-g1} (using \texttt{scall} lemma)]
  Let $n$ be given and make the assumptions of the lemma. Define
  \begin{align*}
    W = & [0 \mapsto \iota_{\malloc,0}]\\
        & [1 \mapsto \iota^{\sta,u} (\perma,\ms_\link)]\\
        & [2 \mapsto \iota^\pwl_{\start_\stk, \addrend_\stk}]\\
        & [3 \mapsto \iota^{\nwl,p}_{\start_\adv,\addrend_\adv}]\\
        & [4 \mapsto \iota^\sta (\perma,\ms_{g1} \uplus \ms_\flag)]
  \end{align*}
  and
  \[
    \ms = \ms_{g1} \uplus 
          \ms_\flag \uplus                
          \ms_\link \uplus 
          \ms_\adv \uplus 
          \ms_\malloc 
  \]
  If we can show
  \begin{equation}
    \label{eq:f4:start-conf-in-obs}
        \npair{(\reg_0\update{\pcreg}{c_\adv}\update{r_\stk}{c_\stk}\update{r_1}{c_{g1}},\ms \uplus \ms_\stk)} \in \observations(W)
  \end{equation}
  then the termination assumption gives us that part of $m$ satisfies a private future world of $W$. Region 4 is permanent, so
  \[
    m(\flag) = 0
  \]
  So it suffices to show Eq. \ref{eq:f4:start-conf-in-obs}. 
  To this end use the FTLR to show $\npair{c_\adv} \in \stder(W)$, so show
  \begin{enumproof}
    \item $\npair{(\start_\adv,\addrend_\adv)} \in \readCond{}(\glob)(W)$\\
      Simple using region $3$ in $W$ and Lemma~\ref{lem:nwlp-subset-pwl}.
    \item $\npair{(\start_\adv,\addrend_\adv)} \in \writeCond{}(\iota^\nwl,\glob)(W)$\\
      Simple using region $3$ in $W$, using
      Lemma~\ref{lem:iota-nwl-address-stratified}.
  \end{enumproof}
  in conclusion $\npair{c_\adv} \in \stder(W)$. We get Eq. \ref{eq:f4:start-conf-in-obs} if we show \ref{pf:g1-3} and \ref{pf:g1-4}:
  \begin{enumproof}[resume]
    \item $\memSat{\ms \uplus \ms_\stk}{W}$ \label{pf:g1-3}\\
      \begin{enumproof}
        \item $\memSat{\ms_{g1} \uplus \ms_\flag}{[4 \mapsto \iota^\sta (\perma,\ms_{g1} \uplus \ms_\flag)]}$ \label{pf:g1-3-1}\\
          Lemma~\ref{lem:mem-sat-static}.
        \item $\memSat{\ms_\stk}{[2 \mapsto \iota^\pwl_{\start_\stk,\addrend_\stk}]}$ \label{pf:g1-3-2}\\
          Lemma~\ref{lem:mem-sat-data-only-std-regions} and assumption that $\ms_\stk$ is all 0.
        \item $\memSat{\ms_{\malloc} \uplus \ms_\link \uplus \ms_\adv}{[0 \mapsto \iota_{\malloc,0}][1 \mapsto \iota^{\sta,u} (\perma,\ms_\link)][3 \mapsto \iota^{\nwl,p}_{\start_\adv,\addrend_\adv}]}$ \label{pf:g1-3-3}\\
          For convenience define
          \[
            W_{\var{mini}} = [0 \mapsto \iota_{\malloc,0}][1 \mapsto \iota^{\sta,u} (\perma,\ms_\link)][3 \mapsto \iota^{\nwl,p}_{\start_\adv,\addrend_\adv}]
          \]
          Partitioning the memory segment in the components of the disjoint union, the $\malloc$ part follows from assumption $\memSat{\ms_\malloc}{[0 \mapsto \iota_{\malloc,0}]}$ and the $\malloc$ specification.

          The linking table part of memory amounts to showing:
          \[
            \npair[n]{\ms_\link} \in H^{\sta,u}(\ms_\link)(1)(\xi^{-1}(W_{\var{mini}}))
          \]
          which in turn amounts to showing
          \[
            \npair[n-1]{c_\malloc} \in \stdvr(W_{\var{mini}})
          \]
          which follows from Lemma~\ref{lem:malloc-in-vr}.

          Showing
          \[
            \npair[n]{\ms_\adv} \in H^\sta_{\start_\adv,\addrend_\adv}(1)(\xi^{-1}(W_{\var{mini}}))
          \]
          is a bit more involved. It amounts to 
          \[
            \forall \addr \in \dom(\ms_\adv) \ldotp \npair[n-1]{\ms_\adv(\addr)} \in \stdvr(W_{\var{mini}})
          \]
          which in turn is trivial for everything but
          \[
            \npair[n-1]{c_\link} \in \stdvr(W_{\var{mini}})
          \]
          This amounts to showing
          \[
            \npair[n-1]{(\link,\link)} \in \readCond{}(\glob)(W_{\var{mini}})
          \]
          which amounts to
          \[
            \iota^{\sta,u} (\perma,\ms_\link) \nsubsim[n-1] \iota^\pwl_{\link,\link}
          \]
          which follows from Lemma~\ref{lem:stau-subset-pwl}.
      \end{enumproof}
      Using Lemma~\ref{lem:disj-mem-sat} repeatedly with \ref{pf:g1-3-1}, \ref{pf:g1-3-2}, and \ref{pf:g1-3-3} gives the desired memory satisfaction.
    \item $\npair{\reg_0\update{r_\stk}{c_\stk}\update{r_1}{c_{g1}}} \in \stdrr(W)$ \label{pf:g1-4}\\
      This amounts to showing
      \begin{enumproof}
        \item $\npair{c_\stk} \in \stdvr(W)$ \\
          The assumptions on $c_\stk$ and $\ms_\stk$ in the lemma entail
          \begin{itemize}
          \item \pointstostack{\reg_0\update{r_\stk}{c_\stk}\update{r_1}{c_{g1}}}{\emptyset}{\ms_\stk}
          \end{itemize}
          and further there is a $\iota^\pwl$ region for $\ms_\stk$ in $W$, so the result follows from Lemma~\ref{lem:stack-cap-vr}.
        \item $\npair{c_{g1}} \in \stdvr(W)$ \\
          Let $n_1 < n$ and $W_1 \futurestr W$ and show
          \[
            \npair[n_1]{\updatePcPerm{c_{g1}}} \in \stder(W_1)
          \]
          To this end assume $n_2 \leq n_1$, $\memSat[n_2]{\ms_1}{W_1}$, and $\npair[n_2]{\reg_1} \in \stdrr(W_1)$ and show
          \[
            \npair[n_2]{(\reg_1\update{\pcreg}{\updatePcPerm{c_{g1}}},\ms_1)} \in \observations(W_1)
          \]
          Using Lemma~\ref{lem:malloc-works}, Lemma~\ref{lem:crtcls-works}, Lemma~\ref{lem:anti-red-obs} (and some others),
          it suffices to show
          \[
            \npair[n_2']{(\reg_2,\ms_2 \uplus \ms_\malloc' \uplus \ms_\cls \uplus \ms_x)} \in \observations(W_2)
          \]
          where
          \[
            W_2 = W_1\update{0}{\iota_\malloc}\update{i_1}{\iota^\sta (\perm,\ms_\cls)}\update{i_2}{\iota_x}
          \]
          where $i_1,i_2\not\in\dom(W_1)$ and $i_1 \neq i_2$ and $\iota_x$ is the region in Definition~\ref{def:iotax-region} which is a region by Lemma~\ref{lem:iotax-is-a-region}. Also
          \begin{itemize}
          \item $\iota_\malloc \futurestr \iota_\malloc'$
          \item $c_x = ((\rwx,\glob),x,x,x)$
          \item $\ms_x = [x \mapsto 0]$
          \item $\memSat[n_2']{\ms_2 \uplus \ms_\malloc' \uplus \ms_\cls \uplus \ms_x}{W_2}$
          \item $c_\env = ((\rwx,\glob,\env,\env,\env))$
          \item $\ms_\env = [\env \mapsto c_x]$
          \item $c_{f4} = ((\rwx,\glob),\mathtt{g1} - \olf,\mathtt{4f},\mathtt{f4})$
          \item $\ms_\cls = \ms_\env \uplus \ms_\act$
          \item 
            \begin{equation*}
              \reg_2(r) =
              \begin{cases}
                \updatePcPerm{\reg_1(r_0)} & r = \pcreg \\
                \reg_1(r_0) & r = r_0 \\
                c_\cls & r = r_1 \\
                0 & \text{otherwise}
              \end{cases}
            \end{equation*}
          \end{itemize}
          Finally assume Hyp-Act:
            \begin{multline}
            \label{eq:g1-hyp-act}
            \forall \reg,\ms \ldotp \reg(\pcreg)=c_\cls \Rightarrow\\
            \exists j \ldotp \forall \ms_f \ldotp (\reg,\ms \uplus \ms_\cls \uplus \ms_f) \step[j] (\reg\update{\pcreg}{\updatePcPerm{c_{f4}}}\update{r_\env}{c_\env},\ms \uplus \ms_\cls \uplus \ms_f)
          \end{multline}
          Show
          \begin{equation}
            \label{eq:g1-obs-to-show}
            \npair[n_2-i]{(\reg_2,\ms_2 \uplus \ms_\malloc' \uplus \ms_\env \uplus \ms_x \uplus \ms_\cls)} \in \observations(W_2)
          \end{equation}
          If $\reg_1(r_0).\perm \not\in \{\entry, \exec, \rwx, \rwlx\}$, then the execution fails after the jump and \label{eq:g1-obs-to-show} is thus trivially true.

          If $\reg_1(r_0).\perm \in \{\entry, \exec, \rwx, \rwlx\}$, then either $\execCond{}$ or $\entryCond{}$ holds for the capability in $\reg_1(r_0)$. Now use $W_2 \futurewk W_1$ with the appropriate condition to get 
          \[
            \npair[n_2-i]{\updatePcPerm{\reg_1(r_0)}} \in \stder(W_2)
          \]
          which in turn gives us \ref{eq:g1-obs-to-show} if we can show the following
          \begin{enumproof}
            \item $\memSat[n_2-i]{\ms_2 \uplus \ms_\malloc' \uplus \ms_\env \uplus \ms_x \uplus \ms_\cls}{W_2}$\\
              We first show the following:
              \begin{itemize}
              \item $\memSat[n_2-i]{\ms_2 \uplus \ms_\malloc'}{W_1\update{0}{\iota_\malloc}}$: we already know this.
              \item $\memSat[n_2-i]{\ms_\env \uplus \ms_\cls}{[i_1 \mapsto \iota^\sta(\perma, \ms_\env \uplus \ms_\cls)]}$: By Lemma~\ref{lem:mem-sat-static}.
              \item $\memSat[n_2-i]{\ms_x}{i_2 \mapsto \iota_x}$: $\ms(x) = 0$, so okay.
              \end{itemize}

            \item $\npair[n_2-i]{\reg_2} \in \stdrr(W_2)$ \\
              Amounts to showing
              \begin{enumproof}
                \item $\npair[n_2-i]{\reg_2(r_0)} \in \stdvr(W_2)$ by assumption $\npair[n_2]{\reg_1} \in \stdrr(W_1)$ and $\stdvr$ monotonicity wrt. $\futurewk$
                \item $\npair[n_2-i]{c_\cls} \in \stdvr(W_2)$ \\
                  Let $n_3 < n_2 - i$ and $W_3 \futurestr W_2$ be given and show
                  \[
                    \npair[n_3]{\updatePcPerm{c_\cls}} \in \stder(W_3)
                  \]
                  To this and let $n_4 \leq n_3$, $\memSat[n_4]{\ms_3}{W_3}$, and $\npair[n_4]{\reg_3} \in \stdrr(W_3)$ and show
                  \begin{equation}
                    \label{eq:g1-c-cls-obs}
                    \npair[n_4]{(reg_3\update{\pcreg}{\updatePcPerm{c_\cls}},\ms_3)} \in \observations(W_3)
                  \end{equation}
                  Let $\ms_3^p$ and $\ms_3^t$ be memory segments such that $\ms_3 = \ms_3^p \uplus \ms_3^t$ and $\memSat[n_4]{\ms_3^p}{\revokeTemp{W_3}}$ (using Lemma~\ref{lem:priv-mono-like2}). By $\memSat[n_4]{\ms_3}{W_3}$ and $W_3 \futurestr W_2$, we know $\ms_\cls \subseteq \ms_3^p$, so using Hyp-Act(\ref{eq:g1-hyp-act}), we get $j$ such that

                  \begin{multline}
                    \forall \ms_f \ldotp (reg_3\update{\pcreg}{\updatePcPerm{c_\cls}},\ms_3^p \uplus \ms_3^t \uplus \ms_f) \step[j]\\ (\reg_3\update{\pcreg}{\updatePcPerm{c_\cls}}\update{r_\env}{c_\env},\ms_3^p \uplus \ms_3^t \uplus \ms_f)
                  \end{multline}
                  Using Lemma~\ref{lem:anti-red-obs} it suffices to show
                  \[
                    \npair[n_4]{(\reg_3\update{\pcreg}{\updatePcPerm{c_\cls}}\update{r_\env}{c_\env},\ms_3^p \uplus \ms_3^t)} \in \observations(W_3)
                  \]
                  Use Lemma~\ref{lem:anti-red-obs} again. This time let $\ms_f''$ be given and take $\ms_r$ to be the part of $\ms_3^t$ that $\reg_3(r_\stk)$ does not govern. By the operational semantics, we know\footnote{the execution may fail, but then the configuration is trivially in the observation relation.}
                  \[
                    (\reg_3\update{\pcreg}{\updatePcPerm{c_\cls}}\update{r_\env}{c_\env},\ms_3^p \uplus \ms_3^t \uplus \ms_f'') \step[j']
                    (\reg_4,\ms_4 \uplus \ms_3^t \uplus \ms_f'')
                  \]
                  where
                  \begin{itemize}
                  \item \lookingat{(reg_4,\ms_4)}{\scall{r_1}{}{r_0,r_1,r_\env}}{c_{\var{next}}}
                    \begin{itemize}
                    \item $c_{\var{next}}$ is the capability pointing to the next instruction.
                    \end{itemize}
                  \item \pointstostack{\reg_4}{\emptyset}{\ms_{\var{unused}}}
                    \begin{itemize}
                    \item \texttt{prepstack} did not fail, so the stack capability must be $\rwlx$ and follow the stack convention.
                    \end{itemize}
                  \item $\reg_4(r_1)$ is a $\glob$ capability.
                    \begin{itemize}
                    \item \texttt{reqglob} did not fail
                    \end{itemize}
                  \item $\ms_4(x) = 0$
                  \item $\reg_4(r_\env) = c_\env$
                  \end{itemize}
                  region $i_2$ (the $\iota_x$ region) can be in either state $0$ or $1$, so to make sure it is in state $0$, we use a private transition. So let $W_4$ be $\revokeTemp{W_3}$ with region $i_2$ in state $0$. We then have
                  \[
                    \memSat[n_4-j-j']{\ms_4}{W_4}
                  \]
                  Now we can use Lemma~\ref{lem:scall-works} to show:
                  \[
                    \npair[n_4-j-j']{(\reg_4,\ms_4 \uplus \ms_r \uplus \emptyset \uplus \ms_\unused)} \in \observations(W_4)
                  \]
                  where $\ms_r$ is the local frame of the \texttt{scall} lemma.
                  \begin{enumproof}
                    \item $\memSat[n_4-j-j']{\ms_4}{\revokeTemp{W_4}}$: follows from $W_4 = \revokeTemp{W_4}$
                    \item Hyp-Callee \label{pf:g1-first-hyp-callee}\\
                      We know $\npair[n_4]{\reg_3(r_1)} \in \stdvr(W_3)$. If this is not a capability that becomes executable when jumped to, then the execution fails, so the register memory segment pair is trivially in the observation relation. If it is executable, then either the $\execCond{}$ or the $\entryCond{}$ holds for appropriate values. We also know that it is a global capability, so we can use it with private future worlds. We have $W_5 = \revokeTemp{W_4}[\iota^\sta (\temp,\emptyset \uplus \ms_\act \uplus \ms_r),\iota^\pwl(\dom(\ms_\unused'))] \futurestr W_3$, for some $\ms_\act$ and $\ms_\unused'$. By the execute/enter condition, we have
                      \[
                        \npair[n_4-j-j']{\updatePcPerm{reg_3(r_1)}} \in \stder(W_5)
                      \]
                      Now it suffices to show
                      \begin{enumproof}
                        \item $\memSat[n_4-j-j'-1]{\ms_5}{W_5}$ for some $\ms_5$ which is one of the assumptions of Hyp-Callee.
                        \item $\npair[n_4-j-j'-1]{\reg_5} \in \stdrr(W_5)$ where $\reg_5$ is as described in the \texttt{scall} lemma Hyp-callee premise.\\
                          Amounts to showing:\lau{too deeply nested for environments}\\
                          1) $\npair[n_4-j-j'-1]{\reg_3(r_1)} \in \stdvr(W_5)$, use Lemma~\ref{lem:stdvr-glob-priv-mono} with $\npair[n_4-j-1]{\reg_3(r_1)} \in \stdvr(W_3)$, the capability is global, and $W_5 \futurestr W_3$.
                          2) The protected return pointer and the stack capability are in the value relation by Hyp-callee assumptions.
                      \end{enumproof}
                      which gives us
                      \[
                        \npair[n_4-j-j'-1]{(\reg_5,\ms_5)} \in \observations(W_5)
                      \]
                    \item Hyp-Cont\\
                      Assume
                      \begin{itemize}
                      \item $n_5 \leq n_4-j-j'-2$
                      \item $W_6 \futurewk \revokeTemp{W_4}$
                      \item $\memSat[n_5]{\ms_6}{\revokeTemp{W_6}}$
                      \item $\reg_6(\pcreg) = c_{\var{next}}$, $\reg_6(r_0) = \reg_3(r_0)$, $\reg_6(r_1) = \reg_3(r_1)$, $\reg(r_\env) = c_\env$
                      \item \pointstostack{\reg_6}{\emptyset}{\ms_\unused''}
                      \end{itemize}
                      Show 
                      \[
                        \npair[n_5]{(\reg_6,\ms_6 \uplus \ms_r \uplus \emptyset \uplus \ms_\unused'')} \in \observations(W_6)
                      \]
                      Use the $\observations$-anti-reduction lemma (Lemma~\ref{lem:anti-red-obs}) followed by the \texttt{scall} lemma (Lemma~\ref{lem:scall-works}). Given $\ms_f'''$, we know by the operational semantics and the fact that the program hasn't changed that
                      \[
                        (\reg_6,\ms_6 \uplus \ms_r \uplus \ms_\unused'' \uplus \ms_f''') \step[k] (\reg_7,\ms_6\update{x}{1} \uplus \ms_r \uplus \ms_\unused'' \uplus \ms_f''')
                      \]
                      where
                      \begin{itemize}
                      \item \lookingat{(\reg_7,\ms_6\update{x}{1})}{\mathtt{scall}\;r([],[r_0,r_\env])}{c_{\var{next}}'}\\
                        $c_{\var{next}}'$ is the current $\pcreg$ capability but looking at \texttt{load} $r_1$ $x$.
                      \item $\reg_7(r_0,r_1,r_\env,r_\stk) = \reg_6(r_0,r_1,r_\env,r_\stk)$
                      \end{itemize}
                      In $\revokeTemp{W_6}$, we don't know which state the $\iota_x$ region is in, but state $1$ is reachable via a public transition, so let $W_7$ be $\revokeTemp{W_6}$ with region $i_2$ in state $1$. It follows easily that
                      \[
                        \memSat[n_5-k]{\ms_6\update{x}{1}}{W_7}
                      \]
                      We continue the proof in item \ref{item:proof-cont}
                  \end{enumproof}
              \end{enumproof}
          \end{enumproof}
      \end{enumproof}
      \item \label{item:proof-cont} At this point, we apply the \texttt{scall} lemma, to get
        \[
          \npair[n_5-k]{(\reg_7,\ms_6\update{x}{1} \uplus \ms_r \uplus \ms_\unused''')} \in \observations(W_7)
        \]
        show
        \begin{enumproof}
          \item $\memSat[n_5-k]{\ms_6\update{x}{1}}{\revokeTemp{W_7}}$, follows from $W_7= \revokeTemp{W_7}$.
          \item Hyp-Callee: Goes like the first Hyp-Callee (\ref{pf:g1-first-hyp-callee}).
          \item Hyp-Cont\\
            Assume:
            \begin{itemize}
            \item $n_6 \leq n_5-k-2$
            \item $W_8 \futurewk \revokeTemp{W_7}$
            \item $\memSat[n_6]{\ms_7}{\revokeTemp{W_8}}$ 
            \item $\reg_8(r_0,r_\env) = \reg_7(r_0,r_\env)$
            \item $\reg_8(\pcreg) = c_{\var{next}}'$
            \item $\reg_8 \text{ points to stack with $\emptyset$ used and $\ms_{\mathit{unused}}^{(6)}$ unused}$ for some $\ms_{\mathit{unused}}^{(6)}$
            \end{itemize}
            Show:
            \[
              \npair[n_6]{(\reg_8,\ms_7 \uplus \ms_r \uplus \emptyset \uplus \ms_\unused^{(5)})} \in \observations(W_8)
            \]
            Use Lemma~\ref{lem:anti-red-obs}. Let $\ms_f^{(4)}$ be given, then
            \[
              (\reg_8,\ms_7 \uplus \ms_r \uplus \emptyset \uplus \ms_\unused^{(5)} \uplus \ms_f^{(4)}) \step[l] (reg_9,\ms_7 \uplus \ms_r \uplus \emptyset \uplus \ms_0 \uplus \ms_f^{(4)})
            \]
            where
            \begin{itemize}
            \item $\reg_9(\pcreg) = \updatePcPerm{\reg_3(r_0)}$ (note $\reg_8(r_0) = \reg_3(r_0)$)
            \item $\reg_9(r_0) = reg_3(r_0)$
            \item For all $r \not\in \{\pcreg,r_0\}$, $\reg_9(r) = 0$.
            \item $\dom(\ms_0) = \dom(\ms_\unused^{(5)})$ and $\forall a \in \dom(\ms_0) \ldotp \ms_0(a) = 0$
            \end{itemize}
            The execution proceeds as above because $\iota_x$ in $W_8$ is in state $1$, so $\ms_7(x) = 1$ which causes the assertion to succeed. Subsequently the stack and most of the registers are cleared.

            Now take $W_{10}$ to be $W_9$ with all the regions in $\dom(\erase{W_3}{\temp})$ reinstated. Now we show the following:
            \begin{enumproof}
            \item $W_{10} \futurewk W_3$ \label{g1:w10futwkw3}\\
              We have
              \[
                \forall r \in \dom(W_3) \ldotp W_3(r) = W_{10}(r)
              \]
              if the region was permanent in $W_3$, then it is there because $W_{10} \futurestr W_3$. If it was temporary, then it is there because it was just reinstated. If it was revoked in $W_3$, then it is still there because the only reinstated region were the temporary ones in $W_3$.

              All the future worlds we have been given have been public, so the regions can only have made public transitions. In $W_3$ region $\iota_x$ is in state 0 or 1. In $W_{10}$ region $\iota_x$ is in state 1. State 1 can be reached from 0 and 1 using a public transition, so the $\iota_x$ in $W_{10}$ is a public future region of the $\iota_x$ in $W_3$.

              In other words, all the regions in $W_3$ have only taken public transitions compared to the corresponding regions in $W_{10}$.

              The relation between the relevant worlds is sketched out in Figure~\ref{fig:worlds-in-awk-ex}.

            \item $\memSat[n_6-l]{\ms_7 \uplus \ms_r \uplus \emptyset \uplus \ms_0}{W_{10}}$ \label{g1:memsatw10}\\
              First notice that from 
              \begin{itemize}
              \item $\npair[n_4]{\reg_3} \in \stdrr(W_3)$
              \item $\memSat[n_4]{\ms_3}{W_3}$
              \item $\reg(r_\stk).\perm = \rwlx$
              \end{itemize}
              using Lemma~\ref{lem:pwl-stack} we get that there exists a region, $r_{\adv \stk}$ such that
              \[
                W_3(r_{\adv \stk}) \nequal \iota_{\stk_a,\stk_b}^\pwl
              \]
              and $\dom(\ms_\unused) \subseteq [\stk_a,\stk_b]$. Now take $\ms_{\adv \stk} = \ms_r|_{[\stk_a,\stk_b]}$ (notice this not all of $[\stk_a,\stk_b]$ is in the domain of $\ms_r$).
              We know
              \begin{equation}
                \label{eq:g1:ms7heapsat}
                \memSat[n_6]{\ms_7}{\revokeTemp{W_8}}
              \end{equation}
              and
              \begin{equation}
                \label{eq:g1:ms3heapsat}
                \memSat[n_4]{\ms_3}{W_3}
              \end{equation}
              which gives us two partitions say $P_8$ and $P_3$ respectively. Now define the partition $P$ as follows:
              \[
                P(r) =
                \begin{cases}
                  P_8(r) & r \in \dom(\erase{W_8}{\perma}) \\
                  \ms_{\adv \stk} \uplus \ms_0 & r = r_{\adv \stk}\\
                  P_3(r) & \text{otherwise}
                \end{cases}
              \]
              Now let $r \in \active{W}$, $n_7 < n_6 - l$, and $W(r) = (\_,s,\_,\_,H)$ and show
              \[
                \npair[n_7]{P(r)} \in H(s)(\xi^{-1}(W_{10})).
              \]
              Consider the following cases
              \begin{enumproof}
                \item $r \in \dom(\erase{W_8}{\perma})$ \\
                  Use \ref{eq:g1:ms7heapsat}, the fact that $W_{10} \futurestr \revokeTemp{W_8}$ and that permanent regions respect future private world.
                \item $r = r_{\adv \stk}$ \\
                  In this case we know the region is $\iota^\pwl_{\stk_a,\stk_b}$, so we need to show
                  \[
                    \npair[n_7]{\ms_{\adv \stk} \uplus \ms_0} \in H^\pwl_{\stk_a,\stk_b}(1)(\xi^{-1}(W_{10}))
                  \]
                  which amounts to showing
                  \[
                    \forall a \in \dom(\ms_0) \ldotp  \npair[n_7-1]{\ms_0(a)} \in \stdvr(W_{10}),
                  \]
                  which is trivial, and
                  \[
                    \forall a \in \dom(\ms_{\adv \stk}) \ldotp \npair[n_7-1]{\ms_{\adv \stk}(a)} \in \stdvr(W_{10})
                  \]
                  here we use that \ref{eq:g1:ms3heapsat} entails
                  \[
                    \forall a \in \dom(\ms_{\adv \stk}) \ldotp \npair[n_4-1]{\ms_{\adv \stk}(a)} \in \stdvr(W_3)
                  \]
                  and the fact that $\stdvr$ is monotone w.r.t $\futurewk$, $W_{10} \futurewk W_3$, and $\stdvr(W_{10})$ is downwards-closed.
                \item otherwise \\
                  Use \ref{eq:g1:ms3heapsat}, $W_{10} \futurewk W_3$, and the fact that for a temporary region $H(s)$ is monotone w.r.t. $\futurewk$.
              \end{enumproof}
              
            \item $\npair[n_6-l]{\reg_9} \in \stdrr(W_{10})$ \label{g1:reg9stdrr}\\
              Most registers are cleared. The only interesting register is $r_0$, so show:
              \[
                \npair[n_6-l]{\reg_9(r_0)}\in \stdvr(W_{10})
              \]
              This follows from $\reg_9(r_0) = \reg_3(r_0)$, $\npair[n_4]{\reg_3} \in \stdrr(W_3)$, $\stdvr$ monotone w.r.t $\futurewk$, $W_{10} \futurewk W_3$.
            \end{enumproof}
            As we were using Lemma~\ref{lem:anti-red-obs}, we need to show
            \[
              \npair[n_6-l]{(reg_9,\ms_7 \uplus \ms_r \uplus \emptyset \uplus \ms_0)} \in \observations(W_{10})
            \]
            To this end the use $\reg_3(r_0) = \reg_9(r_0)$ and $\npair[n_4]{\reg_3(r_0)} \in \stdvr1(W_3)$. Assuming that $\reg_9(r_0).\perm \in \{\entry,\exec,\rwx,\rwlx \}$ (if this is not the case, then it is trivial to show the above as the execution fails), then either the $\execCond{}$ or the $\entryCond{}$ hold for appropriate values. Now use that $n_6-l < n_4$ and $W_{10} \futurewk W_3$ (\ref{g1:w10futwkw3})\footnote{We don't know whether the capability is local or global, but it does not matter as we have a public future world relation between the two worlds.} to get
            \[
              \npair[n_6-l]{\updatePcPerm{\reg_9(r_0)}} \in \stder(W_{10})
            \]
            now using \ref{g1:memsatw10} and \ref{g1:reg9stdrr}, we get the desired result.
        \end{enumproof}
  \end{enumproof}
\begin{figure}[t]
  \centering
\begin{tikzpicture}[main node/.style={}, node distance=1.5cm]
  \node[main node] (1) {$W_3$};
  \node[main node] (label1) [below of=1,yshift=1cm] {\footnotesize{\textit{\texttt{f4} called}}};
  \node[main node,left] (given) [left of=1,yshift=0.05cm] {\parbox{1.5cm}{\footnotesize{Given:}}};

  \node[main node] (2) [above right of=1]{$W_4$};
  \node[main node] (label2) [above of=2,yshift=-1cm] {\textit{\footnotesize{first callback}}};

  \node[main node,left] (const) [above of=given,yshift=-0.4cm] {\parbox{1.5cm}{\footnotesize{Constructed:}}};

  \node[main node] (3) [below right of=2]{$W_6$};
  \node[main node] (label3) [below of=3,yshift=0.75cm] {\footnotesize{\textit{callback returns}}};

  \node[main node] (4) [above right of=3]{$W_7$};
  \node[main node] (label4) [above of=4,yshift=-0.75cm] {\footnotesize{\textit{second callback}}};

  \node[main node] (5) [below right of=4]{$W_8$};
  \node[main node] (label5) [below of=5,yshift=1cm] {\footnotesize{\textit{callback returns}}};

  \node[main node] (6) [above right of=5]{$W_{10}$};
  \node[main node] (label6) [above of=6,yshift=-1cm] {\footnotesize{\textit{\texttt{f4} returns}}};
  \path(1) edge[draw=none] node [sloped, auto=false, allow upside down] {$\sqsubseteq^\priv$} (2)
       (2) edge[draw=none] node [sloped, auto=false, allow upside down] {$\sqsubseteq^\pub$} (3)
       (3) edge[draw=none] node [sloped, auto=false, allow upside down] {$\sqsubseteq^\priv$} (4)
       (4) edge[draw=none] node [sloped, auto=false, allow upside down] {$\sqsubseteq^\pub$} (5)
       (5) edge[draw=none] node [sloped, auto=false, allow upside down] {$\sqsubseteq^\priv$} (6);
\end{tikzpicture}
  \label{fig:worlds-in-awk-ex}
\end{figure}
\end{proof}

\lau{In order for this to be secure using a stack the closure needs to make a couple of checks: 1) The alleged stack capability has to be local and $\rwlx$ - if not, then fail. 2) The callback \texttt{f} has to be a global capability - if not, then fail. 3) The stack needs to be cleared before the first callback. \\If it looks like a stack, works like a stack, and quacks like a stack, then it is probably a stack. The first requirement makes sure that what was passed as the stack can be used as a stack even though it might not be a stack capability that follows the conventional stack discipline. If the capability is actually for part of the callers stack, then we will just overwrite their part of the stack. We are not guaranteed that the caller did not save the stack capability - they did after all have a stack to save it on. In fact, if the caller wants to reuse our part of the stack later, then they have to keep a capability for our part of the stack around. When the caller is not trusted, then this seems like an unsafe practice. This brings us to the second requirement: If the callback is global, then there is no way it can restore the previous stack capability and gain access to our part of the stack. This is the case as the \emph{only} place they could have stored the stack capability is on the stack, and the only way they can restore this is through one of the local stack capabilities. In other words, this requirement makes sure that the callback capability is not derived from the stack pointer (which would allow them to get a stack capability for our part of the stack).\\ The last requirement is to sanitise the stack. This makes sure that the caller does not try to sneak a capability to the callback simply by leaving it somewhere on the stack and hoping that we won't notice and pass it on. }

\section{Logical Relation}
\subsection{Worlds}
Assume a sufficiently large set of states $\States$ that at least contains the states used in this document.
\begin{definition}
  \[
   \Rels = \{(\phi_\pub, \phi) \in \powerset{\States^2}\times \powerset{\States^2} \mid \phi_\pub, \phi \text{ is reflexive and transitive and } \phi_\pub \subseteq \phi \}
  \]
\end{definition}

\begin{theorem}\label{thm:world-existence}
  There exists a \cofe{} $\Wor$ and preorders $\futurestr$ and $\futurewk$ such that $(\Wor,\futurestr)$ and $(\Wor,\futurewk)$ are preordered \cofes{} and there exists an isomorphism $\xi$ such that
  \begin{align*}
    \xi : \Wor \cong \blater (\nats \finparfun ( & \{\revoked\}  + \\
                                          & \{\temp\} \times \States \times \Rels \times (\States \fun (\Wor \monwknefun \UPred{\HeapSegments})) + \\
                                          & \{\perma\} \times \States \times \Rels \times (\States \fun (\Wor \monstrnefun \UPred{\HeapSegments}))))
  \end{align*}
 and for $W, W' \in \Wor$
 \[
   W' \futurestr W \Leftrightarrow \xi(W') \futurestr \xi(W)
 \]
and
 \[
   W' \futurewk W \Leftrightarrow \xi(W') \futurewk \xi(W)
 \]
\end{theorem}

We now define the regions to be
\begin{align*}
  \Regions = & \{\revoked\} \uplus \\
             & \{\temp\} \times \States \times \Rels \times (\States \fun (\Wor \monwknefun \UPred{\HeapSegments})) \uplus \\
             & \{\perma\} \times \States \times \Rels \times (\States \fun (\Wor \monstrnefun \UPred{\HeapSegments}))
\end{align*}
Let $\iota.v$ be the projection of the view of a region.

And the worlds are
\[
  \Worlds = \RegionNames \finparfun \Regions
\]
where $\RegionNames = \nats$.

The two \emph{private future} region relations satisfies the following properties:
\begin{mathpar}
  \inferrule{  (s,s') \in \phi \\
    (v,\phi_\pub,\phi,H) = (v',\phi_\pub',\phi',H')}
  {  (v',s',\phi_\pub',\phi',H') \futurestr (v,s,\phi_\pub,\phi,H) }
  \and
  \inferrule{ r \in \Regions }
  { r \futurestr (\temp,s,\phi_\pub,\phi,H) }
  \and
  \inferrule{ r \in \Regions }
  { r \futurestr \revoked }
\end{mathpar}
The two \emph{public future} region relations satisfies the following properties:
\begin{mathpar}
  \inferrule{  (s,s') \in \phi_\pub \\
    (v,\phi_\pub,\phi,H) = (v',\phi_\pub',\phi',H')}
  {  (v',s',\phi_\pub',\phi',H') \futurewk (v,s,\phi_\pub,\phi,H) }
  \and
  \inferrule{ (\temp,s,\phi_\pub,\phi,H) \in \Regions }
  { (\temp,s,\phi_\pub,\phi,H) \futurewk \revoked }
  \and
  \inferrule{ }
  { \revoked \futurewk \revoked }
\end{mathpar}

The two future world relations satisfy the following properties: They allow for any extension of the current world and all existing worlds are allowed to move to an appropriate future region. That is
\begin{mathpar}
  \inferrule{ \dom(W') \supseteq \dom(W)\\ 
    \forall r \in \dom(W) \ldotp W'(r) \futurewk W(r) }
  { W' \futurewk W }
  \and
  \inferrule{ \dom(W') \supseteq \dom(W)\\ 
    \forall r \in \dom(W) \ldotp W'(r) \futurestr W(r) }
  { W' \futurestr W }
\end{mathpar}


\begin{proof}[Proof of Theorem~\ref{thm:world-existence}]
  The theorem follows from a more general solution theorem for the
  category of $P$ preordered \cofes,
  see~\cite{Birkedal:2010:TCS:411:4102-4122}, \cite{Birkedal:tutorial-notes} and \cite{bizjak-note}. 
  We define two functors $F_1$ and $F_2$ from $P^{op}\times P^{op}$ to
  $P$. 
  \begin{align*}
     F_1((X,{\futurestr}'),(Y,{\futurewk}')) = \\
     (\blater (\nats \finparfun ( & \{\revoked\}  + \\
                                          & \{\temp\} \times \States \times \Rels \times (\States \fun ((Y,{\futurewk}') \monnefun \UPred{\HeapSegments})) + \\
                                          & \{\perma\} \times \States
                                            \times \Rels \times
                                            (\States \fun ((X,{\futurestr}')
                                            \monnefun
                                            \UPred{\HeapSegments})))),
                                        \futurestr)
  \end{align*}
  and
  \begin{align*}
     F_2((X,{\futurestr}'),(Y,{\futurewk}')) = \\
     (\blater (\nats \finparfun ( & \{\revoked\}  + \\
                                          & \{\temp\} \times \States \times \Rels \times (\States \fun ((Y,{\futurewk}') \monnefun \UPred{\HeapSegments})) + \\
                                          & \{\perma\} \times \States
                                            \times \Rels \times
                                            (\States \fun ((X,{\futurestr}')
                                            \monnefun
                                            \UPred{\HeapSegments})))),
                                        \futurewk)
  \end{align*}
  The orderings $\futurestr$ and $\futurewk$ used in the definition of
  $F_1$ and $F_2$ are defined by the properties given above.
  Note that the image of $F_1$ and $F_2$ only differ in the ordering
  relation, i.e., letting $U$ denote the forgetful functor from the
  category of preordered \cofes\ to the category of \cofes, we have
  $U \circ F_1 = U \circ F_2$.
  From~\cite{bizjak-note} it then follows that there exists a \cofe\ $\Wor$ and two
  preorderings $\futurestr$ and $\futurewk$ and an isomorphism $\xi$ satisfying the properties
  claimed in theorem.
  (Here, in the proof, we have written the ordering explicitly on the \cofe\ when
  using monotone non-expansive functions; in the theorem formulation
  we have instead annotated the arrow to indicate which ordering is used.)
\end{proof}


Erase all but a set of views:
\begin{align*}
  \lfloor W \rfloor_S \defeq \lambda r \ldotp 
  \begin{cases}
    W(r) & W(r).v \in S\\
    \bot & \text{otherwise}
  \end{cases}
\end{align*}

Define the function $\activeReg{\cdot}$ as follows
\begin{align*}
  \activeReg{} & : \Worlds \fun 2^\RegionNames \\
  \activeReg{W} & \defeq \dom(\erase{W}{\perma,\temp})
\end{align*}

Memory segment satisfaction:
\begin{align*}
  &\heapSat[\hs]{n}{W} 
    \text{ iff }
    \left\{\begin{aligned}
        &\exists P : \activeReg{W} \rightarrow \HeapSegments \ldotp \\
        &\quad \memSatPar{\ms}{W}{P}
      \end{aligned}\right.
\end{align*}

\begin{align*}
  &\memSatPar{\ms}{W}{P}
    \text{ iff }
    \left\{\begin{aligned}
        &\hs = \biguplus_{r\in\activeReg{W}}P(r) \land \\
        &\forall r \in \activeReg{W} \ldotp\\
        &\quad \exists H,s \ldotp\\
        &\qquad W(r) = (\_,s,\_,\_,H) \land \\
        &\qquad \npair[n]{P(r)} \in H(s)(\xi^{-1}(W))\\
      \end{aligned}\right.
\end{align*}

Standard regions for when writing locally is permitted:
\begin{align*}
  &\iota^\pwl : \cal{P} \fun \Regions\\
  &\iota^\pwl\;A \defeq (\temp,1,=,=,H^\pwl \; A) \\
  \\
  &H^\pwl : \cal{P}(\Addrs) \fun \States \fun (\Wor \monwknefun \UPred{\HeapSegments})\\
  &H^\pwl\; A \; s \; \hat{W} \defeq \left\{\npair{\hs} \middle|
    \begin{aligned}
      &\dom(\hs) = A \land \\
      &\forall \addr \in A \ldotp \npair[n-1]{\hs(\addr)} \in \stdvr(\xi(\hat{W}))
    \end{aligned}
        \right\} \union \{\npair[0]{\ms} \}
\end{align*}
Revoking all temporary regions:
\begin{align*}
  \revokeTemp{} & : \Worlds \fun \Worlds \\
  \revokeTemp{W} & \defeq \lambda r \ldotp 
                   \begin{cases}
                     \revoked            & \text{if }W(r) = (\temp,s,\phi_\pub,\phi,H) \\
                     W(r)                & \text{otherwise}
                   \end{cases}
\end{align*}
Further define
\[
  \iota^\pwl_{\start,\addrend} \defeq \iota^\pwl([\start,\addrend])
\]

\newcommand{\wrev}[1]{\revokeTemp{#1}}

Standard regions for when write local is not allowed:
\begin{align*}
  &\iota^\nwl : \cal{P}(\Addrs) \fun \Regions \\
  &\iota^\nwl\; A \defeq (\temp,1,=,=,H^\nwl\;A) \\
  \\
  &\iota^{\nwl,p} : \cal{P}(\Addrs) \fun \Regions \\
  &\iota^{\nwl,p}\;A \defeq (\perma,1,=,=,H^\nwl\;A)\\
  \\
  &H^\nwl : \cal{P}(\Addrs) \fun \States \fun (\Wor \monstrnefun \UPred{\HeapSegments})\\
  &H^\nwl\;A \; s \;\hat{W} \defeq \left\{\npair{\hs} \middle|
    \begin{aligned}
      &\dom(\hs) = A \land \\
      &\forall \addr \in A \ldotp \\
      &\quad \nonlocal{\ms(\addr)} \land\\ 
      &\quad \npair[n-1]{\hs(\addr)} \in \stdvr(\xi(\hat{W}))
    \end{aligned}
        \right\} \union \{\npair[0]{\ms} \}
\end{align*}

Further define

\begin{align*}
  \iota^\nwl_{\start,\addrend} & \defeq \iota^\nwl([\start,\addrend])\\
  \iota^{\nwl,p}_{\start,\addrend} & \defeq \iota^{\nwl,p}([\start,\addrend])\\
\end{align*}

For convenience define
\[
  \localityReg(\gl,W) \defeq 
  \begin{cases}
    \dom(\erase{W}{\perm,\temp}) & \text{if } \gl = \local \\
    \dom(\erase{W}{\perm}) & \text{if } \gl = \glob
  \end{cases}
\]
$\localityReg(\local,W)$ are the regions that local capabilities may govern - that is permanent and temporary regions. $\localityReg(\glob,W)$ are the regions that global capabilities may govern - that is permanent regions. Now define the following function

We need a notion of subset between regions that is almost $n$-subset, but not quite. The only difference is that the view part of a region is disregarded. Define ``semi $n$-subset'' and ``semi $n$-supset'' as:
\lau{ Using normal the canonical $n$-subset gave some technical problems in the proofs. For instance, it was not possible to have a global read capability for a region governed by a permanent standard region. The issue was that $\iota^\pwl$ is always $\temp$ which gave some issues in the read and write conditions when they used $\iota^\pwl$. Another option was to be more explicit, but we opted for this more elegant(?) solution. }
\begin{mathpar}
  \inferrule{(s,\phi_\pub,\phi) = (s',\phi_\pub',\phi') \\
    \forall \hat{W} \ldotp H \; s \; \hat{W} \nsubeq H' \; s' \; \hat{W} }
  { (v,s,\phi_\pub,\phi,H) \nsubsim (v',s',\phi_\pub',\phi',H')}
\end{mathpar}

\subsection{The logical relation}
The logical relation is defined by several mutual recursive definitions. In order to handle this mutual recursion and show that this definitions are well-defined, Banach's fixed-point theorem can be used. 
We have omitted the details of this construction here, but it is done by parameterising all the definitions by the value relation.

\begin{align*}
&\iota = (v,s,\phi_\pub,\phi,H) \text{ is address-stratified iff }\\
&\qquad\begin{multlined}
  \forall s', \hat{W}, n, \ms, \ms' \ldotp \\
  \npair{\ms},\npair{\ms'} \in H~ s'~ \hat{W} \Rightarrow \\
  \dom(\ms) = \dom(\ms') \wedge \\
  \forall \addr \in
  \dom(\ms)\ldotp \npair{\ms\update{\addr}{\ms'(\addr)}} \in H~ s'~ \hat{W}
\end{multlined}
\end{align*}

\begin{align*}
  & \writeCond{} : (((\Addrs\times\Addrs) \fun\Regions) \times \Globals) \fun \Worlds \monnefun \UPred{\Addrs^2}  \\
  & \writeCond{}(\iota,\gl)(W) =  \\
  & \quad \begin{aligned}[t]
    \{ \npair{(\start,\addrend)} \mid & \;\exists r \in \var{localityReg}(g,W) \ldotp \\
    & \;\quad \exists [\start',\addrend'] \supseteq [\start,\addrend] \ldotp \\
    & \;\qquad W(r)\nsupsim[n-1] \iota_{\start',\addrend'} \text{ and }\\
    & \;\qquad W(r) \text{ is address-stratified } \}
  \end{aligned}
\end{align*}

\begin{align*}
  & \readCond{} : \Globals \fun \Worlds \monnefun \UPred{\Addrs^2}  \\
  & \readCond{}(\gl)(W) =  \\
  & \quad \begin{aligned}[t]
    \{ \npair{(\start,\addrend)} \mid & \;\exists r \in \var{localityReg}(g,W) \ldotp \\
    & \;\quad \exists [\start',\addrend'] \supseteq [\start,\addrend] \ldotp \\
    & \;\qquad W(r)\nsubsim[n] \iota_{\start',\addrend'}^\pwl \}
  \end{aligned}
\end{align*}

\begin{align*}
  & \execCond{}(\gl)(W) = \\
  & \quad
    \begin{aligned}[t]
      \{ \npair{(\perm,\start,\addrend)} \mid &  \forall n' < n \ldotp\\
      & \quad \forall W' \future W \ldotp\\
      & \qquad\forall a \in [\start',\addrend'] \subseteq [\start,\addrend] \ldotp\\
      & \qquad \quad \npair[n']{((\perm,\gl),\start',\addrend',\addr)} \in \stder(W')\}
    \end{aligned} \\
  & \quad \text{where } \gl = \local \Rightarrow \future = \futurewk \\
  & \quad \text{and } \gl = \glob \Rightarrow \future = \futurestr
\end{align*}

\begin{align*}
  & \entryCond{}(\gl)(W) = \\
  & \quad
    \begin{aligned}[t]
      \{ \npair{(\start,\addrend,\addr)} \mid &  \forall n' < n \ldotp\\
      & \quad \forall W' \future W \ldotp\\
      & \qquad \npair[n']{((\exec,\gl),\start,\addrend,\addr)} \in \stder(W')\}
    \end{aligned} \\
  & \quad \text{where } \gl = \local \Rightarrow \future = \futurewk \\
  & \quad \text{and } \gl = \glob \Rightarrow \future = \futurestr
\end{align*}

Now define the value relation as follows:
\begin{align*}
  &\stdvr : \Worlds \monwknefun \UPred{\Words} \\
  &\stdvr\defeq \lambda \; W \ldotp 
    \begin{aligned}[t]
      & \{ \npair{i} \mid i \in \ints \union \{ \infty \} \} 
      \union \\
      & \{ \npair{\stdcap[(\noperm,\gl)] }  \} 
      \union \\
      & \{ \npair{\stdcap[(\readonly,\gl)] } \mid \\
      & \quad \npair{(\start,\addrend)} \in \readCond{}(\gl)(W)\} 
      \union \\
      & \{ \npair{\stdcap[(\readwrite,\gl)] } \mid \\
      & \quad \npair{(\start,\addrend)} \in \readCond{}(\gl)(W) \land \\
      & \quad \npair{(\start,\addrend)} \in \writeCond{}(\iota^\nwl,\gl)(W) \}
      \union \\
      & \{ \npair{\stdcap[(\readwritel,\gl)] } \mid \\
      & \quad \npair{(\start,\addrend)} \in \readCond{}(\gl)(W) \land \\
      & \quad \npair{(\start,\addrend)} \in \writeCond{}(\iota^\pwl,\gl)(W) \}
      \union \\
      & \{ \npair{\stdcap[(\exec,\gl)]} \mid \\
      & \quad \npair{(\start,\addrend)} \in \readCond{}(\gl)(W) \land \\
      & \quad \npair{(\exec,\start,\addrend)} \in \execCond{}(\gl)(W) \} 
      \union \\
      & \{ \npair{\stdcap[(\entry,\gl)]} \mid \\
      & \quad \npair{(\start,\addrend,\addr)} \in \entryCond{}(\gl)(W)\} 
      \union \\
      & \{ \npair{\stdcap[(\rwx,\gl)]} \mid \\
      & \quad \npair{(\start,\addrend)} \in \readCond{}(\gl)(W) \land \\
      & \quad \npair{(\start,\addrend)} \in \writeCond{}(\iota^\nwl,\gl)(W) \land\\
      & \quad \npair{(\rwx,\start,\addrend)} \in \execCond{}(\gl)(W)  \land \\
      & \quad \npair{(\exec,\start,\addrend)} \in \execCond{}(\gl)(W) \}
      \union \\
      & \{ \npair{\stdcap[(\rwlx,\gl)]} \mid \\
      & \quad \npair{(\start,\addrend)} \in \readCond{}(\gl)(W) \land \\
      & \quad \npair{(\start,\addrend)} \in \writeCond{}(\iota^\pwl,\gl)(W) \land\\
      & \quad \npair{(\rwlx,\start,\addrend)} \in \execCond{}(\gl)(W) \land \\
      & \quad \npair{(\rwx,\start,\addrend)} \in \execCond{}(\gl)(W) \land \\
      & \quad \npair{(\exec,\start,\addrend)} \in \execCond{}(\gl)(W) \}
    \end{aligned}
\end{align*}

\begin{align*}
  \observations : &  \Worlds \nefun \UPred{\Regs \times \HeapSegments} \\
  \observations \defeq & \lambda W \ldotp 
                             \{ \npair{(\reg,\hs)} \mid
                             \begin{aligned}[t]
                               & \forall \ms_f, \heap', i \leq n \ldotp \\
                               & \quad(\reg,\hs \uplus \ms_f) \step[i] (\halted,\heap')  \\
                               & \qquad \Rightarrow
                               \begin{aligned}[t]
                                 & \exists W' \futurestr W \ldotp\exists \hs_r, \hs' \ldotp\\
                                 & \quad \heap' = \hs' \uplus \hs_r \uplus \ms_f \land \\ 
                                 & \quad \heapSat[\hs']{n-i}{W'} \}
                               \end{aligned}
                             \end{aligned}
\end{align*}

\begin{align*}
  \stdrr : & \Worlds \monwknefun \UPred{\Regs} \\
  \stdrr \defeq & \lambda W \ldotp
                  \begin{aligned}[t]
                    \{ \npair{\reg} \mid & \;\forall r \in \RegName \setminus \{\pcreg\} \ldotp \\
                    & \;\quad  \npair{\reg(r)} \in \stdvr(W) \}
                  \end{aligned}
\end{align*}

\begin{align*}
  \stder : & \Worlds \nefun \UPred{\Words} \\
  \stder \defeq & \lambda W \ldotp \{ \npair{\pc} \mid 
                  \begin{aligned}[t]
                    & \forall n' \leq n \ldotp\\
                    & \quad \forall \npair[n']{\reg} \in \stdrr(W) \ldotp \\
                    & \quad \qquad  \forall \heapSat[\hs]{n'}{W} \ldotp \\
                    & \quad \qquad \quad \npair[n']{(\reg\update{\pcreg}{\pc},\hs)} \in \observations(W) \}
                  \end{aligned}
\end{align*}

\subsection{Useful regions}
Static region used for parts of memory that should not change.
\begin{align*}
  \iota^\sta (v,\ms) &= (v,1,=,=,H^\sta\;\ms)\\
  H^\sta \; \ms \; s \; \hat{W} = & \{\npair{\ms} \mid n > 0 \} \union \{\npair[0]{\ms'} \mid \ms' \in \Mems \}
\end{align*}

Static region used for parts of memory that should not change and where you pass control to untrusted code.
\begin{align*}
  \iota^{\sta,u} (v,\ms) &= (v,1,=,=,H^{\sta,u}\;\ms)\\
  H^{\sta,u} \; \ms \; s \; \hat{W} = & \left\{\npair{\ms'} \middle|
    \begin{aligned}
      &\ms' = \ms \land\\
      &\forall \addr \in \dom(\ms) \ldotp\\
      & \quad \nonlocal{\ms(\addr)} \land\\
      & \quad \npair[n-1]{\ms(\addr)} \in \stdvr(\xi(\hat{W}))
    \end{aligned}
        \right\} \union \{\npair[0]{\ms'} \mid \ms' \in \Mems \}
\end{align*}

\begin{align*}
  \iota^\cnst (v,n) &= (v,1,=,=,H^\cnst\;n)\\
  H^\cnst \; n' \; s \; \hat{W} = &  \{\npair{\ms} \mid n > 0 \land \forall \addr \in \dom(\ms) \ldotp \ms(\addr) = n' \} \union \{\npair[0]{\ms'} \mid \ms' \in \Mems \}
\end{align*}

\subsection{Lemmas}
\subsubsection{Anti-reduction for the observation relation}

\begin{lemma}[Failing terms are in $\observations$ and $\stder$]
  \label{lem:failed-obs-stder}
  If $(\reg,\ms\uplus\ms_f) \step[*] \failed$ for all $\ms_f$, then
  $\npair{(\reg,\ms)} \in \observations(W)$ for any $W$.

  If $(\reg\update{\pcreg}{\word},\ms) \step[*] \failed$ for all $\reg, \ms$,
  then $\npair{\word} \in \stder(W)$ for any $W$.
\end{lemma}
\begin{proof}
  Follows from the definitions of $\observations(W)$ and $\stder(W)$ using an
  (omitted) determinacy result. \dominique{at least state the result
    somewhere?}
\end{proof}

\begin{lemma}[Anti-reduction for $\observations$]
\label{lem:anti-red-obs}
  \begin{align*}
    & \forall n, n',i,\reg, \reg', \ms, \ms', \ms_r, W, W' \ldotp\\
    & \quad n' \geq n - i \land W' \futurestr W \land \\
    & \quad (\forall \ms_f \ldotp (\reg,\ms \uplus \ms_r \uplus \ms_f) \step[i] (\reg',\ms' \uplus \ms_r \uplus \ms_f)) \land \\
    & \quad \npair[n']{(\reg',\ms')} \in \observations(W') \\
    & \qquad \Rightarrow \npair{(\reg,\ms \uplus \ms_r)} \in \observations(W)
  \end{align*}
\end{lemma}
\begin{proof}[Proof of Lemma~\ref{lem:anti-red-obs}]
  Assume
  \begin{enumproof} 
    \item $n' \geq n - i$ \label{ass:anti-red-steps}
    \item $W_2 \futurestr W_1$ \label{ass:anti-red-ftrworld}
    \item $\forall \ms_f \ldotp (\reg,\ms \uplus \ms_r \uplus \ms_f) \step[i] (\reg',\ms' \uplus \ms_r \uplus \ms_f)$ \label{ass:anti-red-exec}
    \item $\npair[n']{(\reg',\ms')} \in \observations(W_2)$ \label{ass:anti-red-obs}
  \end{enumproof}
  Show 
  \[
    \npair{(\reg,\ms \uplus \ms_r)}  \in \observations(W_1)
  \]
  To this end let $\ms_{\var{frame}}$, $m'$ and $j$ be given and assume
  \begin{equation}
    \label{ass:anti-red-termination}
    (\reg,\ms \uplus \ms_r \uplus \ms_{\var{frame}}) \step[j] (\halted,m')
  \end{equation}
  From \ref{ass:anti-red-exec} instantiated with $\ms_{\var{frame}}$ we know
  \begin{equation}
    \label{ass:anti-red-part-exec}
    (\reg,\ms \uplus \ms_r \uplus \ms_{\var{frame}}) \step[i] (\reg',\ms' \uplus \ms_r \uplus \ms_{\var{frame}})
  \end{equation}
  Using \ref{ass:anti-red-termination} and \ref{ass:anti-red-part-exec}, we get 
  \[
    (\reg',\ms' \uplus \ms_r \uplus \ms_{\var{frame}}) \step[j-i] (\halted,m')
  \]
  Using this with \ref{ass:anti-red-obs} and $\ms_r \uplus \ms_{\var{frame}}$ as frame, we get $W_3 \futurestr W_2$, $\ms''$ and $\ms_{\var{rev}}$ such that
  \begin{enumproof}[resume]
    \item $m' = \ms'' \uplus \ms_{\var{rev}} \uplus (\ms_r \uplus \ms_{\var{frame}})$ \label{ass:anti-red-mem-split}
    \item $\memSat[n'-(j-i)]{\ms''}{W_3}$ \label{ass:anti-red-mem-sat}
  \end{enumproof}
  Now use $\ms_r \uplus \ms_{\var{rev}}$ as the ``revoked'' memory, $\ms''$ as the memory that satisfies some invariants, and $W_3$ as the desired world, then \ref{ass:anti-red-mem-split} gives us the split and by downwards closure \ref{ass:anti-red-mem-sat} gives us the desired memory satisfaction.
\end{proof}

\subsubsection{Standard regions}

\begin{lemma}
  \label{lem:pwl-stack}
  For all $W$, $\start$, $\addrend$, $n$, $\ms$ if
  \begin{itemize}
  \item $\memSat{\ms}{W}$
  \item $\npair{\stdcap} \in \stdvr(W)$
  \item $\start \leq \addrend$
  \item $\perm \in \{ \rwlx, \rwx \}$
  \end{itemize}
  then
  \[
    \exists r, \start',\addrend' \ldotp [\start,\addrend] \subseteq [\start',\addrend'] \land W(\var{r}) \nequal \iota^\pwl_{\start',\addrend'}
  \]
\end{lemma}
\begin{proof}[Proof of Lemma~\ref{lem:pwl-stack}]
  Assume
  \begin{enumproof}
    \item $\npair{((\rwlx, \gl),b,e,a)} \in \stdvr(W)$ \label{ass:pwl-stack:val-rel}
    \item $\memSat{\ms}{W}$ \label{ass:pwl-stack:mem-sat}
  \end{enumproof}
  From Assumption~\ref{ass:pwl-stack:val-rel}, we get $r_1$, $r_2$, $b_1$, $b_2$, $e_1$ and $e_2$ such that
  \begin{enumproof}[resume]
    \item $r_1 \in \localityReg(\gl,W)$
    \item $r_2 \in \localityReg(\gl,W)$
    \item $[b,e] \subseteq [b_1,e_1]$\label{ass:pwl-stack:b1}
    \item $[b,e] \subseteq [b_2,e_2]$ \label{ass:pwl-stack:b2}
    \item $W(r_1) \nsubsim \iota^\pwl_{b_1,e_1}$ \label{ass:pwl-stack:subsim}
    \item $W(r_2) \nsupsim \iota^\pwl_{b_2,e_2}$ \label{ass:pwl-stack:supsim}
    \item $W(r_2)$ is address-stratified. \label{ass:pwl-stack:add-stra}
  \end{enumproof}
From Assumption~\ref{ass:pwl-stack:mem-sat}, we get partitionen $P$ s.t.
\[
  \memSat[n,p]{\ms}{W}
\]
Say $P(r_1) = \ms_1$ and $P(r_2) = \ms_2$. First from $\npair{\ms_1} \in W(r_1).H\;W(r_1).s\;\xi^{(-1)}(W)$ using \label{ass:pwl-stack:subsim}, we get $\npair{\ms_1} \in H^\pwl_{b_1,e_1}\;1\;\xi^{(-1)}(W)$ which means $\dom(\ms_1) = [b_1,e_1]$.

Second we know $\npair{[b_2 \mapsto 0, \dots , e_2 \mapsto 0]} \in H^\pwl_{b_2,e_2}\;1\;\xi^{(-1)}(W)$ and $\npair{\ms_2} \in W(r_2).H\;W(r_2).s\;\xi^{(-1)}(W)$ which by Assumption~\ref{ass:pwl-stack:supsim} and \ref{ass:pwl-stack:add-stra} means $\dom(\ms_2) = [b_2,e_2]$.

Now assume for contradition $r_1 \neq r_2$, then we have a contradiction with $\memSat[n,p]{\ms}{W}$ because $\ms_1$ and $\ms_2$ are not disjoint (by Assumptions~\ref{ass:pwl-stack:b1} and \ref{ass:pwl-stack:b2}). So $r_1=r_2$ which also means $[b_1,e_1] = [b_2,e_2]$, so from Assumption~\ref{ass:pwl-stack:subsim} and \ref{ass:pwl-stack:supsim}, we get $W(r_1) \nsim \iota^\pwl_{b_1,e_1}$ which by Lemma~\ref{lem:pwl-nsim-view} means $W(r_1) \nequal \iota^\pwl_{b_1,e_1}$
\end{proof}

\begin{lemma}
  \label{lem:hpwl-mono}
  $H^\pwl_{\start,\addrend} \; s$ is monotone w.r.t $\futurewk$ for all $s \in \States$ and $\start$ and $\addrend$
\end{lemma}
\begin{proof}[Proof of Lemma~\ref{lem:hpwl-mono}]
  Let $\hat{W}' \futurewk \hat{W}$ be given and let 
  \begin{equation}
    \label{pf:hpwl-mono:ass}
    \npair{\ms} \in H^\pwl_{\start,\addrend}\; s \; \hat{W}
  \end{equation}
  and show
  \[
    \npair{\ms} \in H^\pwl_{\start,\addrend}\; s \; \hat{W}'
  \]
  From \ref{pf:hpwl-mono:ass}, we get $\dom(\ms) = [\start,\addrend]$. Now let $\addr \in [\start,\addrend]$ be given and show
  \[
    \npair[n-1]{\ms(a)} \in \stdvr(\xi(\hat{W}'))
  \]
  now this follows from Lemma~\ref{lem:stdvr-mono-wk}, $\hat{W}' \futurewk \hat{W}$, Theorem~\ref{thm:world-existence}, and Assumption~\ref{pf:hpwl-mono:ass}.
\end{proof}

\begin{lemma}
  \label{lem:iotapwl-is-a-region}
  $\iota^\pwl_{\start,\addrend}$ is a region for all $\start$ and $\addrend$.
\end{lemma}
\begin{proof}[Proof of Lemma~\ref{lem:iotapwl-is-a-region}]
  Follows from Lemma~\ref{lem:hpwl-mono}.
\end{proof}

\begin{lemma}
  \label{lem:iota-pwl-address-stratified}
  $\iota^\pwl_{\start,\addrend}$ is address-stratified.
\end{lemma}
\begin{proof}
  Easy unfolding of definitions.
\end{proof}

\begin{lemma}
  \label{lem:hnwl-mono-str}
  $H^\nwl_{\start,\addrend} \; s$ is monotone w.r.t $\futurestr$ for all $s \in \States$ and $\start$ and $\addrend$
\end{lemma}
\begin{proof}[Proof of Lemma~\ref{lem:hnwl-mono-str}]
  Let $\hat{W}' \futurestr \hat{W}$ be given and let 
  \begin{equation}
    \label{pf:hnwl-mono:ass}
    \npair{\ms} \in H^\nwl_{\start,\addrend}\; s \; \hat{W}
  \end{equation}
  and show
  \[
    \npair{\ms} \in H^\nwl_{\start,\addrend}\; s \; \hat{W}'
  \]
  From \ref{pf:hnwl-mono:ass}, we get $\dom(\ms) = [\start,\addrend]$. Now let $\addr \in [\start,\addrend]$ be given and show
  \begin{enumproof}
    \item \nonlocal{\ms(a)} \label{pf:hnwl-mono:ob1}
    \item $\npair[n-1]{\ms(a)} \in \stdvr(\xi(\hat{W}'))$ \label{pf:hnwl-mono:ob2}
  \end{enumproof}
  \ref{pf:hnwl-mono:ob1} follows trivially from \ref{pf:hnwl-mono:ass}. \ref{pf:hnwl-mono:ob1} follows from Assumption~\ref{pf:hpwl-mono:ass}, \ref{pf:hnwl-mono:ob1} (which we just argued), $\hat{W}' \futurestr \hat{W}$, Theorem~\ref{thm:world-existence}, and Lemma~\ref{lem:stdvr-non-loc-priv-mono}.
\end{proof}

\begin{lemma}
  \label{lem:iotanwl-is-a-region}
  $\iota^\nwl_{\start,\addrend}$ is a region for all $\start$ and $\addrend$.
\end{lemma}
\begin{proof}[Proof of Lemma~\ref{lem:iotanwl-is-a-region}]
  Follows from Lemma~\ref{lem:hnwl-mono-str} and Lemma~\ref{lem:future-pub-impl-future-priv}.
\end{proof}

\begin{lemma}
  \label{lem:iota-nwl-address-stratified}
  $\iota^\nwl_{\start,\addrend}$ is address-stratified.
\end{lemma}
\begin{proof}
  Easy unfolding of definitions.
\end{proof}

\begin{lemma}
  \label{lem:iotanwlp-is-a-region}
  $\iota^{\nwl,p}_{\start,\addrend}$ is a region for all $\start$ and $\addrend$.
\end{lemma}
\begin{proof}[Proof of Lemma~\ref{lem:iotanwlp-is-a-region}]
  Follows from Lemma~\ref{lem:hnwl-mono-str}.
\end{proof}

\begin{lemma}
  \label{lem:iotasta-is-a-region}
  $\iota^\sta(v,\ms)$ is a region for all $v \in \{\perma, \temp\}$ and $\ms$.
\end{lemma}
\begin{proof}[Proof of Lemma~\ref{lem:iotasta-is-a-region}]
  $H^\sta$ does not depend on $\hat{W}$, so it is trivial to show the necessary non-expansive and monotonicity requirements.
\end{proof}

\begin{lemma}
  \label{lem:hstau-mono-str}
  $H^{\sta,u}(\ms) \; s$ is monotone w.r.t $\futurestr$ for all $s \in \States$ and $\ms$.
\end{lemma}
\begin{proof}[Proof of Lemma~\ref{lem:hstau-mono-str}]
  Let $\hat{W}' \futurestr \hat{W}$ be given and let 
  \begin{equation}
    \label{pf:hstau-mono:ass}
    \npair{\ms'} \in H^{\sta,u}(\ms)\; s \; \hat{W}
  \end{equation}
  and show
  \[
    \npair{\ms'} \in H^{\sta,u}(\ms)\; s \; \hat{W}'
  \]
  From \ref{pf:hstau-mono:ass}, we get $\ms' = \ms$. Now let $\addr \in \dom(\ms)$ be given and show
  \begin{enumproof}
    \item \nonlocal{\ms(a)} \label{pf:hstau-mono:ob1}
    \item $\npair[n-1]{\ms(a)} \in \stdvr(\xi(\hat{W}'))$ \label{pf:hstau-mono:ob2}
  \end{enumproof}
  \ref{pf:hstau-mono:ob1} follows trivially from \ref{pf:hstau-mono:ass}. \ref{pf:hstau-mono:ob1} follows from Assumption~\ref{pf:hstau-mono:ass}, \ref{pf:hstau-mono:ob1} (which we just argued), $\hat{W}' \futurestr \hat{W}$, Theorem~\ref{thm:world-existence} and Lemma~\ref{lem:stdvr-non-loc-priv-mono}.
\end{proof}

\begin{lemma}
  \label{lem:iotastau-is-a-region}
  $\iota^{\sta,u}(v,\ms)$ is a region for all $v \in \{\perma, \temp\}$ and $\ms$.
\end{lemma}
\begin{proof}[Proof of Lemma~\ref{lem:iotastau-is-a-region}]
  Follows from Lemma~\ref{lem:hstau-mono-str} and Lemma~\ref{lem:future-pub-impl-future-priv}.
\end{proof}

\begin{lemma}
  \label{lem:hwnl-nsubset-hpwl}
  \[
    H^\nwl_{\start,\addrend} \; s\; \hat{W} \nsubeq H^\pwl_{\start,\addrend} \; s\; \hat{W}
  \]
\end{lemma}
\begin{proof}[Proof of Lemma~\ref{lem:hwnl-nsubset-hpwl}]
  Trivial. Let
  \[
    \npair{\ms} \in H^\nwl_{\start,\addrend} \; s\; \hat{W}
  \]
  and show
  \[
    \npair{\ms} \in H^\pwl_{\start,\addrend} \; s\; \hat{W}
  \]
  From the assumption, we get $\dom(\ms) = [\start,\addrend]$. We further need to show
  \[
    \forall \addr \in \dom(\ms) \ldotp \npair[n-1]{\ms(\addr)} \in \stdvr(\xi(\hat{W}))
  \]
  Given $\addr$, we know from the assumption that
  \[
    \npair[n-1]{\ms(\addr)} \in \stdvr(\xi(\hat{W}))
  \]
\end{proof}

\begin{lemma}
  \label{lem:nwl-subset-pwl}
  \begin{align*}
    & \forall n \in \nats\ldotp \forall \start, \addrend \in \Addrs \ldotp  \\
    & \quad \iota^\nwl_{\start,\addrend} \nsubsim \iota^\pwl_{\start,\addrend}
  \end{align*}
\end{lemma}
\begin{proof}[Proof of Lemma~\ref{lem:nwl-subset-pwl}]
  Let $n$, $\start$, $\addrend$ be given and show
  \[
    \iota^\nwl_{\start,\addrend} \nsubsim \iota^\pwl_{\start,\addrend}
  \]
  They agree on the state and transition systems, so given $\hat{W}$ it suffices to show
  \[
    H^\nwl_{\start,\addrend} \; 1\; \hat{W} \nsubeq H^\pwl_{\start,\addrend} \; 1\; \hat{W}  
  \]
  which is true by Lemma~\ref{lem:hwnl-nsubset-hpwl}.
\end{proof}

\begin{lemma}
  \label{lem:nwlp-subset-pwl}
  \begin{align*}
    & \forall n \in \nats\ldotp \forall \start, \addrend \in \Addrs \ldotp  \\
    & \quad \iota^{\nwl,p}_{\start,\addrend}  \nsubsim \iota^\pwl_{\start,\addrend}
  \end{align*}
\end{lemma}
\begin{proof}[Proof of Lemma~\ref{lem:nwlp-subset-pwl}]
  Follows from Lemma~\ref{lem:hwnl-nsubset-hpwl} (see proof of Lemma~\ref{lem:nwl-subset-pwl}).
\end{proof}

\begin{lemma}
  \label{lem:stau-subset-pwl}
  \begin{align*}
    & \forall n \in \nats\ldotp \forall \start, \addrend \in \Addrs \ldotp \forall v \in \{\perma, \temp\} \ldotp \\
    & \quad \dom(\ms) = [\start,\addrend] \Rightarrow \\
    & \qquad \iota^{\sta,u}_{\start,\addrend}(v,\ms) \nsubsim \iota^\pwl_{\start,\addrend}
  \end{align*}
  \dominique{err.. shouldn't you require that $\npair{\ms(a)}\in\stdvr(W)$ for
    all $W$.  Hm: this probably doesn't hold like that?}
\end{lemma}
\begin{proof}[Proof of Lemma~\ref{lem:stau-subset-pwl}]
  Essentially the same as the proof of Lemma~\ref{lem:nwl-subset-pwl} and Lemma~\ref{lem:hwnl-nsubset-hpwl}.
\end{proof}

\begin{lemma}
  \label{lem:pwl-nsim-view}
  \begin{align*}
    & \forall n \in \nats \ldotp \forall \start, \addrend, b \in \Addrs \ldotp \forall \iota \in \Regions \\
    & \iota \nsim \iota^\pwl_{\start,\addrend} \land \start \leq \addrend \Rightarrow \iota \nequal \iota^\pwl_{\start,\addrend}
  \end{align*}
\end{lemma}
\begin{proof}[Proof of Lemma~\ref{lem:pwl-nsim-view}]
  For $n = 0$ it is trivial, so assume $n > 0$.
  Say $\iota = (v,s,\phi_{\var{pub}},\phi,H)$, then by $\nsim$, we know $s = 1$, $\phi_{\var{pub}} \equiv \phi \equiv =$, and $H = H^\pwl_{\start,\addrend}$. It remains to show that $v=\temp$. To do so, we show that it cannot be the case that $v=\perma$. If $v = \perma$, then $H$ must be monotone with respect to $\futurestr$. If we can show that this is not the case, then for $\iota$ to be a region it must be the case that $v \neq \perma$ and thus $v = \temp$.

  To this end let $b\not\in[\start,\addrend]$ and define the worlds:
  \begin{align*}
    \xi (W)  =& [0 \mapsto \iota^\pwl_{\start,\addrend}]  \\ 
              & [1 \mapsto \iota^\pwl_{b,b}] \\
    \xi (W')  =& [0 \mapsto \iota^\pwl_{\start,\addrend}]  \\ 
              & [1 \mapsto \revoked] 
  \end{align*}
  For these two worlds, we have $\xi(W') \futurestr \xi(W)$ and from mono.\ of $\xi^{-1}$, we have $W' \futurestr W$. Now define the following memory segment:
  \[
    \ms = [\start \mapsto ((\readonly,\local),b,b,b), \start+1 \mapsto 0, \dots , \addrend \mapsto 0]
  \]
  It is the case that
  \[
    \npair{\ms} \in H \; 1 \; W
  \]
  but
  \[
    \npair{\ms} \not\in H \; 1 \; W'
  \]
  as it is not the case that \lau{what if $n=1$? Need to figure out when exactly this holds}
  \[
    \npair[n-1]{((\readonly,\local),b,b,b)} \in \stdvr(\xi(W')).
  \]
  The only other option that remains is $v=\temp$.
\end{proof}
\lau{$a \leq b$ that the interpretation is not empty.}

\subsubsection{Observation relation}
\begin{lemma}[Observation relation ($\observations$) non-expansive]
  \label{lem:obs-rel-ne}
  \[
    W \nequal W' \Rightarrow \observations(W) \nequal \observations(W')
  \]
\end{lemma}
\begin{proof}[Proof of Lemma~\ref{lem:obs-rel-ne}]
\end{proof}

\subsubsection{Register-file relation}
\begin{lemma}[Register-file relation ($\stdrr$) non-expansive]
  \label{lem:stdrr-ne}
  \[
    W \nequal W' \Rightarrow \stdrr(W) \nequal \stdrr(W')
  \]
\end{lemma}
\begin{proof}[Proof of Lemma~\ref{lem:stdrr-ne}]
\end{proof}

\begin{lemma}[Register-file relation ($\stdrr$) monotone wrt $\futurewk$]
  \label{lem:stdrr-mono}
  \[
    W' \futurewk W \Rightarrow \stdrr(W') \nsupeq \stdrr(W)
  \]
\end{lemma}
\begin{proof}[Proof of Lemma~\ref{lem:stdrr-mono}]
\end{proof}

\subsubsection{Expression relation}
\begin{lemma}[Expression relation ($\stder$) non-exapansive]
  \label{lem:stder-ne}
  \[
    W \nequal W' \Rightarrow \stder(W) \nequal \stder(W')
  \]
\end{lemma}
\begin{proof}[Proof of Lemma~\ref{lem:stder-ne}]
\end{proof}

\subsubsection{Permission based conditions}

\begin{lemma}
  \label{lem:revoketemp-readcond}
  If
  \[
    \npair{(\start,\addrend)} \in \readCond{}(\gl)(\revokeTemp{W})
  \]
  then
  \[
    \npair{(\start,\addrend)} \in \readCond{}(\gl)(W)
  \]
\end{lemma}
\begin{proof}[Proof of Lemma~\ref{lem:revoketemp-readcond}]
  \[
    \npair{(\start,\addrend)} \in \readCond{}(\gl)(\revokeTemp{W})
  \]
  Gives $r \in \localityReg(\gl,\revokeTemp{W})$ such that
  \[
    \forall [\start',\addrend'] \subseteq [\start,\addrend] \ldotp \revokeTemp{W}(r) \nsubsim[n] \iota^\pwl_{[\start',\addrend']}
  \]
  Notice $\revokeTemp{W}(r)$ is a $\perma$ region, so $\revokeTemp{W}(r) = W(r)$. Using $r$ as witness, the result is immediate.
\end{proof}

\begin{lemma}
  \label{lem:revoketemp-writecond}
  If
  \[
    \npair{(\start,\addrend)} \in \writeCond{}(\iota,\gl)(\revokeTemp{W})
  \]
  then
  \[
    \npair{(\start,\addrend)} \in \writeCond{}(\iota,\gl)(W)
  \]
\end{lemma}
\begin{proof}[Proof of Lemma~\ref{lem:revoketemp-writecond}]
  \[
    \npair{(\start,\addrend)} \in \writeCond{}(\iota,\gl)(\revokeTemp{W})
  \]
  Gives $r \in \localityReg(\gl,\revokeTemp{W})$ such that
  \[
    \forall [\start',\addrend'] \subseteq [\start,\addrend] \ldotp \revokeTemp{W}(r) \nsupsim[n-1] \iota_{[\start',\addrend']}
  \]
  and 
  \begin{equation*}
    \revokeTemp{W}(r) \text{ is address-stratified}
  \end{equation*}
  Notice $\revokeTemp{W}(r)$ is a $\perma$ region, so $\revokeTemp{W}(r) = W(r)$. Using $r$ as witness, the result is immediate.
\end{proof}

\begin{lemma}
  \label{lem:revoketemp-execcond}
  If
  \begin{itemize}
  \item $\npair{(\perm,\start,\addrend)} \in \execCond{}(\gl)(\revokeTemp{W})$
  \end{itemize}
  then
  \[
    \npair{(\perm,\start,\addrend)} \in \execCond{}(\gl)(W)
  \]
\end{lemma}
\begin{proof}[Proof of Lemma~\ref{lem:revoketemp-execcond}]
  Use Lemma~\ref{lem:rt-w-pub-future-w}.
\end{proof}

\begin{lemma}
  \label{lem:revoketemp-entercond}
  If
  \begin{itemize}
  \item $\npair{(\addr,\start,\addrend)} \in \execCond{}(\gl)(\revokeTemp{W})$
  \end{itemize}
  then
  \[
    \npair{(\addr,\start,\addrend)} \in \execCond{}(\gl)(W)
  \]
\end{lemma}
\begin{proof}[Proof of Lemma~\ref{lem:revoketemp-entercond}]
    Use Lemma~\ref{lem:rt-w-pub-future-w}.
\end{proof}

\begin{lemma}
  \label{lem:wc-pwl-implies-wc-nwl}
  If
  \[
    \npair{(\start,\addrend)} \in \writeCond{}(\iota^\pwl,\local)(W)
  \]
  then
  \[
    \npair{(\start,\addrend)} \in \writeCond{}(\iota^\nwl,\local)(W)
  \]
\end{lemma}
\begin{proof}[Proof of lemma~\ref{lem:wc-pwl-implies-wc-nwl}]
  Follows from Lemma~\ref{lem:nwlp-subset-pwl}.
\end{proof}

\begin{lemma}[$\mathit{readCondition}$ monotone w.r.t $\futurewk$]
  \label{lem:readcond-mono-pub}
  If
  \begin{itemize}
  \item $W' \futurewk W$
  \item $\npair{(\start,\addrend)} \in \readCond{}(\gl)(W)$
  \end{itemize}
  then
  \[
    \npair{(\start,\addrend)} \in \readCond{}(\gl)(W')
  \]
\end{lemma}
\begin{proof}[Proof of Lemma~\ref{lem:readcond-mono-pub}]
\end{proof}

\begin{lemma}[$\mathit{readCondition}$ global monotonicity w.r.t $\futurestr$]
  \label{lem:readcond-mono-priv}
  If
  \begin{itemize}
  \item $W' \futurestr W$
  \item $\npair{(\start,\addrend)} \in \readCond{}(\glob)(W)$
  \end{itemize}
  then
  \[
    \npair{(\start,\addrend)} \in \readCond{}(\glob)(W')
  \]
\end{lemma}
\begin{proof}[Proof of Lemma~\ref{lem:readcond-mono-priv}]
  $\readCond{}(\glob)(W)$ picks a $\perma$ region from $W$. $\perma$ regions are persistent over $\futurestr$, so we can use the region that the assumption gives us.
\end{proof}

\begin{lemma}[$\mathit{readCondition}$ downwards-closed]
  \label{lem:readcond-dc}
  If
  \begin{itemize}
  \item $ n' \leq n$
  \item $\npair[n]{(\start,\addrend)} \in \readCond{}(\gl)(W)$
  \end{itemize}
  then
  \[
    \npair[n']{(\start,\addrend)} \in \readCond{}(\gl)(W)
  \]
\end{lemma}
\begin{proof}[Proof of Lemma~\ref{lem:readcond-dc}]
\end{proof}

\begin{lemma}[$\mathit{writeCondition}$ monotone w.r.t $\futurewk$]
  \label{lem:writecond-mono-pub}
  If
  \begin{itemize}
  \item $W' \futurewk W$
  \item $\iota \in \{\iota^\pwl,\iota^\nwl,\iota^{(\nwl,p)}\}$
  \item $\npair{(\start,\addrend)} \in \writeCond{}(\iota,\gl)(W)$
  \end{itemize}
  then
  \[
    \npair{(\start,\addrend)}\in \writeCond{}(\iota,\gl)(W')
  \]
\end{lemma}
\begin{proof}[Proof of Lemma~\ref{lem:writecond-mono-pub}]
\end{proof}

\begin{lemma}[$\mathit{writeCondition}$ global monotonicity w.r.t $\futurestr$]
  \label{lem:writecond-mono-priv}
  If
  \begin{itemize}
  \item $W' \futurestr W$
  \item $\iota \in \{\iota^\nwl,\iota^{(\nwl,p)}\}$
  \item $\npair{(\start,\addrend)} \in \writeCond{}(\iota,\glob)(W)$
  \end{itemize}
  then
  \[
    \npair{(\start,\addrend)}\in \writeCond{}(\iota,\glob)(W')
  \]
\end{lemma}
\begin{proof}[Proof of Lemma~\ref{lem:writecond-mono-priv}]
  $\writeCond{}(\iota,\glob)(W)$ picks a $\perma$ region from $W$. $\perma$ regions are persistent over $\futurestr$, so we can use the region that the assumption gives us.
\end{proof}

\begin{lemma}[$\mathit{writeCondition}$ downwards-closed]
  \label{lem:writecond-dc}
  If 
  \begin{itemize}
  \item $n' \leq n$
  \item $\iota \in \{\iota^\pwl,\iota^\nwl,\iota^{(\nwl,p)}\}$
  \item $\npair{(\start,\addrend)} \in \writeCond{}(\iota,\gl)(W)$
  \end{itemize}
  then 
  \[
    \npair[n']{(\start,\addrend)} \in \writeCond{}(\iota,\gl)(W)
  \]
\end{lemma}
\begin{proof}[Proof of Lemma~\ref{lem:writecond-dc}]
\end{proof}

\begin{lemma}[$\mathit{execCondition}$ monotone w.r.t $\futurewk$]
  \label{lem:execcond-mono-pub}
  If
  \begin{itemize}
  \item $W' \futurewk W$
  \item $\perm \in \{\exec, \rwx, \rwlx\}$
  \item $\npair{(\perm,\start,\addrend)} \in \execCond{}(\gl)(W)$
  \end{itemize}
  then
  \[
    \npair{(\perm,\start,\addrend)}\in \execCond{}(\iota,\gl)(W')
  \]
\end{lemma}
\begin{proof}[Proof of Lemma~\ref{lem:execcond-mono-pub}]
\end{proof}

\begin{lemma}[$\mathit{execCondition}$ global monotonicity w.r.t $\futurestr$]
  \label{lem:execcond-mono-priv}
  If
  \begin{itemize}
  \item $W' \futurestr W$
  \item $\perm \in \{\exec, \rwx \}$
  \item $\npair{(\perm,\start,\addrend)} \in \execCond{}(\glob)(W)$
  \end{itemize}
  then
  \[
    \npair{(\perm,\start,\addrend)}\in \execCond{}(\glob)(W')
  \]
\end{lemma}
\begin{proof}[Proof of Lemma~\ref{lem:execcond-mono-priv}]
  Assume $W_2 \futurestr W_1$, $\perm \in \{\exec, \rwx \}$ and $\npair{(\perm,\start,\addrend)} \in \execCond{}(\glob)(W_1)$. Now let $W_3 \futurestr W_2$, $\addr \in [\start',\addrend'] \subseteq [\start,\addrend]$, and $n' < n$, and show
  \[
    \npair{((\perm,\glob),\start',\addrend',\addr)} \in \stder(W_3)
  \]
  by transitivity we have $W_3 \futurestr W_1$, so the result follows from $\npair{(\perm,\start,\addrend)} \in \execCond{}(\glob)(W_1)$.
\end{proof}

\begin{lemma}[$\mathit{execCondition}$ downwards-closed]
  \label{lem:execcond-dc}
  If
  \begin{itemize}
  \item $n' \leq n$
  \item $\perm \in \{\exec, \rwx, \rwlx\}$
  \item $\npair{(\perm,\start,\addrend)} \in \execCond{}(\gl)(W)$
  \end{itemize}
  then 
  \[
    \npair[n']{(\perm,\start,\addrend)} \in \execCond{}(\gl)(W)
  \]
\end{lemma}
\begin{proof}[Proof of Lemma~\ref{lem:execcond-dc}]
  Follows easily from definition.
\end{proof}

\begin{lemma}[$\mathit{enterCondition}$ monotone w.r.t $\futurewk$]
  \label{lem:entrycond-mono-pub}
  If
  \begin{itemize}
  \item $W' \futurewk W$
  \item $\npair{(\addr,\start,\addrend)} \in \entryCond{}(\gl)(W)$
  \end{itemize}
  then
  \[
    \npair{(\addr,\start,\addrend)}\in \entryCond{}(\iota,\gl)(W')
  \]
\end{lemma}
\begin{proof}[Proof of Lemma~\ref{lem:entrycond-mono-pub}]
  Follows easily from definition.
\end{proof}

\begin{lemma}[$\mathit{enterCondition}$ global monotonicity w.r.t $\futurestr$]
  \label{lem:entrycond-mono-priv}
  If
  \begin{itemize}
  \item $W' \futurestr W$
  \item $\npair{(\addr,\start,\addrend)} \in \entryCond{}(\glob)(W)$
  \end{itemize}
  then
  \[
    \npair{(\addr,\start,\addrend)}\in \entryCond{}(\glob)(W')
  \]
\end{lemma}
\begin{proof}[Proof of Lemma~\ref{lem:entrycond-mono-priv}]
  Assume $W_2 \futurestr W_1$ and $\npair{(\addr,\start,\addrend)} \in \entryCond{}(\glob)(W_1)$. Now let $W_3 \futurestr W_2$, $n' < n$, and show
  \[
    \npair{\stdcap[(\exec,\glob)]} \in \stder(W_3)
  \]
  by transitivity we have $W_3 \futurestr W_1$, so the result follows from $\npair{(\addr,\start,\addrend)} \in \entryCond{}(\glob)(W_1)$.
\end{proof}

\begin{lemma}[$\mathit{enterCondition}$ downwards-closed]
  \label{lem:entrycond-dc}
  If
  \begin{itemize}
  \item $n' \leq n$
  \item $\npair{(\addr,\start,\addrend)} \in \entryCond{}(\gl)(W)$
  \end{itemize}
  then 
  \[
    \npair[n']{(\addr,\start,\addrend)} \in \entryCond{}(\gl)(W)
  \]
\end{lemma}
\begin{proof}[Proof of Lemma~\ref{lem:entrycond-dc}]
\end{proof}

\subsubsection{LR Sanity lemmas}
\begin{lemma}
  \label{lem:mem-sat-ne}
  \begin{align*}
    &\forall \hs, n, W \nequal W' \ldotp \\
    &\quad \heapSat[\hs]{n}{W} \land W \nequal W' \Rightarrow \heapSat[\hs]{n}{W'}
  \end{align*}
\end{lemma}
\begin{proof}[Proof of Lemma~\ref{lem:mem-sat-ne}]
\end{proof}

\begin{lemma}[Heap satisfaction downwards closure]
  \label{lem:heap-sat-dc}
  \begin{align*}
    &\forall \hs, n' \leq n, W \ldotp \\
    &\quad \heapSat[\hs]{n}{W} \Rightarrow \heapSat[\hs]{n'}{W}
  \end{align*}
\end{lemma} 
\begin{proof}[Proof of Lemma~\ref{lem:heap-sat-dc}]
  Let $\ms$, $n' \leq n$, and $W$ be given and assume
  \[
    \heapSat[\hs]{n}{W}
  \]
  This assumption gives us $P : \activeReg{W} \fun \MemSegments$ such that
  \begin{enumproof}
    \item $\hs = \biguplus_{r\in\activeReg{W}}P(r)$ \label{pf:heap-sat:prop1}
    \item \label{pf:heap-sat:prop2}
      \begin{align*}
        &\forall r \in \activeReg{W} \ldotp\\
        &\quad \exists H,s \ldotp\\
        &\qquad W(r) = (\_,s,\_,\_,H) \land \\
        &\qquad \npair[n']{P(r)} \in H(s)(\xi^{-1}(W))
      \end{align*}
  \end{enumproof}
  Using $P$ as witness, \ref{pf:heap-sat:prop1} is the first condition we need. Now let $r$ be given and use \ref{pf:heap-sat:prop2} to get $H$ and $s$ such that
  \begin{enumproof}[resume]
    \item $W(r) = (\_,s,\_,\_,H)$ 
    \item $\npair[n]{P(r)} \in H(s)(\xi^{-1}(W))$ \label{pf:heap-sat:interp}
  \end{enumproof}
We now need to show
\[
        \npair[n']{P(r)} \in H(s)(\xi^{-1}(W))
\]
which follows from \ref{pf:heap-sat:interp}, $n' \leq n$, and $H(s)(\xi^{-1}(W))$ is a $\UPred{\MemSegments}$.
\end{proof}

\begin{lemma}
  \label{lem:wl-local}
  If
  \begin{itemize}
  \item $\memSat{\ms}{W}$
  \item $\npair{\stdcap} \in \stdvr(W)$
  \item $\start \leq \addrend$
  \item $\perm \in \{\rwlx,\rwl\}$
  \end{itemize}
  then 
  \[
    \gl = \local
  \]
\end{lemma}
\begin{proof}[Proof of Lemma~\ref{lem:wl-local}]
It follows as a consequence of Lemma~\ref{lem:pwl-stack}. The $n$-equality forces the region to be $\temp$, so for the region name to be in $\localityReg(\gl,W)$, the locality must be $\local$.
\end{proof}

\subsubsection{Malloc safe to pass to adversary}
\label{sec:malloc-valrel}

\begin{lemma}[Safe values are safe to invoke.]
  \label{lem:safe-values-safe-invoke}
  If $\npair[n+1]{w} \in \stdvr(W)$, then $\npair{\updatePcPerm{w}} \in \stder(W)$.
\end{lemma}
\begin{proof}
  \begin{enumproof}
  \item Case $w = \stdcap$ and $\start \leq \addr \leq \addrend$
    and $\perm \in \{ \exec,\rwx, \rwlx \}$:
    \begin{enumproof}
    \item $\npair[n+1]{(\perm,\start,\addrend)}\in\execCond{\gl}(W)$.\\
      By: definition of $\stdvr(W)$ using the fact that $\perm \in \{ \exec,\rwx, \rwlx \}$.
    \item $\npair[n]{((\perm,\gl),\start,\addrend,\addr)} \in \stder(W)$:
      By definition of $\execCond{}$ using the fact that $\start \leq \addr \leq \addrend$.
    \end{enumproof}
  \item Case $w = \stdcap$ and $\start \leq \addr \leq \addrend$
    and $\perm = \entry$:
    \begin{enumproof}
    \item $\npair[n+1]{(\start,\addrend,\addr)}\in\entryCond{\gl}(W)$.\\
      By: definition of $\stdvr(W)$ using the fact that $\perm = \entry$.
    \item $\npair[n]{((\exec,\gl),\start,\addrend,\addr)} \in \stder(W)$:
      By definition of $\entryCond{}$ using the fact that $\start \leq \addr \leq \addrend$.
    \item $\updatePcPerm{w} = ((\exec,\gl),\start,\addrend,\addr)$:\\
      By definition of $\updatePcPerm{\cdot}$
    \end{enumproof}
  \item Otherwise:
    $\npair{\updatePcPerm{w}} \in \stder(W)$:\\
    By Lemma~\ref{lem:failed-obs-stder}.
  \end{enumproof}
\end{proof}

\begin{lemma}[Malloc is safe to pass to adversary]
  \label{lem:malloc-in-vr}
    For $c_\malloc$ that satisfies the specification for malloc with region $\iota_{\malloc,0}$,  if $W(r) \futurestr \iota_{\malloc,0}$, then
  $\npair{c_\malloc} \in \stdvr(W)$ for all $n$.
\end{lemma}
\begin{proof}
  \begin{enumproof}
  \item $c_\malloc = ((\entry,\glob),\base,\addrend,\addr)$.\\
    By: the malloc specification (Specification~\ref{spec:malloc}).
  \item Suffices: $\npair{(\start,\addrend,\addr)}\in\entryCond{\glob}(W)$.\\
    By definition of $\stdvr(W)$.
  \item Assume: $n' < n$, $W'\futurestr W$.\\
    Suffices: $\npair[n']{((\exec,\glob),\start,\addrend,\addr)} \in
    \stder(W')$.\\
    By: definition of the $\entryCond{}$
  \item Assume: $n''\leq n'$,
    $\npair[n'']{\reg}\in\stdrr(W')$, $\heapSat[\ms]{n''}{W'}$\\
    Suffices:
    $\npair[n'']{(\reg\update{\pcreg}{((\exec,\glob),\start,\addrend,\addr)},\ms)}\in\observations(W')$\\
    By: definition of $\stder(W')$
  \item Assume: $i < n''$,
    $(\reg\update{\pcreg}{((\exec,\glob),\start,\addrend,\addr)},\ms\uplus\ms_f)\step[i]
    (\halted,\mem')$\\
    Suffices: $\exists W''\futurestr W',\ms_r,\ms'\ldotp$ $\mem' =
    \ms'\uplus\ms_r\uplus\ms_f$ and $\heapSat[\ms']{n''-i}{W''}$\\
    By: definition of $\observations(W')$
  \item $W'(r) \futurestr \iota_{\malloc,0}$\\
    Easy from: $W'\futurestr W$ and $W(r) \futurestr \iota_{\malloc,0}$ using
    transitivity of $\futurestr$.
  \item $\exists P : \activeReg{W'} \rightarrow \HeapSegments \ldotp$
    $\memSatPar[n'']{\ms}{W'}{P}$, i.e. $\hs = \biguplus_{r\in\activeReg{W'}}P(r)$ and
    $\forall r \in \activeReg{W'} \ldotp$ $\exists H,s \ldotp$ $W'(r) =
    (\_,s,\_,\_,H)$ and $\npair[n'']{P(r)} \in H(s)(\xi^{-1}(W'))$\\
    By: definition of $\heapSat[\ms]{n''}{W'}$.
  \item Define $\ms_{\var{frame}} = \left(\biguplus_{r'\in\activeReg{W'},r'\neq
        r}P(r')\right)\uplus\ms_f$. Then $\ms\uplus\ms_f = P(r) \uplus
    \ms_{\var{frame}}$ and $\npair[n'']{P(r)} \in W'(r).H~(W'(r).s)~(\xi^{-1}(W'))$.
    Easy from the previous point.
  \item $\npair[n'']{P(r)} \in W'(r).H~(W'(r).s)~(\xi^{-1}([r\mapsto W'(r)]))$, i.e.
    $\heapSat[P(r)]{n''}{[r\mapsto W'(r)]}$.\\
    By: the malloc specification (Specification~\ref{spec:malloc}) from the
    previous point.
  \item Case: $\reg(r_1) \in \ints$ and $\reg(r_1) \geq 0$
    \begin{enumproof}
    \item Define $\var{size} = \reg(r_1)$
    \item $\exists \Phi' \in \ExecConfs, \ms_{\var{footprint}}',
      \ms_{\var{alloc}} \in \HeapSegments, j \in \nats, j > 0 \land b',e'\in
      \Addrs, \iota_\malloc' \in \Regions \ldotp$
      $(\reg\update{\pcreg}{((\exec,\glob),\start,\addrend,\addr)},\ms\uplus\ms_f)
      \step[j] \Phi'$ and $\memheap[\Phi']=\ms_{\var{footprint}}' \uplus
      \hs_{\var{alloc}} \uplus \ms_{\var{frame}}$ and $\iota_{\malloc}'
      \futurewk W'(r)$ and $\heapSat[\ms_{\var{footprint}}']{n''-j}{[r \mapsto
        \iota_\malloc']}$ and $\dom(\hs_{\var{alloc}}) = [b',e']$ and $\forall a
      \in [b',e']\ldotp \hs_{\var{alloc}}(a) = 0$ and $\memreg[\Phi'] =
      \memreg[\Phi]\update{\pcreg}{\updatePcPerm{w_{\var{ret}}}}\update{r_1}{((\rwx,\glob),b',e',b')}$
      and $\var{size} - 1 = e'-b' )$ with $w_{\var{ret}} = \memreg(r_1)$.\\
    By: the malloc specification (Specification~\ref{spec:malloc}).
    \item Define $W'' = W'\update{r}{\iota_\malloc'}[i\mapsto
      \iota^\nwl_{b',e'}]$ for $i \not\in\dom(W')$. We have that $W''\futurewk
      [r\mapsto\iota_\malloc']$ and $W'' \futurewk W'$.\\
      By: definition of $\futurewk$ , using the fact that
      $\iota_\malloc'\futurewk W(r)$.
    \item $\npair[n''']{(\start',\addrend')}\in\readCond{\glob}(W'')$ for all $n'''$:\\
      By: definition of $\readCond{}$, using the region $W''(i)$ and Lemma~\ref{lem:nwl-subset-pwl}.
    \item $\npair[n''']{(\start',\addrend')}\in\writeCond{\iota^\nwl,\glob}(W'')$ for all $n'''$:\\
      By: definition of $\writeCond{}$, using the region $W''(i)$.
    \item
      $\npair[n''']{(p,\start',\addrend')}\in\execCond{\iota^\nwl,\glob}(W'')$
      for all $n'''$, $p \in \{\rwx,\exec\}$:\\
      By: the definition of $\execCond{}, $ the FTLR (Theorem~\ref{thm:ftlr})
      using Lemmas~\ref{lem:writecond-mono-priv}, \ref{lem:readcond-mono-priv}
      and the previous two points.
    \item $\npair[n'']{((\rwx,\glob),b',e',b')}\in\stdvr(W'')$:\\
      By: definition of $\stdvr(W'')$ and the above three points.
    \item
      $\npair[n''-j]{\memreg[\Phi]\update{r_1}{((\rwx,\glob),b',e',b')}}\in\stdrr(W'')$:\\
      By Lemma~\ref{lem:stdrr-dc}, Lemma~\ref{lem:stdrr-mono} using the fact
      that $W''\futurewk W'$ and $\npair[n'']{\memreg}\in\stdvr(W')$, together
      with the previous point.
    \item $\npair[n''']{\ms_{\var{alloc}}} \in \iota^\nwl_{b',e'}.H~
      \iota^\nwl_{b',e'}.s~ W''$ for any $n'''$:\\
      By definition of $\iota^\nwl$, $H^\nwl$ and $\stdvr(\cdot)$ and the facts
      that $\dom(\hs_{\var{alloc}}) = [b',e']$ and $\forall a
      \in [b',e']\ldotp \hs_{\var{alloc}}(a) = 0$.
    \item Define $\ms' = \left(\biguplus_{r'\in\activeReg{W'},r'\neq
          r}P(r')\right)\uplus \ms_{\var{footprint}}'\uplus\ms_{\var{alloc}}$.
      Then
      $\memheap[\Phi'] = \ms' \uplus \ms_f$ and $\heapSat[\ms']{n''-j}{W''}$:\\
      By the facts that $\memheap[\Phi']=\ms_{\var{footprint}}' \uplus
      \hs_{\var{alloc}} \uplus \ms_{\var{frame}}$, $\ms_{\var{frame}} =
      \left(\biguplus_{r'\in\activeReg{W'},r'\neq r}P(r')\right)\uplus\ms_f$,
      the previous point, the facts that
      $\heapSat[\ms_{\var{footprint}}']{n''-j}{[r \mapsto \iota_\malloc']}$ and
      $W''\futurewk [r\mapsto\iota_\malloc']$, the facts that ($\forall r \in
      \activeReg{W'} \ldotp$ $\exists H,s \ldotp$ $W'(r) =, (\_,s,\_,\_,H)$ and
      $\npair[n'']{P(r)} \in H(s)(\xi^{-1}(W'))$) and $W'' \futurewk W'$ and the
      public monotonicity and downwards closedness of all regions, and finally
      the definition of $W''$.
    \item $\npair[n''-j+1]{w_{\var{ret}}} \in \stdvr(W'')$:\\
      By Lemma~\ref{lem:stdvr-mono-wk}, the fact that $W''\futurewk W'$,
      Lemma~\ref{lem:stdvr-dc}, and the fact that
      $\npair[n'']{w_{\var{ret}}}\in\stdvr(W')$, which follows from
      $w_{\var{ret}} = \memreg(r_1)$ and $\npair[n'']{\reg}\in\stdrr(W')$.
    \item $\npair[n''-j]{\updatePcPerm{w_{\var{ret}}}} \in \stder(W'')$:\\
      By Lemma~\ref{lem:safe-values-safe-invoke} from the previous point.
    \item
      $\npair[n''-j]{(\memreg[\Phi]\update{r_1}{((\rwx,\glob),b',e',b')}\update{\pcreg}{\updatePcPerm{w_{\var{ret}}}},
        \ms')} \in \observations(W'')$:\\
      By: definition of $\stder(W'')$, using the previous point and the facts
      that\\
      $\npair[n''-j]{\memreg[\Phi]\update{r_1}{((\rwx,\glob),b',e',b')}}\in\stdrr(W'')$,
      $\heapSat[\ms']{n''-j}{W''}$
    \item $i > j$ and $\Phi'\step[i-j]
      (\halted,\mem')$.\\
      By combining
      $(\reg\update{\pcreg}{((\exec,\glob),\start,\addrend,\addr)},\ms\uplus\ms_f)\step[i]
      (\halted,\mem')$ with
      $(\reg\update{\pcreg}{((\exec,\glob),\start,\addrend,\addr)},\ms\uplus\ms_f)
      \step[j] \Phi'$ using Lemma~\ref{lem:determinacy}.
    \item $\exists W''' \futurestr W'', \hs_r, \hs'' \ldotp$ $\heap' = \hs'' \uplus \hs_r \uplus \ms_f$ and $\heapSat[\hs'']{n-i}{W'''}$.\\
      By: definition of $\observations(W''')$ from the two previous points.
    \item $W''' \futurestr W'$: \\
      By Lemma~\ref{lem:future-priv-pub-trans}, using the previous point and the fact that $W'' \futurewk
      W'$.
    \end{enumproof}
  \item Case: $\reg(r_1)\not\in\ints \vee \reg(r_1) < 0$
    \begin{enumproof}
    \item $\exists j\ldotp (\reg\update{\pcreg}{((\exec,\glob),\start,\addrend,\addr)},\ms\uplus\ms_f)\step[j]
      \failed$\\
      By: the malloc specification (Specification~\ref{spec:malloc}).
    \item Contradiction with 
      $(\reg\update{\pcreg}{((\exec,\glob),\start,\addrend,\addr)},\ms\uplus\ms_f)\step[i]
      (\halted,\mem')$
    \end{enumproof}
  \end{enumproof}
\end{proof}

\subsubsection{Fundamental theorem of logical relations}

\begin{lemma}[Conditions for load instruction are sufficient]
  \label{lem:conds-load-suffice}
  If
  \begin{itemize}
  \item $\heapSat[\memheap]{n}{W}$
  \item $\var{c} = ((\perm,\gl),\start,\addrend,\addr)$
  \item $\npair{c}\in\stdvr(W)$
  \item $\readAllowed{\perm}$
  \item $\withinBounds{\var{c}}$
  \end{itemize}
  then $\npair[n-1]{\memheap(\addr)} \in \stdvr(W)$
\end{lemma}
\begin{proof}
  \begin{enumproof}
  \item $\npair{(\start,\addrend)} \in \readCond{\gl}(W)$: follows by
    definition of $\stdvr{\cdot}$ from $\npair{c}\in\stdvr(W)$.
  \item $\exists r \in \var{localityReg}(g,W)$, $[\start',\addrend'] \supseteq
    [\start,\addrend] \ldotp$ $W(r)\nsubsim[n]
    \iota_{\start',\addrend'}^\pwl$. By definition of $\readCond{\gl}(W)$. \label{step:unfold-readCond}
  \item $\exists P : \activeReg{W} \rightarrow \HeapSegments \ldotp
    \memSatPar{\memheap}{W}{P}$. By definition of $\heapSat[\memheap]{n}{W}$.
  \item $\memheap = \biguplus_{r\in\activeReg{W}}P(r)$ and $\forall r \in
    \activeReg{W}$, $\exists H,s \ldotp$ $W(r) = (\_,s,\_,\_,H)$
    and $\npair[n]{P(r)} \in H(s)(\xi^{-1}(W))$. By definition of
    $\memSatPar{\memheap}{W}{P}$. \label{step:unfold-phi-sat-part}
  \item $r \in \var{localityReg}(g,W)\subseteq \activeReg{W}$. By definition of
    $\var{localityReg}(\cdot)$ and $\activeReg{\cdot}$.
  \item $\exists H,s \ldotp$ $W(r) = (\_,s,\_,\_,H)$
    and $\npair[n]{P(r)} \in H(s)(\xi^{-1}(W))$.
    By specializing the result from
    Step~\ref{step:unfold-phi-sat-part} to the $r$ from
    Step~\ref{step:unfold-readCond}.
  \item $\npair[n]{P(r)} \in H^\pwl_{\start',\addrend'}(s)(\xi^{-1}(W))$.
    Follows by combining $\npair[n]{P(r)} \in H(s)(\xi^{-1}(W))$ with
    $W(r)\nsubsim[n] \iota_{\start',\addrend'}^\pwl$ from Step~\ref{step:unfold-readCond}.
  \item $\dom(P(r)) = [\start',\addrend']$ and for all $\addr' \in
    [\start',\addrend']\ldotp$
    $\npair[n-1]{P(r)(\addr')}\in\stdvr(\xi(\xi^{-1}(W)))$. By definition of $H^\pwl_{\start',\addrend'}$.
  \item $\addr \in [\start,\addrend]\subseteq[\start',\addrend']$. By combining
    $\withinBounds{c}$ with the fact that $[\start',\addrend'] \supseteq
    [\start,\addrend] \ldotp$ from Step~\ref{step:unfold-readCond}.
  \item In particular, we get: $\memheap(\addr) = P(r)(\addr)$ and
    $\npair[n-1]{P(r)(\addr)}\in\stdvr(W)$.
  \end{enumproof}
\end{proof}

 \begin{lemma}[Conditions for lea instruction are sufficient]
   \label{lem:conds-lea-suffice}
   If
   \begin{itemize}
   \item $\npair{((\perm,\gl),\start,\addrend,\addr)}\in\stdvr(W)$
   \item $\perm \neq \entry$
   \end{itemize}
 
   then $\npair{((\perm,\gl),\start,\addrend,\addr')} \in \stdvr(W)$
 \end{lemma}
 \begin{proof}
   Follows by inspection of the cases in the definition of $\stdvr(W)$: $a$ is
   ignored in all cases except where $\perm = \entry$.
 \end{proof}
 
 \begin{lemma}[pwl writecond implies nwl]
   \label{lem:pwl-writecond-implies-nwl}
   If $\npair{(\start,\addrend)} \in \writeCond{}(\iota^\pwl,\gl)(W)$ then
   $\npair{(\start,\addrend)} \in \writeCond{}(\iota^\nwl,\gl)(W) \}$.
 \end{lemma}
 \begin{proof}
   \begin{enumproof}
   \item $\exists r \in \var{localityReg}(g,W) \ldotp$ $\exists
     [\start',\addrend'] \supseteq [\start,\addrend] \ldotp$ $W(r)\nsupsim[n-1]
     \iota^\pwl_{\start',\addrend'}$ and $W(r)$ is address-stratified: by definition of $\writeCond{}$.
   \item Suffices: $W(r)\nsupsim[n-1] \iota^\nwl_{\start',\addrend'}$. By
     definition of $\writeCond{}$
   \item $W(r)\nsupsim[n-1] \iota^\pwl_{\start',\addrend'} \nsupsim[n-1]
     \iota^\nwl_{\start',\addrend'}$: follows by Lemma~\ref{lem:nwl-subset-pwl}.
   \end{enumproof}
 \end{proof}
 
 \begin{lemma}[execCond implies entryCond]
   \label{lem:execCond-implies-entryCond}
   If $\npair{(\exec,\start,\addrend)} \in \execCond{}(\gl)(W)$ then
     $\npair{(\start,\addrend,\addr)} \in \entryCond{}(\gl)(W)$.
 \end{lemma}
 \begin{proof}
   \begin{enumproof}
   \item Assume: $n' < n$, $W' \future W$ where $\gl = \local \Rightarrow
     \future = \futurewk$ and $\gl = \glob \Rightarrow \future = \futurestr$
     \\
     Suffices: $\npair[n']{((\exec,\gl),\start,\addrend,\addr)} \in \stder(W')$
   \item Case $\addr \in [\start,\addrend]$: Follows from the definition of
     $\execCond{}$.
   \item Case $\addr \not\in[\start,\addrend]$: Follows by
     Lemma~\ref{lem:failed-obs-stder}.
   \end{enumproof}
 \end{proof}
 
 \begin{lemma}[Conditions for restrict instruction are sufficient]
   \label{lem:conds-restrict-suffice}
   If
   \begin{itemize}
   \item $\npair{((\perm,\gl),\start,\addrend,\addr)}\in\stdvr(W)$
   \item $(\perm',\gl')\sqsubseteq (\perm,\gl)$
   \end{itemize}
 
   then  $\npair{((\perm',\gl'),\start,\addrend,\addr)}\in \stdvr(W)$
 \end{lemma}
 \begin{proof}
   By inspection of the definition of $\stdvr(W)$, everything follows trivially
   except the following.

   \begin{enumproof}
   \item If $\npair{(\start,\addrend)} \in \writeCond{}(\iota^\pwl,\gl)(W)$ then
     $\npair{(\start,\addrend)} \in \writeCond{}(\iota^\nwl,\gl)(W)$: holds
     by lemma~\ref{lem:pwl-writecond-implies-nwl}.
   \item If $\npair{(\exec,\start,\addrend)} \in \execCond{}(\gl)(W)$ then
     $\npair{(\start,\addrend,\addr)} \in \entryCond{}(\gl)(W)$.
   \end{enumproof}
 \end{proof}
 
 \begin{lemma}[Conditions for subseg instruction are sufficient]
   \label{lem:conds-subseg-suffice}
   If
   \begin{itemize}
   \item $\npair{((\perm,\gl),\start,\addrend,\addr)} \in \stdvr(W)$
   \item $\start \leq \start'$
   \item $\addrend' \leq \addrend$
   \item $\perm \neq \entry$
   \end{itemize}
 
   then $\npair{((\perm,\gl),\start',\addrend',\addr)} \in \stdvr(W)$
 \end{lemma}
 \begin{proof}
   Follows easily from the definitions of $\stdvr(W)$, $\readCond{}$,
   $\writeCond{}$, $\execCond{}$.
 \end{proof}

 \begin{lemma}[Conditions for store instruction are sufficient]
   \label{lem:conds-store-suff}
   If 
   \begin{itemize}
   \item $\ms = \ms' \uplus \ms_f$
   \item $\heapSat[\ms']{n}{W}$
   \item $((\perm,\gl),\start,\addrend,\addr) = c$
   \item $\npair{c}\in\stdvr(W)$
   \item $\writeAllowed{\perm}$
   \item $\withinBounds{\var{c}}$
   \item $\npair{\var{w}}\in\stdvr(W)$
   \item if $\var{w} = ((\_,\local),\_,\_,\_)$, then $\perm \in
     \{\rwlx,\readwritel \}$
   \end{itemize}
 
   then $\addr \in \dom(\ms')$ (i.e. $\ms\update{a}{w} =
   \ms'\update{a}{w}\uplus\ms_f$) and
   $\heapSat[{\ms'\update{\addr}{\var{w}}}]{n}{W}$
 \end{lemma}
 \begin{proof}
   \begin{enumproof}
   \item $\npair{(\start,\addrend)} \in \writeCond{}(\iota,\gl)(W)$ where $\iota
     = \iota^\pwl$ or $\iota = \iota^\nwl$ and (if $\var{w} =
     ((\_,\local),\_,\_,\_)$, then $\iota = \iota^\pwl$).

     By definition of $\stdvr(W)$ and $\writeAllowed{}$, from
     $\npair{c}\in\stdvr(W)$, $((\perm,\gl),\start,\addrend,\addr) = c$ and
     $\writeAllowed{\perm}$ and the fact that (if $\var{w} =
     ((\_,\local),\_,\_,\_)$, then $\perm \in \{\rwlx,\readwritel \}$)
   \item $\exists r \in \var{localityReg}(\gl,W) \ldotp$ $\exists
     [\start',\addrend'] \supseteq [\start,\addrend] \ldotp$ $W(r)\nsupsim[n-1]
     \iota_{\start',\addrend'}$ and $W(r)$ is address-stratified. By definition
     of $\writeCond{}$.
   \item $\exists P : \activeReg{W} \rightarrow \HeapSegments \ldotp$
     $\memSatPar{\ms'}{W}{P}$. By definition of $\heapSat[\ms']{n}{W}$.
   \item $\ms' = \biguplus_{r\in\activeReg{W}}P(r)$ and $\forall r \in
     \activeReg{W} \ldotp$ $\exists H,s \ldotp$ $W(r) = (\_,s,\_,\_,H)$ and
     $\npair[n]{P(r)} \in H(s)(\xi^{-1}(W))$. By definition of
     $\memSatPar{\ms'}{W}{P}$.
   \item $\exists H,s \ldotp$ $W(r) = (\_,s,\_,\_,H)$ and
     $\npair[n]{P(r)} \in H(s)(\xi^{-1}(W))$. By instantiating the previous
     point to the $r$ from the $\writeCond{}$.
   \item $\npair{w} \in \iota.H~(\iota.s)~(\xi^{-1}(W))$ by definition of
     $\iota^\pwl$, $\iota^\nwl$ and the fact that (if $\var{w} =
     ((\_,\local),\_,\_,\_)$, then $\iota = \iota^\pwl$).
   \item Define $\ms'_w$ such that $\dom(\ms'_w) = [\start',\addrend']$,
     $\ms'_w(\addr) = w$ and $\ms'_w(\addr') = 0$ for $\addr' \neq \addr$. It's
     easy to show from the previous point that $\npair{\ms'_w} \in
     H(s)(\xi^{-1}(W))$.
   \item $\dom(P(r)) = \dom(\ms'_w) = [\start',\addrend'] \ni \addr$ and
     $\npair{P(r)\update{\addr}{w}} \in H(s)(\xi^{-1}(W))$ by applying the fact
     that $W(r)$ is address-stratified, combined with the previous point.
   \item Define $P'(r) = P(r)\update{a}{w}$ and $P'(r') = P(r')$ for $r' \neq r$.
   \item $\ms'\update{a}{w} = \biguplus_{r\in\activeReg{W}}
     P'(r)$ and $\ms'\update{a}{w} :_{n,P'} W$. By definition of
     $\memSatPar{\ms'}{W}{P}$ and the previous two points.
   \end{enumproof}
 \end{proof}
  
\begin{theorem}[Fundamental theorem of logical relations]
  \label{thm:ftlr}
  For all $n$, $\perm$, $\start$, $\addrend$, $\addr$, $\gl$, $W$  \\
  If one of the following holds:
  \begin{itemize}
  \item \[
      \begin{gathered}
        \perm = \exec \land\\
        \npair{(\start,\addrend)} \in \readCond{}(\gl)(W)
      \end{gathered}
    \]
  \item \[
      \begin{gathered}
        \perm = \rwx \land \\
        \npair{(\start,\addrend)} \in \readCond{}(\gl)(W) \land\\
        \npair{(\start,\addrend)} \in \writeCond{}(\iota^\nwl,\gl)(W)
      \end{gathered}
    \]
  \item \[
      \begin{gathered}
        \perm = \rwlx \land\\
        \npair{(\start,\addrend)} \in \readCond{}(\gl)(W) \land\\
        \npair{(\start,\addrend)} \in \writeCond{}(\iota^\pwl,\gl)(W),
      \end{gathered}
    \]
  \end{itemize}
  then
  \[
    \npair{((\perm,\gl),\start,\addrend,\addr)} \in \stder(W)
  \]
\end{theorem}
\begin{proof}
  \begin{enumproof}
  \item By induction on $n$. In other words, assume that the theorem already
    holds for all $n' < n$.
  \item Assume: $n' \leq n$, $\npair[n']{\reg}\in \stdrr(W)$, $\heapSat[\hs]{n'}{W}$.\\
    Suffices: $\npair[n']{(\reg\update{\pcreg}{((\perm,\gl),\start,\addrend,\addr)},\hs)}\in\observations(W)$.\\
    By: definition of $\stder(W)$.
  \item Assume: $\ms_f$, $\heap'$, $i \leq n'$, $\Phi =
    (\reg\update{\pcreg}{((\perm,\gl),\start,\addrend,\addr)},\hs \uplus \ms_f)$
    and $\Phi \step[i] (\halted,\mem')$,\\
    Suffices: $\exists W' \futurestr W$, $\hs_r$, $\hs' \ldotp$
    $\mem' = \hs' \uplus \hs_r \uplus \ms_f$ and $\heapSat[\hs']{n'-i}{W'}$\\
    By: definition of $\observations(W)$ \label{suff-after-eval}
  \item $i \neq 0$, since
    $(\reg\update{\pcreg}{((\perm,\gl),\start,\addrend,\addr)},\hs \uplus
    \ms_f)\neq (\halted,\mem')$ for any $\mem'$. \\
    Therefore, assume w.l.o.g. that $i = 1+i'$,
    \begin{equation*}
      \Phi \step \var{conf'} \step[i'] (\halted,\mem')
    \end{equation*} \label{step:ip-non-zero}
  \item $n \geq n' > 0$, since otherwise $i = 0$ (because $i \leq n'\leq n$) and this is
    impossible by the previous point.
  \item $\npair[n']{\memreg(\pcreg)}\in\stdvr(W)$. \label{step:ftlr-pc-vr}
    Proof:
    \begin{enumproof}
    \item Assume: $\perm' \in \{\exec,\rwx,\rwlx\}$ with $\perm'
      \sqsubseteq\perm$\\
      Suffices: $\npair[n']{(\perm',\start,\addrend)}\in\execCond{g}(W)$\\
      By: the definition of $\stdvr(\cdot)$ using the assumptions
    \item Assume: $n''<n'$, $W'\future W$, $\addr' \in [\start,\addrend]$, $\gl
      = \local \Rightarrow \future = \futurewk$,
      $\gl=\glob\Rightarrow\future = \futurestr$.\\
      Suffices: $\npair[n'']{((\perm,\gl),\start,\addrend,\addr')} \in
      \stder(W')$.
      By: definition of $\execCond{g}(W)$
    \item By induction, using the assumptions and Lemmas~\ref{lem:readcond-dc}
      and~\ref{lem:writecond-dc}.
    \end{enumproof}
  \item For all $r \in \RegName$, $\npair[n']{\memreg(r)}\in\stdvr(W)$. \label{step:ftlr-reg-rr}
    \begin{enumproof}
    \item Case $r \neq \pcreg$: follows from $\npair[n']{\reg}\in \stdrr(W)$ by
      definition of $\stdrr(W)$.
    \item Case $r = \pcreg$: by step~\ref{step:ftlr-pc-vr}.
    \end{enumproof}
  \item By inspection of the definitions of $\Phi\step \var{conf'}$ and
    $\sem{\decode(\memheap(\addr))}$ and $\updatePcPerm{\cdot}$ and
    $\stdUpdatePc{\cdot}$, it is easy to see that one of the following cases
    must hold:
  \item Case $\var{conf'} = \failed$: contradiction, since it is not possible
    that $\failed \step[i'] (\halted,\mem')$. 
  \item Case $\var{conf'} = (\halted,\mem)$:
    \begin{enumproof}
    \item Then $i' = 0$ and $\mem' =
      \mem$\\
      Follows from $(\halted, \mem) \step[i'] (\halted,\mem)$
    \item For $W' = W$, $\ms_r = \emptyset$ and $\ms' = \ms$, we have that
      $\mem = \hs' \uplus \hs_r \uplus \ms_f$ and
      $\heapSat[\hs']{n'-1}{W'}$ (using Lemma~\ref{lem:heap-sat-dc}).
    \end{enumproof}
  \item Case $\var{conf'} = \updateReg[\Phi'']{\pcreg}{\var{newPc}}$, and additionally, one of the following
    holds:
    \begin{itemize}
    \item $\memheap[\Phi''] = \memheap[\Phi]$
    \item $\memheap[\Phi''] = \memheap[\Phi]\update{\addr'}{\var{w}}$, with
      $\memreg[\Phi](r_1) = ((\perm',\gl'),\start',\addrend',\addr') = c$ and
      $\writeAllowed{\perm'}$ and $\withinBounds{\var{c}}$ and $\var{w} =
      \memreg(r_2)$ and if $\var{w} = ((\_,\local),\_,\_,\_)$, then $\perm' \in
      \{\rwlx,\readwritel \}$
    \end{itemize}
    and also one of the following holds:
    \begin{itemize}
    \item $\var{newPc}= \updatePcPerm{\memreg(\lv)}$
    \item $\var{newPc} = ((\perm',\gl'),\start',\addrend',\addr' + 1)$ and
      $\memreg[\Phi](\pcreg) = ((\perm',\gl'),\start',\addrend',\addr')$
    \end{itemize}
    and finally, for all $r \in \RegName$, one of the following
    holds:
    \begin{itemize}
    \item $\memreg[\Phi''](r) = \memreg[\Phi](r)$
    \item $\memreg[\Phi''](r) = z$ for some $z \in \ints$
    \item $\memreg[\Phi''](r) = \var{w}$ and $\memreg(r_2) =
      ((\perm',\gl'),\start',\addrend',\addr') = \var{c}$ and
      $\readAllowed{\perm'}$ and $\withinBounds{\var{c}}$ and $\var{w} =
      \memheap(\addr')$
    \item $\memreg[\Phi''](r) = \var{c}$ and $\memreg(r_1) =
      ((\perm',\gl'),\start',\addrend',\addr')$ and $\perm' \neq \entry$ and
      $\var{c} = ((\perm',\gl'),\start',\addrend',\addr' + z)$ for some $z \in
      \ints$
    \item $\memreg[\Phi''](r) = \var{c}$ and $\memreg(r) =
      ((\perm',\gl'),\start',\addrend',\addr')$ and $(\perm'',\gl'')\sqsubseteq
      (\perm',\gl')$ and $c = ((\perm'',\gl''),\start',\addrend',\addr')$
    \item $\memreg[\Phi''](r) = \var{c}$ and $\memreg(r) =
      ((\perm',\gl'),\start',\addrend',\addr')$ and $\start' \leq \start''$ and
      $\addrend'' \leq \addrend'$ and $c =
      ((\perm',\gl'),\start'',\addrend'',\addr')$ and $\perm' \neq \entry$
    \end{itemize}
    In this case, we have:
    \begin{enumproof}
    \item $\memheap[\Phi''] = \ms'' \uplus \ms_f$ and $\heapSat[\ms'']{n'-1}{W}$.
      \begin{enumproof}
      \item Case $\memheap[\Phi''] = \memheap[\Phi]$:
        Then $\memheap[\Phi''] = \ms\uplus\ms_f$ and
        $\heapSat[\ms]{n'-1}{W}$ follows by Lemma~\ref{lem:heap-sat-dc}.
      \item Case $\memheap[\Phi''] = \memheap[\Phi]\update{\addr'}{\var{w}}$,
        with $\memreg[\Phi](r_1) = ((\perm',\gl'),\start',\addrend',\addr') = c$
        and $\writeAllowed{\perm'}$ and $\withinBounds{\var{c}}$ and $\var{w} =
        \memreg(r_2)$ and if $\var{w} = ((\_,\local),\_,\_,\_)$, then $\perm'
        \in \{\rwlx,\readwritel \}$.\\
        The facts that $\memheap[\Phi''] = \ms''\uplus\ms_f$ and
        $\heapSat[\ms'']{n'-1}{W}$ follow by Lemmas~\ref{lem:conds-store-suff}
        and~\ref{lem:heap-sat-dc} using the fact that $\heapSat[\ms]{n'}{W}$ and
        $\npair[n']{\memreg(r_1)}\in\stdvr(W)$ and
        $\npair[n']{\memreg(r_2)}\in\stdvr(W)$ which follows from
        Step~\ref{step:ftlr-reg-rr}.
      \end{enumproof}
    \item For all $r \in \RegName$, $\npair[n'-1]{\memreg[\Phi''](r)}\in\stdvr(W)$.\label{step:phip-regs-stdvr}
      \begin{enumproof}
      \item Case $\memreg[\Phi''](r) = \memreg[\Phi](r)$:
        $\npair[n'-1]{\memreg[\Phi''](r)} \in \stdvr(W)$ follows from
        Step~\ref{step:ftlr-reg-rr} using Lemma~\ref{lem:stdvr-dc}.
      \item $\memreg[\Phi''](r) = z$ for some $z \in \ints$.
        $\npair[n'-1]{\memreg[\Phi''](r)} \in \stdvr(W)$ follows by definition of
        $\stdvr(\cdot)$
      \item $\memreg[\Phi''](r) = \var{w}$ and $\memreg(r_2) =
        ((\perm',\gl'),\start',\addrend',\addr') = \var{c}$ and
        $\readAllowed{\perm'}$ and $\withinBounds{\var{c}}$ and $\var{w} =
        \memheap(\addr')$:\\
        $\npair[n'-1]{\memreg[\Phi''](r)} \in \stdvr(W)$ follows by
        Lemmas~\ref{lem:conds-load-suffice} using the fact that
        $\heapSat[\memheap]{n'}{W}$ and $\npair[n']{\memreg(r_2)}\in\stdvr(W)$ which
        we have from step~\ref{step:ftlr-reg-rr}.
      \item $\memreg[\Phi''](r) = \var{c}$ and $\memreg(r_1) =
        ((\perm',\gl'),\start',\addrend',\addr')$ and $\perm' \neq \entry$ and
        $\var{c} = ((\perm',\gl'),\start',\addrend',\addr' + z)$ for some $z \in
        \ints$: \\
        $\npair[n'-1]{\memreg[\Phi''](r)} \in \stdvr(W)$ follows by
        Lemmas~\ref{lem:conds-lea-suffice} and~\ref{lem:stdvr-dc} using the
        fact that $\npair[n']{\memreg(r_1)}\in\stdvr(W)$
        which we have from step~\ref{step:ftlr-reg-rr}.
      \item $\memreg[\Phi''](r) = \var{c}$ and $\memreg(r) =
        ((\perm',\gl'),\start',\addrend',\addr')$ and $(\perm'',\gl'')\sqsubseteq
        (\perm',\gl')$ and $c = ((\perm'',\gl''),\start',\addrend',\addr')$: \\
        $\npair[n'-1]{\memreg[\Phi''](r)} \in \stdvr(W)$ follows by
        Lemmas~\ref{lem:conds-restrict-suffice} and~\ref{lem:stdvr-dc} using the
        fact that $\npair[n']{\memreg(r)}\in\stdvr(W)$
        which follows from $\npair[n']{\memreg}\in\stdrr(W)$ by
        definition.
      \item $\memreg[\Phi''](r) = \var{c}$ and $\memreg(r) =
        ((\perm',\gl'),\start',\addrend',\addr')$ and $\start' \leq \start''$ and
        $\addrend'' \leq \addrend'$ and $c =
        ((\perm',\gl'),\start'',\addrend'',\addr')$ and $\perm' \neq \entry$: \\
        $\npair[n'-1]{\memreg[\Phi''](r)} \in \stdvr(W)$ follows by
        Lemmas~\ref{lem:conds-subseg-suffice} and~\ref{lem:stdvr-dc} using the
        fact that $\npair[n']{\memreg(r)}\in\stdvr(W)$
        which follows from $\npair[n']{\memreg}\in\stdrr(W)$ by definition.
      \end{enumproof}
    \item $\npair[n'-1]{\memreg[\Phi'']}\in\stdrr(W)$: Follows from the previous
      point by definition of $\stdrr(W)$.
    \item $\npair[n'-1]{\var{newPc}} \in \stder(W)$:
      \begin{enumproof}
      \item Case $\var{newPc}= \updatePcPerm{\memreg(\lv)}$: We distinguish the
        following cases:
        \begin{enumproof}
        \item Case $\memreg(\lv) = ((\entry,\gl'),\start',\addrend',\addr')$:
          \begin{enumproof}
          \item $\npair[n']{\memreg(\lv)} \in \stdvr(W)$.  Follows from
            Step~\ref{step:ftlr-reg-rr}.
          \item $\npair[n']{(\start',\addrend',addr')} \in \entryCond{\gl'}(W)$.
            By definition of $\stdvr(W)$ from the previous point.
          \item $\npair[n'-1]{((\exec,\gl'),\start',\addrend',\addr')} \in
            \stder(W)$: By definition of $\entryCond{\cdot}$ and taking $n' =
            n'-1$ and $W' = W$
          \item $\updatePcPerm{\memreg(\lv)} =
            ((\exec,\gl'),\start',\addrend',\addr')$: by definition of $\updatePcPerm{\cdot}$.
          \end{enumproof}
        \item Case $\memreg(\lv) = ((\perm',\gl'),\start',\addrend',\addr')$
          with $\perm' \in \{\exec,\rwx,\rwlx\}$ and $\withinBounds{\memreg(\lv)}$:
          \begin{enumproof}
          \item $\npair[n']{\memreg(\lv)} \in \stdvr(W)$.  Follows from
            Step~\ref{step:ftlr-reg-rr}.
          \item $\npair[n']{(\perm',\start',\addrend',\addr')} \in \execCond{\gl'}(W)$.
            By definition of $\stdvr(W)$ from the previous point.
          \item $\npair[n'-1]{((\perm',\gl'),\start',\addrend',\addr')} \in
            \stder(W)$: By definition of $\execCond{\cdot}$, taking $n' = n'-1$,
            $W' = W$ and $\addr = \addr'$. Note that $\addr' \in
            [\start',\addrend']$ because we have $\withinBounds{\memreg(\lv)}$.
          \item $\updatePcPerm{\memreg(\lv)} =
            ((\perm',\gl'),\start',\addrend',\addr')$: by definition of
            $\updatePcPerm{\cdot}$.
          \end{enumproof}
        \item Case not ($\memreg(\lv) =
          ((\entry,\gl'),\start',\addrend',\addr')$) and not ($\memreg(\lv) =
          ((\perm',\gl'),\start',\addrend',\addr')$ with $\perm' \in
          \{\exec,\rwx,\rwlx\}$ and $\withinBounds{\memreg(\lv)}$):
          \begin{enumproof}
          \item $\updatePcPerm{\memreg(\lv)} = \memreg(\lv)$: by definition of 
            $\updatePcPerm{\cdot}$.
          \item $(\reg\update{\pcreg}{\memreg(\lv)},\ms)\step\failed$ for any
            $\reg$,$\ms$: by definition of the evaluation relation.
          \item $\npair[n'-1]{\var{newPc}} \in \stder(W)$: by
            Lemma~\ref{lem:failed-obs-stder} using the previous point.
          \end{enumproof}
        \end{enumproof}
      \item Case $\var{newPc} = ((\perm',\gl'),\start',\addrend',\addr' + 1)$ and
        $\memreg[\Phi''](\pcreg) = ((\perm',\gl'),\start',\addrend',\addr')$:
        \begin{enumproof}
        \item Case $\perm' \in \{\exec,\rwx,\rwlx\}$ and $\start'\leq
          \addr'+1\leq \addrend'$:
          \begin{enumproof}
          \item $\npair[n'-1]{\memreg[\Phi''](\pcreg)} \in \stdvr(W)$: by
            Step~\ref{step:phip-regs-stdvr}.
          \item $\npair[n'-1]{((\perm',\gl'),\start',\addrend',\addr' + 1)} \in \stdvr(W)$: by
            Lemma~\ref{lem:conds-lea-suffice} from the previous point.
          \item One of the following holds:
            \begin{itemize}
            \item \[
                \begin{gathered}
                  \perm' = \exec \land\\
                  \npair[n'-1]{(\start',\addrend')} \in \readCond{}(\gl)(W)
                \end{gathered}
              \]
            \item \[
                \begin{gathered}
                  \perm' = \rwx \land \\
                  \npair[n'-1]{(\start',\addrend')} \in \readCond{}(\gl)(W) \land\\
                  \npair[n'-1]{(\start',\addrend')} \in \writeCond{}(\iota^\nwl,\gl)(W)
                \end{gathered}
              \]
            \item \[
                \begin{gathered}
                  \perm' = \rwlx \land\\
                  \npair[n'-1]{(\start',\addrend')} \in \readCond{}(\gl)(W) \land\\
                  \npair[n'-1]{(\start',\addrend')} \in \writeCond{}(\iota^\pwl,\gl)(W),
                \end{gathered}
              \]
            \end{itemize}
           This follows from the previous point by  definition of $\stdvr(W)$
         \item $\npair[n'-1]{((\perm',\gl'),\start',\addrend',\addr' + 1)} \in
           \stder(W)$: By the induction hypothesis of this lemma using the
           previous point.
          \end{enumproof}
        \item Case not ($\perm' \in \{\exec,\rwx,\rwlx\}$ and $\start'\leq
          \addr'+1\leq \addrend'$):
          The result follows by Lemma~\ref{lem:failed-obs-stder}.
        \end{enumproof}
      \end{enumproof}
    \item
      $\npair[n'-1]{(\memreg[\Phi'']\update{\pcreg}{\var{newPc}},\ms'')}\in\observations(W)$:
      by definition of $\stder(W)$ using the above three points.
    \item $\exists W' \futurestr W$, $\hs_r$, $\hs' \ldotp$
      $\mem = \hs' \uplus \hs_r \uplus \ms_f$ and $\heapSat[\hs']{n'-i}{W'}$\\
      By: definition of $\observations(W)$ using the previous step and the
      evaluation $\var{conf'}\step[i'] (\halted,\mem')$ from
      Step~\ref{step:ip-non-zero}.
    \end{enumproof}
  \end{enumproof}
\end{proof}

\subsubsection{Scall macro-instruction correctness}
\begin{definition}
  We say that $(\reg,\ms) \text{ is looking at } [i_0,\cdots,i_n] \text{ followed by } c_{\mathit{next}}$ 
  iff
  \begin{itemize}
  \item $\reg(\pcreg) = ((p,g),b,e,a)$
  \item $p = \rwx$, $p = \exec$, or $p = \rwlx$
  \item $a+n\leq e$, $b\leq a\leq e$
  \item $\ms(a+0,\cdots,a+n) = [i_0,\cdots,i_n]$
  \item $c_{\mathit{next}} = ((p,g),b,e,a+n+1)$
  \end{itemize}
\end{definition}

\begin{definition}
  We say that $\reg \text{ points to stack with $\ms_\stk$ used and $\ms_{\mathit{unused}}$ unused}$
  iff
  \begin{itemize}
  \item $\reg(r_\stk) =((\rwlx,\local),b_\stk,e_\stk,a_\stk)$
  \item $\dom(\ms_{\mathit{unused}}) = [a_\stk+1,\cdots,e_\stk]$
  \item $\dom(\ms_\stk) = [b_\stk,\cdots,a_\stk]$ \lau{Maybe make it clear what happens when $\ms_\stk$ is empty}
  \item $b_\stk - 1\leq a_\stk$
  \end{itemize}
\end{definition}

\begin{lemma}[$\mathtt{scall}$ works]
  \label{lem:scall-works}
  If
  \begin{itemize}
  \item $\memSat[n]{\ms}{\revokeTemp{W}}$ 
  \item $\dom(\ms_f) \cap (\dom(\ms_\stk \uplus \ms_{\mathit{unused}} \uplus \ms)) = \emptyset$
  \item $(\reg,\ms) \text{ is looking at }
    \mathtt{scall}\;r(\overline{r_{\mathit{arg}}},
    \overline{r_{\mathit{priv}}}) \text{ followed by } c_{\mathit{next}}$
  \item $\reg \text{ points to stack with $\ms_\stk$ used and $\ms_{\mathit{unused}}$ unused}$
  \item[Hyp-Callee] If
    \begin{itemize}
    \item $\dom(\ms_{\mathit{unused}}) = \dom(\ms_{\mathit{act}}
      \uplus \ms_{\mathit{unused}}')$,
    \item $W' =
      \revokeTemp{W}[\iota^{\sta}(\temp,\ms_\stk\uplus\ms_{\mathit{act}} \uplus \ms_f),\iota^{\pwl}(\dom(\ms_{\mathit{unused}}'))]$,
    \item $\memSat[n-1]{\ms''}{W'}$
    \item $\reg' \text{ points to stack with $\emptyset$ used and $\ms_{\mathit{unused}}'$ unused}$
    \item $\reg'= \reg_0[\pcreg\mapsto\updatePcPerm{\reg(r)},
      \overline{r_{\mathit{arg}}} \mapsto \reg(\overline{r_{\mathit{arg}}}),r_0
      \mapsto c_{\mathit{ret}}, r_\stk \mapsto c_\stk, r \mapsto \reg(r)]$ 
    \item $\npair[n-1]{c_{\mathit{ret}}} \in \stdvr(W')$
    \item $\npair[n-1]{c_\stk} \in \stdvr(W')$
    \end{itemize}

    then we have that $\npair[n-1]{(\reg',\ms'')} \in \observations(W')$
  \item[Hyp-Cont] If
    \begin{itemize}
    \item $n' \leq n-2$
    \item $W'' \futurewk \revokeTemp{W}$
    \item $\memSat[n']{\ms''}{\revokeTemp{W''}}$ 
    \item for all $r$, we have that:
      \begin{equation*}
        \reg'(r)
        \begin{cases}
          = c_{\mathit{next}} &\text{ if } r = \pcreg\\
          = \reg(r)&\text{ if } r \in \overline{r_{\mathit{priv}}}\\
          \in \stdvr(\revokeTemp{W''}) &\text{ if $\reg'(r)$ is a global capability and } r \not\in \{\pcreg,\overline{r_{\mathit{priv}}}, r_\stk\}
        \end{cases}
      \end{equation*}
    \item $\reg' \text{ points to stack with $\ms_\stk$ used and $\ms_{\mathit{unused}}''$ unused}$ for some $\ms_{\mathit{unused}}''$
    \end{itemize}

    then we have that $\npair[n']{(\reg',\ms'' \uplus \ms_f \uplus \ms_\stk \uplus \ms_{\var{unused}}'')} \in \observations(W'')$
  \end{itemize}
  Then 
  \begin{itemize}
    \item $\npair{(\reg,\ms \uplus \ms_f \uplus \ms_\stk \uplus \ms_{\var{unused}})} \in \observations(W)$
  \end{itemize}

\end{lemma}
\begin{proof}
  Assume $n$ is sufficiently large to execute all the steps up to and including the jump of $\mathtt{scall}\; r(\overline{r_{\var{arg}}},\overline{r_{\var{priv}}})$. If this is not the case, then in any given memory frame the execution will not halt successfully fast enough.

  Further assume
  \begin{enumproof}
  \item $\memSat[n]{\ms}{\revokeTemp{W}}$ \label{pf:scall-works:mem-sat-ms}
  \item $\dom(\ms_f) \cap (\dom(\ms_\stk \uplus \ms_{\mathit{unused}} \uplus \ms)) = \emptyset$
  \item $(\reg,\ms) \text{ is looking at }
    \mathtt{scall}\;r(\overline{r_{\mathit{arg}}},
    \overline{r_{\mathit{priv}}}) \text{ followed by } c_{\mathit{next}}$
  \item $\reg \text{ points to stack with $\ms_\stk$ used and $\ms_{\mathit{unused}}$ unused}$
  \item Hyp-Callee \label{pf:scall-works:item:Hyp-Callee}
  \item Hyp-Cont\label{pf:scall-works:item:Hyp-Cont}
  \end{enumproof}
  Now we wish to apply Lemma~\ref{lem:anti-red-obs}. To this end let $\ms_{\var{frame}}$ be given. Executing the \texttt{scall} gives us
  \[
    (\reg,\ms \uplus \ms_f \uplus \ms_\stk \uplus \ms_{\var{unused}} \uplus \ms_{\var{frame}}) \step[i] (\reg_1,\ms \uplus \ms_f \uplus \ms_\stk \uplus \ms_{\var{act}} \uplus \ms_\unused' \uplus \ms_{\var{frame}})
  \]
  where
  \begin{enumproof}[resume]
  \item $i \leq n$
  \item $\ms_\act$ contains activation record, $\reg(\overline{r_{\var{priv}}})$, the code return capability, and the full stack capability ($\reg(r_\stk)$ with the pointer adjusted).
  \item $\forall a \in \dom(\ms_\unused') \ldotp \ms_\unused'(a) = 0$
  \item $\dom(\ms_\unused) = \dom(\ms_\act \uplus \ms_\unused')$
  \item $\reg_1(r_0) = c_{\var{ret}} = ((\entry,\local),\_,\_,\_)$ where the range of authority is the same as $\reg(r_\stk)$ and it points to the first instruction of the activation code. 
  \item \pointstostack{\reg_1}{\emptyset}{\ms_\unused'} \label{pf:scall-works:item:points-to-empty-and-unusedp}
  \item $\reg_1(\pcreg) = \updatePcPerm{\reg(\pcreg)}$ 
  \item $\reg_1(r) = \reg(r)$
  \item $\reg_1(\overline{r_{\var{args}}}) = \reg(\overline{r_{\var{args}}})$
  \item $\forall r' \in \RegName \setminus \{\pcreg, r_\stk, r, \overline{r_{\var{args}}} \} \ldotp \reg_1(r') = 0$
  \end{enumproof}
  In order to use Lemma~\ref{lem:anti-red-obs}, we now need to show
  \[
    \npair[n_1]{(\reg_1,\ms \uplus \ms_f \uplus \ms_\stk \uplus \ms_{\var{act}} \uplus \ms_\unused')} \in \observations(W_1)
  \]
  where
  \[
    W_1 = \revokeTemp{W}[\iota^{\sta}(\temp,\ms_\stk\uplus\ms_{\mathit{act}} \uplus \ms_f),\iota^{\pwl}(\dom(\ms_{\mathit{unused}}'))]
  \]
  to this end use Hyp-Callee (\ref{pf:scall-works:item:Hyp-Callee}). To use this everything is satisfied directly by assumptions but the following:
  \begin{enumproof}[resume]
    \item $\memSat[n-1]{\ms \uplus \ms_f \uplus \ms_\stk \uplus \ms_{\var{act}} \uplus \ms_\unused'}{W_1}$\\
      Here we apply Lemma~\ref{lem:disj-mem-sat}. By assumption~\ref{pf:scall-works:mem-sat-ms} we have $\memSat[n]{\ms}{\revokeTemp{W}}$. So it suffices to show
      \[
        \memSat[n-1]{\ms_f \uplus \ms_\stk \uplus \ms_{\var{act}} \uplus \ms_\unused'}{[\iota^{\sta}(\temp,\ms_\stk\uplus\ms_{\mathit{act}} \uplus \ms_f),\iota^{\pwl}(\dom(\ms_{\mathit{unused}}'))]}
      \]
      This turns out to be trivial as $\ms_f$, $\ms_\stk$, and $\ms_{\var{act}}$ match the static region. $\ms_\unused'$ is all zeroes, to it trivially satisfies the $\iota^\pwl$ region.
    \item $\npair[n-1]{\reg'(r_\stk)} \in \stdvr(W_1)$ \\
      Use Lemma~\ref{lem:stack-cap-vr} with \ref{pf:scall-works:item:points-to-empty-and-unusedp} and that $W_1$ has region $\iota^{\pwl}(\dom(\ms_{\mathit{unused}}')$.
    \item $\npair[n-1]{c_{\var{ret}}} \in \stdvr(W_1)$ \\
      To this end let 
      \begin{enumproof}
        \item $n' < n-1$ \label{pf:scall-works:np-n1}
        \item $W_2 \futurewk W_1$ \label{pf:scall-works:w2-futurewk-w1}
      \end{enumproof}
      be given and show
      \[
        \npair[n']{\updatePcPerm{c_{\var{ret}}}} \in \stder(W_2)
      \]
      To this assume 
      \begin{enumproof}[resume]
        \item $n'' \leq n'$
        \item $\npair[n'']{\reg_2} \in \stdrr(W_2)$ \label{pf:scall-worsk:reg2-reg-rel}
        \item $\memSat[n'']{\ms'}{W_2}$ \label{pf:scall-works:w2-memsat}
      \end{enumproof}
      be given and show
      \begin{equation}
        \label{pf:scall-works:obsw2}
        \npair[n'']{(\reg_2\update{\pcreg}{\updatePcPerm{c_{\var{ret}}}},\ms')} \in \observations(W_2)
      \end{equation}

      From \ref{pf:scall-works:w2-futurewk-w1} and \ref{pf:scall-works:w2-memsat}, we can deduce that the memory can be split in the following way:
      \[
        \ms' = \ms'' \uplus \ms_r \uplus \ms_\stk \uplus \ms_\act \uplus \ms_\unused'' \uplus \ms_f
      \]
      where $\ms''$ is the "permanent" part of memory we get from Lemma~\ref{lem:priv-mono-like}, $\ms_r$ is the part "revoked" of memory from the same lemma that is not otherwise specified, and $\dom(\ms_\unused') = \dom(\ms_\unused'')$. From Lemma~\ref{lem:priv-mono-like} we also get
      \begin{enumproof}[resume]
        \item $\memSat[n'']{\ms''}{\revokeTemp{W_2}}$ \label{pf:scall-works:rt-w2-memsat}
      \end{enumproof}

      Assume $n''$ is large enough to execute the rest of the \texttt{scall} instructions. If $n''$ is not large enough, then \ref{pf:scall-works:obsw2} is trivial to show. To show \ref{pf:scall-works:obsw2} apply Lemma~\ref{lem:anti-red-obs} again where $\ms_r$ is the revoked part. Let $\ms_{\var{frame}}'$ be given, the execution until just after the \texttt{scall} proceeds as follows:
      \[
        (\reg_2\update{\pcreg}{\updatePcPerm{c_{\var{ret}}}},\ms' \uplus \ms_{var{frame}}') \step[j] (\reg_3,\ms' \uplus \ms_{\var{frame}}')
      \]
      where
      \begin{enumproof}[resume]
      \item \label{pf:scall-works:reg3-contents}
        \[
          \reg_3(r) =
          \begin{cases}
            c_{\var{next}} & r = \pcreg \\
            c_\stk & r = r_\stk \\
            \reg(r) & r \in \{\overline{r_{\var{priv}}} \}\\
            \reg_2(r) & \text{ otherwise}
          \end{cases}
        \]
        \item \pointstostack{\reg_3}{\ms_\stk}{\ms_\act \uplus \ms_\unused''} \label{pf:scall-works:reg3-stack}
      \end{enumproof}
      
      At this point, we use Hyp-Cont (\ref{pf:scall-works:item:Hyp-Cont}) to show the observation predicate condition of Lemma~\ref{lem:anti-red-obs}:
      \[
        \npair[n'']{(reg_3,\ms'' \uplus \ms_\stk \uplus \ms_\act \uplus \ms_\unused'' \uplus \ms_f)} \in \observations(W_2)
      \]
      which 
      \begin{itemize}
      \item $n'' \leq n-2$ \\
        Follows from (\ref{pf:scall-works:np-n1})
      \item $W_2 \futurewk \revokeTemp{W}$ \\
        We have 
        \[
          W_1 \futurewk \revokeTemp{W}
        \]
        and assumption \ref{pf:scall-works:w2-futurewk-w1} we get this by transitivity of $\futurewk$.
      \item $\memSat[n'']{\ms''}{\revokeTemp{W_2}}$ \\
        Exactly \ref{pf:scall-works:rt-w2-memsat}.        
      \item for all $r$, we have that:
        \begin{equation*}
          \reg_3(r)
          \begin{cases}
            = c_{\mathit{next}} &\text{ if } r = \pcreg\\
            = \reg(r)&\text{ if } r \in \overline{r_{\mathit{priv}}}\\
            \in \stdvr(\revokeTemp{W_2}) &\text{ if $\reg_3(r)$ is a global capability and } r \not\in \{\pcreg,\overline{r_{\mathit{priv}}}, r_\stk\} 
          \end{cases}
        \end{equation*}\\
        The two first cases follows from \ref{pf:scall-works:reg3-contents}. The third follow from assumption \ref{pf:scall-worsk:reg2-reg-rel} and \ref{lem:stdvr-glob-priv-mono}.
      \item $\reg' \text{ points to stack with $\ms_\stk$ used and $\ms_\act \uplus \ms_{\mathit{unused}}''$ unused}$\\
        Exactly \ref{pf:scall-works:reg3-stack}.
      \end{itemize}
  \end{enumproof}
\end{proof}

\subsubsection{Malloc macro-instruction correctness}
\begin{definition}
  We say that ``\linksto{(\reg,\ms)}{\var{key}}{j}{c_\malloc}'' 
  iff
  \begin{itemize}
  \item $\reg(pc) = \stdcap$
  \item $\ms(\start) = ((\_,\_),\start_\link,\_,\_)$
  \item $\ms(\start_\link+j) = c$
  \end{itemize}
\end{definition}

\begin{lemma}[$\mathtt{malloc}$ works]
  \label{lem:malloc-works}
  If
  \begin{itemize}
  \item \lookingat{(\reg,\ms)}{\mathtt{malloc}\;r\;k}{c_{\var{next}}}
  \item $k \geq 0$
  \item \linksto{(\reg,\ms)}{\malloc}{k}{c_\malloc}
  \item $c_\malloc$ satisfies the $\malloc$ specification with $\iota_{\malloc,0}$
  \item $W \futurestr [i \mapsto\iota_{\malloc,0}]$
  \item $\memSat{\ms}{W}$
  \item $\ms = \ms' \uplus \ms_{\var{footprint}}$
  \item $\memSat{\ms_{\var{footprint}}}{[i \mapsto W(i)]}$
  \item[Hyp-Cont]
    If
    \begin{itemize}
    \item $n' \leq n-1$
    \item $\iota_\malloc \futurewk W(i)$
    \item $\memSat[n']{\ms_{\var{footprint}}' \uplus \ms'}{W[i \mapsto \iota_\malloc]}$
    \item $\memSat[n']{\ms_{\var{footprint}}'}{[i \mapsto \iota_\malloc]}$
      \begin{equation*}
        \reg'(r') = 
        \begin{cases}
          c_{\var{next}} & r' = \pcreg \\
          ((\rwx,\glob),\start,\addrend,\addr) & r' = r \\
          \reg(r) & r' \not\in \RegName_t \union \{\pcreg, r, r_1\}
        \end{cases}
      \end{equation*}
    \item $\addrend - \start = k - 1$
    \item $\dom(\ms_{\var{alloc}}) = [\start,\addrend]$
    \item $\forall \addr \in [\start,\addrend]\ldotp \ms_{\var{alloc}}(\addr) = 0$
    \end{itemize}
    Then we have $\npair[n']{(\reg',\ms' \uplus \ms_{\var{footprint}}'  \uplus \ms_{\var{alloc}})} \in \observations(W[\iota_\malloc])$
  \end{itemize}
  Then
  \[
    \npair{(\reg,\ms)} \in \observations(W)
  \]
\end{lemma}

\subsubsection{Create closure macro-instruction correctness}
\begin{lemma}[$\mathtt{crtcls}$ works]
  \label{lem:crtcls-works}
  If
  \begin{itemize}
  \item \lookingat{(\reg,\ms)}{\mathtt{crtcls}\;\overline{(x,r)}\;r}{c_{\var{next}}}
  \item \linksto{(\reg,\ms)}{\malloc}{k}{c_\malloc}
  \item $c_\malloc$ satisfies the $\malloc$ specification with $\iota_{\malloc,0}$
  \item $W \futurestr [i \mapsto\iota_{\malloc,0}]$
  \item $\memSat{\ms}{W}$
  \item $\ms = \ms' \uplus \ms_{\var{footprint}}$
  \item $\memSat{\ms_{\var{footprint}}}{[i \mapsto W(i)]}$
  \item[Hyp-Cont]
    If
    \begin{itemize}
    \item $n' \leq n$
    \item $\iota_\malloc \futurewk W(i)$
    \item $\memSat[n']{\ms' \uplus \ms_{\var{footprint}}'}{W[i \mapsto \iota_\malloc]}$
    \item $\memSat{\ms_{\var{footprint}}'}{[i \mapsto \iota_\malloc]}$
    \item
      \begin{equation*}
        \reg'(r') =
        \begin{cases}
          c_{\var{next}} & r' = \pcreg \\
          c_\cls = ((\entry,\glob),\start,\addrend,\start+2) & r' = r_1 \\
          \reg(r)      & r' \not\in \{\pcreg, r_1\} \union \RegName_t
        \end{cases}
      \end{equation*}
    \item $\ms_\cls = \ms_\act \uplus \ms_\env$
    \item $c_\cls = ((\entry,\glob),\dots)$
    \item $c_\env = ((\readwrite,\glob),\start_\env,\addrend_\env,\start_\env)$
    \item $\dom(\ms_\env) = [\start_\env,\addrend_\env]$
    \item $\ms_\env(\start_\env,\dots,\addrend_\env) = \reg(\overline{r})$
    \item Hyp-act\\ If
      \begin{itemize}
      \item $\reg''(\pcreg) = \updatePcPerm{c_\cls}$
      \end{itemize}
      Then  $\exists k \ldotp \forall \ms_f \ldotp (\reg'',\ms'' \uplus \ms_\cls \uplus \ms_f) \step[k] (\reg''',\ms'' \uplus \ms_\cls \uplus \ms_f)$
      where
        \begin{equation*}
          \reg'''(r') =
          \begin{cases}
            c_\env & r' = c_\env \\
            \updatePcPerm{\reg(r)} & r' = \pcreg \\
            \reg''(r') & r' \not\in \RegName_t
          \end{cases}
        \end{equation*}
    \end{itemize}
    Then we have $\npair[n']{(\reg',\ms' \uplus \ms_{\var{footprint}} \uplus \ms_\cls)} \in \observations(W[i \mapsto \iota_\malloc])$
  \end{itemize}
  Then 
  \[
    \npair{(\reg,\ms)} \in \observations(W)
  \]
\end{lemma}

\subsubsection{Stack helper lemmas}
\begin{lemma}
  \label{lem:stack-exec-cond-help}
  If
  \begin{itemize}
  \item $\perm \in \{\exec,\rwx,\rwlx \}$
  \item $\npair{(\start,\addrend)} \readCond{}(\local)(W)$
  \item $\npair{(\start,\addrend)} \writeCond{}(\iota^\pwl,\local)(W)$
  \end{itemize}
  then
  \[
    \npair[n]{\perm,\start,\addrend} \in \execCond{}(\local)(W)
  \]
\end{lemma}
\begin{proof}[Proof of Lemma~\ref{lem:stack-exec-cond-help}]
Assume
\begin{enumproof}
\item $\perm \in \{\exec,\rwx,\rwlx \}$
\item $\npair{(\start,\addrend)} \readCond{}(\local)(W)$ \label{pf:stack-exec:ass:readcond}
\item $\npair{(\start,\addrend)} \writeCond{}(\iota^\pwl,\local)(W)$ \label{pf:stack-exec:ass:writecond}
\end{enumproof}
Let $W' \futurewk W$, $\addr$, and $n' \leq n$ be given and show
\[
  \npair[n']{((\perm,\local),\start,\addrend,\addr)} \in \stder(W')
\]
Consider each of the three cases for $\perm$:
\begin{enumproof}[resume]
  \item $\perm = \rwlx$\\
    In this case $\iota= \iota^\pwl$. If we use the FTLR (Theorem~\ref{thm:ftlr}), then we are done. It suffices to show:
    \begin{enumproof}
      \item $\npair[n']{(\start,\addrend)} \in \readCond{}(\local)(W')$ \\
        Follows from Lemma~\ref{lem:readcond-mono-pub}, Lemma~\ref{lem:readcond-dc}, and assumption~\ref{pf:stack-exec:ass:readcond}.
      \item $\npair[n']{(\start,\addrend)} \in \writeCond{}(\iota^\pwl,\local)(W')$\\
        Follows from Lemma~\ref{lem:writecond-mono-pub}, Lemma~\ref{lem:readcond-dc}, and assumption~\ref{pf:stack-exec:ass:writecond}.
    \end{enumproof}
  \item $\perm = \exec$\\
    In this case $\iota= \iota^\nwl$. If we use the FTLR (Theorem~\ref{thm:ftlr}), then we are done. It suffices to show:
    \begin{enumproof}
      \item $\npair[n']{(\start,\addrend)} \in \readCond{}(\local)(W')$ \\
        Follows from Lemma~\ref{lem:readcond-mono-pub}, Lemma~\ref{lem:readcond-dc}, and assumption~\ref{pf:stack-exec:ass:readcond}.
      \item $\npair[n']{(\start,\addrend)} \in \writeCond{}(\iota^\nwl,\local)(W')$\\
        Follows from Lemma~\ref{lem:wc-pwl-implies-wc-nwl}, Lemma~\ref{lem:writecond-mono-pub}, Lemma~\ref{lem:readcond-dc}, and assumption~\ref{pf:stack-exec:ass:writecond}.
    \end{enumproof}
  \item $\perm = \rwx$\\
    In this case $\iota= \iota^\nwl$. If we use the FTLR (Theorem~\ref{thm:ftlr}), then we are done. It suffices to show:
    \begin{enumproof}
      \item $\npair[n']{(\start,\addrend)} \in \readCond{}(\local)(W')$ \\
        Follows from Lemma~\ref{lem:readcond-mono-pub}, Lemma~\ref{lem:readcond-dc}, and assumption~\ref{pf:stack-exec:ass:readcond}.
      \end{enumproof}
\end{enumproof}
\end{proof}

\begin{lemma}[Stack capability in value relation]
  \label{lem:stack-cap-vr}
  If
  \begin{itemize}
  \item \pointstostack{\reg}{\emptyset}{\ms}
  \item $\exists r \ldotp W(r) = \iota^\pwl (\dom(\ms))$
  \end{itemize}
  then
  \[
    \npair{\reg(r_\stk)} \in \stdvr(W)
  \]
\end{lemma}
\begin{proof}[Proof of Lemma~\ref{lem:stack-cap-vr}]
  Say 
\[
  \reg(r_\stk) = c_\stk = ((\rwlx,\local),\start,\addrend,\_)
\]
  Show
  \begin{enumproof}
    \item $\npair{(\start,\addrend)} \in \readCond{}(\local)(W)$ : \label{pf:stack-cap-read}\\
      Amounts to
      \[
        \iota^\pwl (\dom(\ms)) \nsubsim \iota^\pwl_{\start,\addrend}
      \]
      which is true as they are even equal.
    \item $\npair{(\start,\addrend)} \in \writeCond{}(\iota^\pwl,\local)(W)$ : \label{pf:stack-cap-write}\\
      Using Lemma~\ref{lem:iota-pwl-address-stratified}, this amounts to
      \[
        \iota^\pwl (\dom(\ms)) \nsupsim \iota^\pwl_{\start,\addrend}
      \]
      which is true as they are even equal.
      \item $\npair{(\rwlx,\start_\stk, \addrend_\stk)} \in \execCond{}(\local)(W)$\\
        Using \ref{pf:stack-cap-write} and \ref{pf:stack-cap-read}, we can use Lemma~\ref{lem:stack-exec-cond-help}.
      \item $\npair{(\rwx,\start_\stk, \addrend_\stk)} \in \execCond{}(\local)(W)$\\
        Using \ref{pf:stack-cap-write} and \ref{pf:stack-cap-read}, we can use Lemma~\ref{lem:stack-exec-cond-help}.
      \item $\npair{(\exec,\start_\stk, \addrend_\stk)} \in \execCond{}(\local)(W)$\\
        Using \ref{pf:stack-cap-read} and \ref{pf:stack-cap-write}, we can use Lemma~\ref{lem:stack-exec-cond-help}.
  \end{enumproof}
\end{proof}

\subsubsection{Memory Segment Satisfaction}
We expect the following lemmas to hold true:


\begin{lemma}[Revoke temporary memory satisfaction]
  \label{lem:priv-mono-like}
  \begin{align*}
    &      \forall \hs, n, W, W' \ldotp \\
    & \quad  \heapSat[\hs]{n}{W} \Rightarrow\\
    & \qquad \exists \hs', \hs_r \ldotp \\
    & \qquad \quad \hs = \hs' \uplus \hs_r \land \heapSat[\hs']{n}{\wrev{W}}
  \end{align*}
\end{lemma}
\begin{proof}[Proof of Lemma~\ref{lem:priv-mono-like}]
\end{proof}

\begin{lemma}[Revoke temporary memory satisfaction 2]
  \label{lem:priv-mono-like2}
  \begin{align*}
    & \forall \ms, n, W, R : \activeReg{W} \fun \MemSegments  \ldotp \\
    & \quad \memSatPar{\ms}{W}{P} \Rightarrow \\
    & \qquad \exists \ms', \ms_r\ldotp \\
    & \qquad \quad \ms = \ms' \uplus \ms_r \land \\
    & \qquad \quad \memSatPar{\ms'}{\wrev{W}}{P|_{\dom(\erase{W}{\perma})}} \land \\
    & \qquad \quad \ms_r = \biguplus_{r \in \erase{W}{\temp}} P(r) \land \\
    & \qquad \quad \ms' = \biguplus_{r \in \erase{W}{\perma}} P(r)
  \end{align*}
\end{lemma}
\begin{proof}[Proof of Lemma~\ref{lem:priv-mono-like2}]
\end{proof}

\begin{lemma}[Revoke temporary memory with stack]
  \label{lem:revoke-temp-stack}
  \begin{align*}
    & \forall n, \ms, W, \reg, r_\stk,\gl,\start,\addrend,\addr \ldotp\\
    & \quad \memSat{\ms}{W} \land \npair{\reg} \in \stdrr(W) \land \\
    & \quad \reg(r_\stk) = ((\rwlx,\gl),\start,\addrend,\addr) \land b \leq e \\
    & \qquad \exists \ms', \ms_r \ldotp \\
    & \qquad \quad \memSat{\ms'}{\revokeTemp{W}} \land \ms = \ms' \uplus \ms_r
  \end{align*}
\end{lemma}
\begin{proof}[Proof of Lemma~\ref{lem:revoke-temp-stack}]
\end{proof}

\begin{lemma}[Disjoint memory satisfaction]
  \label{lem:disj-mem-sat}
  \begin{align*}
    & \forall n \ldotp \forall \ms, \ms', \ms'' \ldotp \forall W, W', W'' \ldotp\\
    & \quad \ms'' = \ms \uplus \ms' \land W'' = W \uplus W' \land \memSat{\ms}{W} \land \memSat{\ms'}{W'} \Rightarrow \\
    & \qquad \heapSat[\ms'']{n}{W''}
  \end{align*}
\end{lemma}
\begin{proof}[Proof of Lemma~\ref{lem:disj-mem-sat}]
\end{proof}

\begin{lemma}[Memory satisfaction and static regions]
  \label{lem:mem-sat-static}
  \[
    \memSat{\ms}{[i \mapsto \iota^\sta(v,\ms)]}
  \]
\end{lemma}
\begin{proof}[Proof of Lemma~\ref{lem:mem-sat-static}]
\end{proof}

\begin{lemma}[Data only memory and standard regions]
  \label{lem:mem-sat-data-only-std-regions}
  If
  \begin{itemize}
  \item $\forall \addr \in \dom(\ms)\ldotp \ms(a) \in \nats$
  \item $\iota \in \{\iota^\pwl,\iota^\nwl,\iota^{\nwl,p} \}$
  \end{itemize}
  then
  \[
    \memSat{\ms}{[i \mapsto \iota(\dom(\ms))]}
  \]
\end{lemma}
\begin{proof}[Proof of Lemma~\ref{lem:mem-sat-data-only-std-regions}]
\end{proof}

\subsubsection{Future worlds}
\begin{lemma}[World public future world of revoked world]
  \label{lem:rt-w-pub-future-w}
  \[
    \forall W \ldotp \revokeTemp{W} \futurewk W
  \]
\end{lemma}
\begin{proof}[Proof of Lemma~\ref{lem:rt-w-pub-future-w}]
  For all $r$ where $W(r) = (\temp,s,\phi_\pub,\phi,H)$, we have $\revokeTemp{W} = \revoked$. By the public future region relation we have
  \[
    W(r) = (\temp,s,\phi_\pub,\phi,H) \futurewk \revokeTemp{W}(r) = \revoked
  \]
  all other regions remain unchanged, so this follows by reflexivity of the public future region relation.
\end{proof}

\begin{lemma}[World private future world of revoked world]
  \label{lem:rt-w-priv-future-w}
  \[
    \forall W \ldotp \wrev{W} \futurestr W
  \]
\end{lemma}
\begin{proof}[Proof of Lemma~\ref{lem:rt-w-priv-future-w}]
\end{proof}

\begin{lemma}[Public future world relation included in private future world relation]
  \label{lem:future-pub-impl-future-priv}
  \[
   W' \futurewk W \Rightarrow W' \futurestr W
  \]
\end{lemma}
\begin{proof}[Proof of Lemma~\ref{lem:future-pub-impl-future-priv}]
\end{proof}

\begin{lemma}[Transitivity proberties between private and public future worlds]
  \label{lem:future-priv-pub-trans}
\[
  W'' \futurestr W' \land W' \futurewk W \Rightarrow W'' \futurestr W
\]
and
\[
  W'' \futurewk W' \land W' \futurestr W \Rightarrow W'' \futurestr W
\]
\end{lemma}
\begin{proof}[Proof of Lemma~\ref{lem:future-priv-pub-trans}]
\end{proof}

\begin{lemma}
  \label{lem:wk-nequal-and-future}
  \begin{align*}
    &\forall n,W_1,W_2,W_1' \ldotp
    &\quad W_1 \nequal W_2 \land W_1' \futurewk W_1 \Rightarrow \exists W_2' \ldotp W_2' \nequal W_1' \land W_2' \futurewk W_2
  \end{align*}
\end{lemma}
\begin{proof}[Proof of Lemma~\ref{lem:wk-nequal-and-future}]
  Construct $W_2'$ as follows:
  \[
    W_2(r) =
    \begin{cases}
      (v_1',s_1',\phi_{\pub2},\phi_2,H_2) & \arraycolsep=0pt
      \begin{array}[t]{l}
        \text{ if $r \in \dom(W_2)$ and $W_1'(r) = (v_1',s_1',\_,\_,\_)$}\\
        \quad \text{ and $W_2(r) = (\_,\_,\phi_{\pub2},\phi_2,H_2)$}\\
      \end{array} \\
      W_1'(r) & \text{ otherwise}
    \end{cases}
  \]
  Notice $\dom(W_2') = \dom(W_1')$.
\end{proof}

\begin{lemma}
  \label{lem:str-nequal-and-future}
  \begin{align*}
    &\forall n,W_1,W_2,W_1' \ldotp
    &\quad W_1 \nequal W_2 \land W_1' \futurestr W_1 \Rightarrow \exists W_2' \ldotp W_2' \nequal W_1' \land W' \futurestr W_2
  \end{align*}
\end{lemma}
\begin{proof}[Proof of Lemma~\ref{lem:str-nequal-and-future}]
  Construct $W_2'$ as follows:
  \[
    W_2(r) =
    \begin{cases}
      (v_1',s_1',\phi_{\pub2},\phi_2,H_2) & \arraycolsep=0pt
      \begin{array}[t]{l}
        \text{ if $r \in \dom(W_2)$ and $W_1'(r) = (v_1',s_1',\_,\_,\_)$}\\
        \quad \text{ and $W_2(r) = (\_,\_,\phi_{\pub2},\phi_2,H_2)$}\\
      \end{array} \\
      W_1'(r) & \text{ otherwise}
    \end{cases}
  \]

\end{proof}

\subsubsection{Value relation}
\begin{lemma}[Value relation downwards closed]
  \label{lem:stdvr-dc}
  \[
    n' \leq n \land \npair{w} \in \stdvr(W) \Rightarrow \npair[n']{w} \in \stdvr(W)
  \]
\end{lemma}
\begin{proof}
  By definition of $\stdvr(W)$ using Lemma~\ref{lem:readcond-dc},
  \ref{lem:writecond-dc}, \ref{lem:execcond-dc} and \ref{lem:entrycond-dc}.
\end{proof}

\begin{lemma}[Register relation downwards closed]
  \label{lem:stdrr-dc}
  \[
    n' \leq n \land \npair{w} \in \stdrr(W) \Rightarrow \npair[n']{w} \in \stdrr(W)
  \]
\end{lemma}
\begin{proof}
  By definition of $\stdrr(W)$ using Lemma~\ref{lem:stdvr-dc}.
\end{proof}

\begin{lemma}[Value relation monotone wrt $\futurewk$]
  \label{lem:stdvr-mono-wk}
  \[
    W' \futurewk W \land \npair{w} \in \stdvr(W) \Rightarrow \npair{w} \in \stdvr(W')
  \]
\end{lemma}
\begin{proof}[Proof of lemma~\ref{lem:stdvr-mono-wk}]
  Follows from Lemma~\ref{lem:readcond-mono-pub}, Lemma~\ref{lem:writecond-mono-pub}, Lemma~\ref{lem:execcond-mono-pub}, and Lemma~\ref{lem:entrycond-mono-pub}.
\end{proof}

\begin{lemma}
  \label{lem:revoketemp-vr-implies-vr}
  If
  \[
    \npair{w} \in \stdvr(\revokeTemp{W})
  \]
  then
  \[
    \npair{w} \in \stdvr(W)
  \]
\end{lemma}
\begin{proof}[Proof of Lemma~\ref{lem:revoketemp-vr-implies-vr}]
  Follows from Lemma~\ref{lem:revoketemp-readcond}, Lemma~\ref{lem:revoketemp-writecond}, Lemma~\ref{lem:revoketemp-execcond}, and Lemma~\ref{lem:revoketemp-entercond}.
\end{proof}

\begin{lemma}[Global capabilities monotone wrt $\futurestr$]
  \label{lem:stdvr-glob-priv-mono}
  \begin{align*}
    & \forall n, \perm, \start, \addrend, \addr, W, W' \ldotp \\
    & \quad  \npair{\stdcap[(\perm,\glob)]} \in \stdvr(W) \land W' \futurestr W\\
    & \qquad \Rightarrow \npair{\stdcap[(\perm,\glob)]} \in \stdvr(W') 
  \end{align*}
\end{lemma}
\begin{proof}[Proof of Lemma~\ref{lem:stdvr-glob-priv-mono}]
  Assume 
  \begin{enumproof}
  \item $\perm \not\in \{\rwl, \rwlx\}$
  \item $W' \futurestr W$
  \item $\npair{\stdcap[(\perm,\glob)]} \in \stdvr(W)$
  \end{enumproof}
  and show
  \[
    \npair{\stdcap[(\perm,\glob)]} \in \stdvr(W')
  \]
  to this end consider the possible cases of $\perm$ and show that each of the necessary conditions hold:
  \begin{enumproof}
    \item $\perm = \noperm$\\
      Trivial
    \item $\perm = \readonly$ \\
      Follows from Lemma~\ref{lem:readcond-mono-priv}.
    \item $\perm = \readwrite$ \\
      Follows from Lemma~\ref{lem:readcond-mono-priv} and Lemma~\ref{lem:writecond-mono-priv}.
    \item $\perm = \exec$ \\
      Follows from Lemma~\ref{lem:readcond-mono-priv} and Lemma~\ref{lem:execcond-mono-priv}.
    \item $\perm = \rwx$ \\
      Follows from Lemma~\ref{lem:readcond-mono-priv}, Lemma~\ref{lem:writecond-mono-priv}, and Lemma~\ref{lem:execcond-mono-priv}.
    \item $\perm = \entry$ \\
      Lemma~\ref{lem:entrycond-mono-priv}
  \end{enumproof}
\end{proof}

\begin{lemma}[Non local words monotone wrt $\futurestr$]
  \label{lem:stdvr-non-loc-priv-mono}
  \begin{align*}
    & \forall n, \perm, \start, \addrend, \addr, W, W',w \ldotp \\
    & \quad \nonlocal{w} \land \\
    & \quad \npair{w} \in \stdvr(W) \land W' \futurestr W \\
    & \qquad \Rightarrow \npair{w} \in \stdvr(W') 
  \end{align*}
\end{lemma}
\begin{proof}[Proof of Lemma~\ref{lem:stdvr-non-loc-priv-mono}]
  If $w = \stdcap[(\perm,\glob)]$, then let follows from Lemma~\ref{lem:stdvr-glob-priv-mono}.

  If $w \in \ints$, then it follows from the fact that $i \in \stdvr(W'')$ for all $i \in \ints$ and $W'' \in \Worlds$.
\end{proof}

\section{Other examples and applications}
\label{sec:other_apps}
This section contains some ideas about other examples and applications than the
ticket dispenser example.

\subsection{Stack and return pointer handling without OS involvement using local
  capabilities}
The idea of this example would be to work out and prove a calling convention
that enforces well-bracketed control flow and encapsulation of local variables
using CHERI's local capabilities.

When one function invokes another function, the essential idea is that:
\begin{itemize}
\item Stack pointer is passed as a local and store-local capability.
\item Return pointer is passed as a local capability.
\end{itemize}

Since local pointers cannot leave the registers except into regions for which a
store-local capability is available, this basic idea seems to enforce a number
of useful properties: well-bracketedness of control flow and encapsulation of
private state stored on the stack. On the other hand, it also seems to validate
the standard C treatment of the stack: the stack can be reused after a function
returns, even between distrusting parties. However, safety/security of this
design is very non-trivial and seems to rely on some non-trivial reasoning:

\paragraph{Only stack is store-local?}
A critical assumption is that adversary code has no way to \emph{store} local
capabilities except on the stack. The reason that it is fine to store local
capabilities on the stack is that the adversary only has a \emph{local}
capability to the stack and cannot usefully store that capability anywhere.
However, this means that we need to rely on the runtime system of our
programming language to be careful when handing out store-local capabilities:
only the libc startup code should initialise the stack as store-local and malloc
should \emph{not} produce them. This basically means that the libc
initialisation code (or whatever component produces the initial stack pointer)
is part of our TCB.

\paragraph{Requirement for clearing the stack}
Imagine the following trusted C function:

\begin{verbatim}
void myfunction(){
  advfunction1();
  advfunction2();
}
\end{verbatim}

where advfunction1() and advfunction2() are adversary functions. In the standard
C treatment of the stack, advfunction2() would get the same stack pointer as
advfunction1(). This is supposed to be safe since advfunction1() cannot have
kept capabilities for the stack after its execution. But what if we require that
the two functions have no way of communicating with each other? Concretely,
advfunction1() has access to some secrets that must not be leaked to
advfunction2(). How can we prevent advfunction1() from storing the secret
somewhere on the stack and relying on advfunction2() from receiving the same
stack pointer where it can read the secret? The most obvious solution seems to
be that we should fully clear the stack (overwrite it with zeros) after the
return of any adversary function, but this could cause an important overhead.
Perhaps the processor should accommodate this with a special instruction that can
zero the entire array that a capability points to?

\paragraph{What do return pointers look like?}
An important question is what return pointers look like? Since we want to
protect the caller from the callee, it's important that the return pointer is
opaque, i.e. an entry pointer. The entry pointer will point to a closure that
contains the next instruction to execute, as well as the previous stack pointer.
But since stack pointers are local, this means that the return pointer closure
should be stored in a region of memory for which we have store-local permission,
i.e. on the stack. This means we need the following in our calling convention:
before invoking a function, we push the stack pointer and the instruction
pointer after invocation on the stack, we construct a return pointer by
copying the stack pointer, limiting it to these two entries and making it an
entry pointer.  Then we shrink the stack pointer to the unused part of the stack
and jump. 

\paragraph{Only one-way protection in higher-order settings?}
Another important point is that, in a sense, local capabilities provide only
one-way protection: the caller is protected from the callee but not vice-versa.
Concretely: when invoking a function with some arguments marked as local, the
caller is guaranteed that the callee will not have been able to store the
capabilities anywhere (except perhaps on the stack, see above). However, the
callee seems to have more limited guarantees: Particularly, the caller may have
kept its own stack capability and this stack capability may (and typically will)
also cover the part of the stack that is ``owned'' by the callee.  In this
sense, the guarantees are more limited than in a linear language.

So what does this mean? In a first-order language, this is all fine, but what if
we are in a higher-order language. Imagine the following (in some ML-like language):

\begin{verbatim}
let f = fun callback =>
          let ... in 
          let ret = callback() in
          ...
//adversary top function
let advtop = f( (fun y => ...) )
\end{verbatim}

Our trusted function f is invoked by the adversary (from function advtop()) and
wants to invoke an untrusted callback received from the adversary. When invoking
the closure, we don't want it to be able to access f's local variables which it
has stored on the stack. To achieve this, we only give it a stack pointer that
covers the part of the stack that is unused by f. However, the callback may be
implemented as an entry pointer that carries capabilities, particularly the
capability to advtop's stack pointer, which includes the part of the stack that
is now used by f and contains f's local variables.

So how do we deal with this? Perhaps we should use the fact that this is only
possible when f's callback argument is allocated to some part of the memory to
which advtop has store-local permissions (since the callback contains a
reference to the stack to which advtop only has a local capability). I see
basically three ways to do this, all based on the idea of enforcing that the
callback should be constructed in a part of memory for which no store-local
permissions are available:
\begin{itemize}
\item One way to exclude the scenario is to require that callbacks are provided
  as non-local capabilities. The downside of this is that local callbacks can be
  useful for the caller to prevent the callee from storing them.
\item Another way to exclude the scenario is to require that the stack is
  allocated in a fixed part of the address space and to check that callbacks
  point outside of this region before invoking them.
\item Perhaps we should require that store-local permissions cannot be removed
  from a capability and simply require that callback pointers do not have
  store-local set. Perhaps we can allow store-local permissions to be given up,
  but only if the corresponding part of memory is fully zeroed in the process
  (or at least all local capabilities stored in the region).
\end{itemize}

\subsection{A result to prove...}
\label{sec:os-less-stack-property}

The simplest thing that comes to mind as a formal result for all of the above is
to look at a concrete program that clearly relies on properties like
well-bracketed control flow and encapsulation of local variables and prove it
correct. As a concrete example: we might show an assembly program that
corresponds to the following (a higher-order program that crosses trust
boundaries and relies on local variable encapsulation and well-bracketed control
flow):
\begin{verbatim}
let trustedCode = fun adversary =>
                    let x = ref 0 in
                    let callback = fun adv2 => 
                      x := !x + 1; 
                      let y = ref (!x) in 
                      adv2 unit; 
                      assert (!x == !y); 
                      x := !x - 1)
                    let _ = adversary callback
                    assert (!x == 0)
\end{verbatim}
\lau{I have inserted some line breaks for readability. not sure what is going on here (the parenthesis after y is unmatched.)}
\section{Related reading}
\label{sec:related-reading}

This is a list of related work that might be interesting to read in the context
of this project.

\subsection{Capability machines}
\label{sec:rw-cap-machines}

\subsubsection{M-Machine}
More than 20 years ago, \cite{Carter:1994:HSF:195473.195579} have described the
use of capabilities in the M-Machine. They do seem to have a reference for the
instruction set after all~\citep{Dally1997Memo59}; it seems like the server was
just temporarily down when we were looking for this the first time...

\subsubsection{CHERI}

The CHERI processor is a much more recent capability machine, described
by~\cite{Woodruff:2014:CCM:2665671.2665740,Watson2015Cheri}.

Another result of this project is also CheriBSD: an adaptation of FreeBSD to the
CHERI
processor.\footnote{\url{http://www.cl.cam.ac.uk/research/security/ctsrd/cheri/cheribsd.html}}
It is not separately described in a published paper, but mentioned in the papers
cited above and in some tech reports (see url). This work includes a
pure-capability ABI that could provide some interesting examples.

The CHERI team also has a webpage with all of their CHERI-related publications
(including TRs and
such)\footnote{\url{http://www.cl.cam.ac.uk/research/security/ctsrd/cheri/}}.

\subsection{Logical Relations}
\label{sec:rw-log-rel}

Some papers on logical relations that are relevant for this work are the
following:

\cite{Hur:2011:KLR:1926385.1926402} describe a logical relation between ML and
a (standard) assembly language for expressing compiler correctness.  Relevant
because they target an assembly language, and they use biorthogonality.

\cite{Dreyer:2010:IHS:1863543.1863566} describe a logical relation for a ML-like
language and use public/private transitions to reason about well-bracketed
control flow. Relevant because we are considering to cover an example of
enforcing well-bracketed control flow in a capability machine.

\cite{Devriese:2016ObjCap} describe a logical relation for a JavaScript-like
language with object capabilities.  Relevant because it treats object
capabilities, albeit in a JavaScript-like lambda calculus. It also deals with an
untyped language, using a semantic unitype.

\bibliographystyle{plainnat}
\bibliography{refs}

\end{document}